\documentclass[3p,10pt,a4paper,twoside,fleqn]{elsarticle}
\usepackage{amssymb,amsmath,latexsym}
\usepackage{pxfonts}
\newtheorem{theorem}{Theorem}[section]
\newtheorem{lemma}{Lemma}[section]

\newtheorem{remark}{Remark}[section]

\newtheorem{example}{Example}[section]

\font\tenscr=rsfs10  scaled \magstep0 \font\sevenscr=rsfs7 scaled \magstep0
\font\fivescr=rsfs5 scaled \magstep0 \skewchar\tenscr='177 \skewchar\sevenscr='177
\skewchar\fivescr='177
\newfam\scrfam \textfont\scrfam=\tenscr \scriptfont\scrfam=\sevenscr
\scriptscriptfont\scrfam=\fivescr
\def\mathscr{\fam\scrfam}
\renewcommand{\cal}{\mathscr}

\def\endproof{\qed\endtrivlist}
\expandafter\let\csname endproof*\endcsname=\endproof

\def\qedsymbol{\ifmmode\bgroup\else$\bgroup\aftergroup$\fi
  \vcenter{\hrule\hbox{\vrule height.6em\kern.6em\vrule}\hrule}\egroup}
\def\qed{\ifmmode\else\unskip\nobreak\fi\quad\qedsymbol}

\renewcommand{\iff}{\Leftrightarrow}
\renewcommand{\implies}{\Rightarrow}
\newcommand{\lra}{\leftrightarrow}
\renewcommand{\le}{\leqslant}
\renewcommand{\ge}{\geqslant}

\newcommand{\natp}[1]{\pi_{{\!}_{\scriptstyle{#1}}}}

\usepackage{color}

\begin{document}

\journal{\qquad}

\title{\bf\large Reduction of fuzzy automata by means of fuzzy quasi-orders\tnoteref{t1}}
\tnotetext[t1]{Research supported by Ministry  of Science and Technological Development, Republic
of Serbia, Grant No. 174013}
\author[fsmun]{Aleksandar Stamenkovi\' c}
\ead{aca@pmf.ni.ac.rs}

\author[fsmun]{Miroslav \'Ciri\'c\corref{cor}}
\ead{mciric@pmf.ni.ac.rs}

\author[fsmun]{Jelena Ignjatovi\'c}
\ead{jekaignjatovic73@gmail.com}

\cortext[cor]{Corresponding author. Tel.: +38118224492; fax: +38118533014.}

\address[fsmun]{Faculty of Sciences and Mathematics, University of Ni\v s, Vi\v segradska 33, P. O. Box 224, 18000 Ni\v s, Serbia}

\begin{abstract}
In our recent paper we have established close relationships
between state reduction of a fuzzy recognizer and~resolution of a
particular system of fuzzy relation equations. In that paper we
have also studied reductions by means of those solutions which are
fuzzy equivalences.~In this paper we will see that in some cases
better reductions~can~be obtained using the solutions of this
system that are fuzzy quasi-orders.~Generally, fuzzy quasi-orders
and fuzzy~equivalences are equally good in the state reduction,
but we show that right and left invariant fuzzy quasi-orders give
better reductions than right and left invariant fuzzy
equivalences. We also show that alternate reductions by means of
fuzzy quasi-orders give better results than alternate reductions
by means of fuzzy~equivalences.~Furthermore we study a more
general type of fuzzy quasi-orders, weakly right and left
invariant ones, and we~show that they are closely related to
determinization of fuzzy recognizers.~We also demonstrate some applications of weakly
left invariant fuzzy quasi-orders in conflict analysis of fuzzy discrete
event systems. \end{abstract}


\maketitle

\section{Introduction}

Unlike deterministic finite automata (DFA), whose efficient
minimization is possible, the state minimi\-zation problem for
non-deterministic finite automata (NFA)~is com\-putationally~hard
({\small PSPACE}-complete, \cite{JR.93,Yu.97}) and  known
algorithms like in \cite{CC.03a,KW.70,Mel.99,Mel.00,Seng.92}
cannot be used in practice.~For that reason, many researchers
aimed their atten\-tion~to NFA state reduction methods which do
not~neces\-sarily give a minimal one, but they give "reasonably"
small NFAs that can be constructed efficiently.~The
basic~idea~of~redu\-cing the number of states of NFAs by computing
and merging~indis\-tinguishable states~resembles
the~minimi\-zation algorithm for DFAs, but is
more~com\-plica\-ted.~That led~to the concept of a right
invariant~equi\-va\-lence on an NFA, studied by Ilie and Yu
\cite{IY.02,IY.03},  who showed that they can be used to
construct~small NFAs from regular expressions.~In particular, both
the partial derivative automaton and the follow~automa\-ton of a
given regular expression are factor automata of the position
automaton  with respect to the right invariant equivalences
(cf.~\cite{CZ.01a,CZ.01b,IY.02a,IY.03,IY.03a}).~Right invariant
equivalences have been also
studied~in~\cite{CCK.00,CSY.05,CC.04,IY.03,INY.04,ISY.05}.~Moreover,
the same concept was studied under the name "bisimulation
equivalence" in~many areas of computer science and mathematics,
such as modal~logic, concurrency theory, set theory, formal
verification, model checking, etc., and numerous algorithms have
been proposed~to~compute the greatest bisimulation equivalence on
a given labeled graph or a labeled transition system (cf.
\cite{LV.95,Milner.80,Milner.89,Milner.99,Park.81,RM-C.00,Sangiorgi.07}).~The
faster algorithms are based on the crucial equivalence between the
greatest bisimulation equivalence and the relational coarsest
partition problem (see \cite{DPP.04,GPP.03,KS.90,RT.08,PT.87}).

Better results in the state reduction of NFAs can be achieved in
two ways.~The first one was also proposed~by Ilie and Yu in
\cite{IY.02,IY.03,INY.04,ISY.05} who introduced the dual concept
of a left invariant equivalence on an NFA and showed that even
smaller NFAs can be obtained alternating reductions by means of
right invariant and left invariant equivalences. On the other
hand, Champarnaud and Coulon in \cite{CC.03,CC.04} proposed use of
quasi-orders (preorders) instead of equivalences and showed that
the method based on quasi-orders gives better reductions than the
method based on equivalences.~They gave an algorithm for computing
the greatest right invariant and~left invariant quasi-orders on an
NFA working in a polynomial time, which was later improved in
\cite{INY.04,ISY.05}.~

Fuzzy finite automata are generalizations of NFAs, and the above
mentioned problems concerning minimiza\-tion and reduction of NFAs
are also present in the work with fuzzy~auto\-mata.~Reduction of
the number of states of~fuzzy automata was studied in
\cite{BG.02,CM.04,LL.07,MMS.99,MM.02,P.06}, and the algorithms
given there were also based on the idea~of compu\-ting and merging
indistinguishable states.~They were called minimization
algorithms, but the term minimization is not adequate because it
does not involve the usual construction of the minimal~fuzzy
automaton in the set of all fuzzy automata recognizing a given
fuzzy language, but just the proce\-dure~of compu\-ting and
merging indistinguishable states.~Therefore, these are essentially
just state reduction algorithms.

In the deterministic case we can effectively detect and merge
indistinguishable~states, but in the non-deterministic case we
have sets of states and it is seemingly very difficult to decide
whether two states are distin\-guish\-able or not.~What we shall
do in this paper is find a superset such that one is certain not
to merge state~that should not be merged.~There can always be
states which could be merged but detecting those is too
computationally expensive.~In the case of fuzzy automata, this
problem is even worse because we work with fuzzy sets of
states.~However, it turned out that in the non-deterministic case
indistinguishability can be successfully
modelled~by~equi\-valences and quasi-orders. In
\cite{CSIP.07,CSIP.09} we have shown that in the fuzzy case the
indistinguishability can be mod\-elled by fuzzy equivalences, and
here we show that this can be done by
fuzzy~quasi-orders.~It~is~worth noting that in all previous papers
dealing with reduction of fuzzy automata
(cf.~\cite{BG.02,CM.04,LL.07,MMS.99,MM.02,P.06}) only reductions
by means of crisp equivalences have been investigated. In this
paper, as well as \cite{CSIP.07,CSIP.09}, show that better
reductions~can~be~achieved employing fuzzy relations, namely,
fuzzy equivalences and fuzzy quasi-orders.

In contrast to \cite{CSIP.07,CSIP.09}, where we have started from
a fuzzy equivalence on a set of states $A$ of a~fuzzy~auto\-maton
$\cal A$, here we start from an arbitrary fuzzy quasi-order $R$ on
$A$, we form the set $A/R$ of all~aftersets of~$R$, and we turn
the fuzzy transition function~on~$A$ into a related fuzzy
tran\-sition function on $A/R$. This results in the afterset fuzzy
automaton ${\cal A}/R$.~If, in addition, $\cal A$ is a fuzzy
recognizer, then we~also turn its fuzzy sets of initial and
terminal states into related fuzzy sets of initial and terminal
states of the afterset fuzzy recognizer ${\cal A}/R$.~In a similar
way, we construct the foreset fuzzy~recognizer of $\cal A$
w.r.t.~$R$, but we show that they are isomorphic, and hence, it is
enough to consider only afterset fuzzy automata and
recog\-nizers.~However, if we do not impose any restriction on
$R$, then the~afterset fuzzy recognizer ${\cal A}/R$ does not
necessary recognize the same fuzzy language as~$\cal
A$.~We~show~that~$\cal A$ and ${\cal A}/R$ recognize the same
fuzzy language, i.e., they are equivalent, if and only if $R$ is a
solution~to a particular system~of fuzzy rela\-tion equations
including $R$, as an unknown fuzzy quasi-order,
transition~rela\-tions on $\cal A$ and~fuzzy sets of initial and
terminal states.~This system, called the {\it general system\/},
has at least one solution in~the~set ${\cal Q}(A)$ of all fuzzy
quasi-orders on $A$, the equality relation on $A$. Nevertheless,
to obtain~the~best possible reduction~of $\cal A$, we have to find
the greatest solu\-tion to the general system in ${\cal Q}(A)$, if
it exists, or to find as big a solu\-tion as possible.~The general
system does not necessary have the greatest solution (Example \ref{ex:not.great}),
and~also, it may consist of infinitely many equations, and~finding its
nontrivial solutions may be a very difficult task. For that reason we
aim our attention to some instances~of the general system.~These
instances have to be as general as possible, but they have to be
easier to solve. From a practical point of view, these instances
have to consist of finitely many equations.

In Section \ref{sec:RLIFQO} we study two instances of the general
system whose solutions, called the right and left invariant fuzzy
quasi-orders, are common generalization of right and left
invariant quasi-orders and equivalences, studied in
\cite{CC.03,CC.04,INY.04,ISY.05}, and right and left invariant
fuzzy equivalences, studied in \cite{CSIP.07,CSIP.09}.~Using a
methodology similar to the one developed in \cite{CSIP.07,CSIP.09}
for fuzzy equivalences, we give a characterization of right
invariant fuzzy quasi-orders on a fuzzy automaton $\cal A$, and we
prove that they form a complete lattice whose greatest element
gives the best reduction of $\cal A$ by means of fuzzy
quasi-orders of this type.~Then~by Theorem \ref{th:alg} we give  a
procedure for computing the greatest right invariant fuzzy
quasi-order contained in a given fuzzy quasi-order, which works if
the underlying structure $\cal L$ of truth values is locally
finite,~but it does not necessary work if $\cal L$  is not locally
finite.~In particular, it works for classical fuzzy automata~over
the G\"odel structure, but it does not necessary work for fuzzy
automata over the Goguen~(product)~structure.~We also characterize
the greatest right invariant fuzzy quasi-order in the case when
the structure~$\cal L$ satisfies certain distributivity conditions
for join and multiplication over infima. This characterization
hold, for example, whenever multiplication is assumed to be a
continuous t-norm on the real unit interval, and hence, they hold
for {\L}ukasiewicz, Goguen and G\"odel structures.~Although the
results, as well as the methodology,~are~similar to those obtained
in \cite{CSIP.07,CSIP.09} for fuzzy equivalences, there are some
important differences which justify our study of state reductions
by means of fuzzy quasi-orders.~Generally, fuzzy quasi-orders and
fuzzy equivalences are equally good in the reduction of fuzzy
automata and recognizers, as we have~shown by  Theorem
\ref{th:lang.ER}.~However, Example~\ref{ex:R.ri} shows that the
right invariant fuzzy quasi-orders give better reductions than
right invariant fuzzy equivalences.~Moreover, the iterative
procedure for computing the greatest~right invariant fuzzy
quasi-order on a fuzzy automa\-ton $\cal A$, given in Theorem
\ref{th:alg}, can terminate in a finite number of steps even if a
similar iteration procedure for computing the greatest right
invariant fuzzy equivalence on  $\cal A$, developed in
\cite{CSIP.07,CSIP.09}, does not terminate in a finite number of
steps  (Example \ref{ex:prod2}).~It is worth~noting~that
the greatest right and left invariant fuzzy quasi-orders are calculated using iterative procedures, but these calculations are not approximative.~Whenever these procedures
terminate in a finite number of steps, exact solutions to the considered systems of fuzzy relation equations are obtained.

As we have noted, the procedure for computing the greatest right
invariant fuzzy quasi-orders on fuzzy automata does not necessary
work if the underlying structure $\cal L$ of truth values is not
locally finite.~For~that reason  in Section \ref{sec:APPROX} we consider some special types of right and left invariant fuzzy quasi-orders, and
we show that the greatest fuzzy quasi-orders of these types  can be effectively computed even if  $\cal L$ is not necessary locally finite.~By
Theorem~\ref{th:alg.c}~we give an  iterative procedure for
computing the greatest right invariant~crisp quasi-order contained
in a given crisp or fuzzy quasi-order.~This procedure works if
$\cal L$ is any complete residuated lattice, and~even if $\cal L$
is a lattice ordered monoid. On the other hand, as Example
\ref{ex:FQO-QO} shows, in cases when we are able to effectively
compute the~greatest~right invariant fuzzy quasi-order, using it,
better state reduction can be achieved than by using the greatest
right invariant crisp quasi-order.~We also study the strongly
right invariant fuzzy quasi-orders, which can be effectively
computed without any iteration procedure, by solving a simpler
system of~fuzzy relation equations. This procedure works if $\cal
L$ is any complete residuated lattice, even though we also show
that reductions by means of the greatest strongly right invariant
fuzzy quasi-orders give worse results than reductions by means of
the greatest right invariant ones (Example \ref{ex:ri.sri}).

In addition to special types of right and left invariant fuzzy
quasi-orders considered in Section \ref{sec:APPROX}, in
Section~\ref{sec:WEAKLY} we study some more general types of these
fuzzy quasi-orders -- the weakly right and left invariant fuzzy quasi-orders.~We show that the weakly right invariant fuzzy
quasi-orders on a fuzzy recognizer $\cal A$ form a principal ideal
of the lattice of quasi-orders on the set of states of $\cal A$.
We give a procedure for computing the great\-est element of this
principal ideal (Theorem \ref{th:syst.R.tau}), and we show that
weakly right invariant fuzzy quasi-orders give better reductions
than right invariant ones.~However, although the system of fuzzy
relation equations that defines the weakly right invariant fuzzy
quasi-orders consists of fuzzy relation equations whose greatest
solutions can be easily computed, computing the greatest solution
to the whole system is computationally hard.~Namely, the number of
equations may be exponential in the number of states of $\cal A$,
or it may~even be infinite.~This is an immediate consequence of
the fact that the procedure for computing the~greatest weakly
right invariant fuzzy quasi-order on $\cal A$ includes the
procedure for determinization of the reverse fuzzy recognizer of
$\cal A$ developed in \cite{ICB.08}, whereas the procedure for
computing the greatest weakly left invariant one includes
determinization of $\cal A$.

In Section 7 we show that even better results
in the state reduction can be obtained by alternating~reductions
by means of the greatest right and left invariant fuzzy
quasi-orders, or by means of the greatest weakly right and left
invariant ones.~First we show that if we  reduce a fuzzy
automaton~using~the greatest right invariant fuzzy quasi-order,
repeated reduction using  right invariant fuzzy quasi-orders can
not decrease the number of states. The number of states can be
decreased if we apply reduction by means of the greatest left
invariant fuzzy quasi-order.~The same observation is true for left
invariant fuzzy quasi-orders, as well as for weakly right and left
invariant ones.~We also show that alternate reductions starting
with~a~(weakly) right invariant fuzzy quasi-order, and those
starting with a (weakly) left invariant one,~can have different
lengths, and related alternate reducts can have different number
of states (Example \ref{ex:all.in.one}).~Moreover, there is no a
general procedure to decide which of these alternate reductions
will give better results.~Also,~there is no a general procedure to
decide whether we have reached the smallest number of states
in~alternate reductions.

Let us note that Champarnaud and Coulon \cite{CC.03,CC.04}, Ilie,
Navarro and Yu \cite{INY.04}, and~Ilie, Solis-Oba~and Yu
\cite{ISY.05} studied the state reduction of non-deterministic
recognizers by means of~right~and left invariant~quasi-orders.
However, they do not used the afterset or foreset recognizers
w.r.t.~a quasi-order $R$. Instead, they used~the factor recognizer
w.r.t.~the natural equivalence $E_R$ of $R$.~Although these
recognizers~are equivalent and have the same number of states,
there are some differences in their use if one works with
alternate reductions. Indeed, by Example \ref{ex:all.in.one} we
also show that in some cases alternate reductions by~means of
natural~equiva\-lences of right and left invariant quasi-orders
and alternate reductions by means of right and left invariant
equivalences do not decrease the number of states, while the
alternate reductions by~means of  right and left invariant
quasi-orders decrease this number.

Finally, in Section 8 we demonstrate some applications of weakly left invariant fuzzy quasi-orders in~the fuzzy discrete event systems theory.~We show that every fuzzy recognizer $\cal A$ is conflict-equivalent with~the afterset fuzzy recognizer ${\cal
A}/R$ w.r.t.~any weakly left invariant fuzzy quasi-order $R$ on $\cal A$.~For the
sake~of~con\-flict analysis, this means that in the parallel~composi\-tion of fuzzy recog\-nizers every component can be~replaced by such afterset fuzzy recognizer, what results in a smaller fuzzy recognizer to be analysed, and do not affect conflicting properties of the components.~It is also interesting to study applications of~fuzzy quasi-orders for reducing automaton states in other branches of the theory of discrete event systems, for example in the fault diagnosis, and these applications will be a subject of our future research.

Note again that the meaning of state reductions by means of fuzzy quasi-orders and fuzzy equivalences is in their possible effectiveness, as opposed to the minimization which is not effective.~However,~by~Theorem 3.5 we show that there exists a fuzzy recognizer such that no its state reduction by means of fuzzy quasi-orders or fuzzy equivalences provide a minimal fuzzy recognizer.

\section{Preliminaries}

\subsection{Fuzzy sets and relations}\medskip

In this paper we will use complete residuated lattices as
structures of membership values.~A {\it residuated lattice\/} is
an algebra ${\cal L}=(L,\land ,\lor , \otimes ,\to , 0, 1)$ such
that
\begin{itemize}
\parskip=-2pt
\item[{\rm (L1)}] $(L,\land ,\lor , 0, 1)$ is a lattice with the
least element $0$ and the greatest element~$1$, \item[{\rm (L2)}]
$(L,\otimes ,1)$ is a commutative monoid with the unit $1$,
\item[{\rm (L3)}] $\otimes $ and $\to $ form an {\it adjoint
pair\/}, i.e., they satisfy the {\it adjunction property\/}: for
all $x,y,z\in L$,
\begin{equation}\label{eq:adj}
x\otimes y \le z \ \iff \ x\le y\to z .
\end{equation}
\end{itemize}
If, in addition, $(L,\land ,\lor , 0, 1)$ is a complete lattice, then ${\cal L}$ is called a {\it
complete residuated lattice\/}.

The operations $\otimes $ (called {\it multiplication\/}) and $\to
$ (called {\it residuum\/}) are intended for modeling the
conjunction and implication of the corresponding logical calculus,
and supremum ($\bigvee $) and infimum ($\bigwedge $) are intended
for modeling of the existential and general quantifier,
respectively. An operation $\lra $ defined by
\begin{equation}\label{eq:bires}
x\lra y = (x\to y) \land (y\to x),
\end{equation}
called {\it biresiduum\/} (or {\it biimplication\/}), is used for
modeling the equivalence of truth values. It can be easily
verified that with respect to $\le $, $\otimes $ is isotonic in
both arguments, and $\to $ is isotonic in the second and antitonic
in the first argument. Emphasizing their monoidal structure, in
some sources residuated lattices are called integral, commutative,
residuated $\ell $-monoids \cite{Hohle.95}. It can be easily
verified that with respect to $\le $, $\otimes $ is isotonic in
both arguments, $\to $ is isotonic in the second and antitonic in
the first argument, and for any $x,y,z\in L$ and any
$\{x_i\}_{i\in I},\{y_i\}_{i\in I}\subseteq L$, the following
hold:
\begin{eqnarray}
&& x\to y \le x\otimes z\to y\otimes z ,\label{eq:res.mult} \\
&& \Bigl(\bigvee_{i\in I}x_i\Bigr)\otimes x=\bigvee_{i\in I}(x_i\otimes x) ,\label{eq:inf.dist} \\
&& \bigwedge_{i\in I}(x_i\to y_i)\le \Bigl(\bigwedge_{i\in I}x_i\Bigr)\to
\Bigl(\bigwedge_{i\in I}y_i\Bigr) \label{eq:res.inf.inf}\\
&& \bigwedge_{i\in I}(x_i\to y_i)\le \Bigl(\bigvee_{i\in I}x_i\Bigr)\to
\Bigl(\bigvee_{i\in I}y_i\Bigr). \label{eq:res.inf.sup}
\end{eqnarray}
For other properties of complete residuated lattices one can refer
to \cite{Bel.02,BV.05}.

The most studied and applied structures of truth values, defined
on the real unit interval $[0,1]$ with\break $x\land y =\min
(x,y)$ and $x\lor y =\max (x,y)$, are the {\it {\L}ukasiewicz
structure\/} ($x\otimes y = \max(x+y-1,0)$, $x\to y=
\min(1-x+y,1)$), the {\it Goguen} ({\it product\/}) {\it
structure\/} ($x\otimes y = x\cdot y$, $x\to y= 1$ if $x\le y$
and~$=y/x$ otherwise) and the {\it G\"odel structure\/} ($x\otimes
y = \min(x,y)$, $x\to y= 1$ if $x\le y$ and $=y$
otherwise).~More~generally, an algebra $([0,1],\land ,\lor ,
\otimes,\to , 0, 1)$ is a complete~resi\-duated lattice if and
only if $\otimes $ is a left-continuous t-norm and the residuum is
defined by $x\to y = \bigvee \{u\in [0,1]\,|\, u\otimes x\le y\}$.
Another important set of truth values is the set
$\{a_0,a_1,\ldots,a_n\}$, $0=a_0<\dots <a_n=1$, with $a_k\otimes
a_l=a_{\max(k+l-n,0)}$ and $a_k\to a_l=a_{\min(n-k+l,n)}$. A
special case of the latter algebras is the two-element Boolean
algebra of classical logic with the support $\{0,1\}$.~The only
adjoint pair on the two-element Boolean algebra consists of the
classical conjunction and implication operations.~This structure
of truth values we call the {\it Boolean structure\/}.~A
residuated~lattice $\cal L$ satisfying $x\otimes y=x\land y$ is
called a {\it Heyting algebra\/}, whereas a Heyting algebra
satisfying the prelinearity axiom $(x\to y)\lor (y\to x)=1$ is
called a {\it G\"odel algebra\/}. If any finitelly
generated~sub\-algebra of residuated lattice $\cal L$ is finite,
then $\cal L$ is called {\it locally finite\/}.~For example, every
G\"odel algebra, and hence, the G\"odel structure, is locally
finite, whereas the product structure is not locally finite.

In the further text $\cal L$ will be a complete residuated
lattice.~A {\it fuzzy subset\/} of a set $A$ {\it over\/} ${\cal
L}$, or~simply a {\it fuzzy subset\/} of $A$, is any mapping from
$A$ into $L$.~Ordinary crisp subsets~of~$A$ are considered as
fuzzy subsets of $A$ taking membership values in the set
$\{0,1\}\subseteq L$.~Let $f$ and $g$ be two fuzzy subsets of
$A$.~The {\it equality\/} of $f$ and $g$ is defined as the usual
equality of mappings, i.e., $f=g$ if and only if $f(x)=g(x)$, for
every $x\in A$. The {\it inclusion\/} $f\le g$ is also defined
pointwise: $f\le g$ if and only if $f(x)\le g(x)$, for every $x\in
A$. Endowed with this partial order the set $L^A$ of all fuzzy
subsets of $A$ forms a complete residuated lattice, in which the
meet (intersection) $\bigwedge_{i\in I}f_i$ and the join (union)
$\bigvee_{i\in I}f_i$ of an arbitrary family $\{f_i\}_{i\in I}$ of
fuzzy subsets of $A$ are mappings from $A$ into $L$ defined by
\[
\left(\bigwedge_{i\in I}f_i\right)(x)=\bigwedge_{i\in I}f_i(x), \qquad \left(\bigvee_{i\in
I}f_i\right)(x)=\bigvee_{i\in I}f_i(x),
\]
and the {\it product\/} $f\otimes g$ is a fuzzy subset defined by
$f\otimes g (x)=f(x)\otimes g(x)$, for every $x\in A$.
The {\it crisp~part\/} of a fuzzy subset $f$ of $A$ is a crisp
subset $\hat f=\{a\in A\,|\, f(a)=1\}$ of $A$.~We~will~also consider $\hat f$ as a
mapping $\hat f:A\to L$ defined by $\hat f(a)=1$, if $f(a)=1$, and $\hat f(a)=0$,
if $f(a)<1$.

A {\it fuzzy relation\/} on a set $A$ is any mapping from $A\times
A$ into $L$, that is to say, any fuzzy subset of $A\times A$, and
the equality, inclusion, joins, meets and ordering of fuzzy
relations are defined as for fuzzy sets.~The set of all fuzzy
relations on $A$ will be denoted by ${\cal R}(A)$.

For fuzzy relations $P,Q\in {\cal R}(A)$, their {\it
composition\/}~$P\circ Q$ is a fuzzy relation on $A$ defined by
\begin{equation}\label{eq:comp.rr}
(P \circ Q )(a,b)=\bigvee_{c\in A}\,P(a,c)\otimes Q(c,b),
\end{equation}
for all $a,b\in A$, and for a fuzzy subset $f$ of $A$ and a fuzzy
relation $P\in {\cal R}(A)$, the {\it compositions\/} $f\circ P$
and $P\circ f$ are fuzzy subsets of $A$ defined by
\begin{equation}\label{eq:comp.sr}
(f \circ P)(a)=\bigvee_{b\in A}\,f(b)\otimes P(b,a),\quad
(P \circ f)(a)=\bigvee_{b\in A}\,P(a,b)\otimes f(b),
\end{equation}
for any $a\in A$. Finally, for fuzzy subsets $f$ and $g$ of $A$ we write
\begin{equation}\label{eq:comp.ss}
f \circ g =\bigvee_{a\in A}\,f(a)\otimes g(a) .
\end{equation}
The value $f\circ g$ can be interpreted as the "degree of
overlapping" of $f$ and $g$.

For any $P,Q,R\in {\cal R}(A)$ and any $\{P_i\}_{i\in
I},\{Q_i\}_{i\in
I}\subseteq {\cal R}(A)$, the following hold:

\begin{eqnarray}
&& (P\circ Q)\circ R = P\circ (Q\circ R), \label{eq:comp.as} \\
&& P\le Q\ \ \text{implies}\ \ P\circ R \le  Q\circ R \ \ \text{and}\ \ R\circ P \le  R\circ Q, \label{eq:comp.mon} \\
&& P\circ \bigl(\bigvee_{i\in I}Q_i\bigr) = \bigvee_{i\in I}(P\circ Q_i) , \ \
\bigl(\bigvee_{i\in I}P_i\bigr)\circ Q = \bigvee_{i\in I}(P_i\circ Q) \label{eq:comp.sup} \\
&& P\circ \bigl(\bigwedge_{i\in I}Q_i\bigr) \le \bigwedge_{i\in I}(P\circ Q_i) , \ \
\bigl(\bigwedge_{i\in I}P_i\bigr)\circ Q \le \bigwedge_{i\in I}(P_i\circ Q). \label{eq:comp.inf}
\end{eqnarray}
We can also easily
verify that
\begin{equation}\label{eq:comp.as2}
(f\circ P)\circ Q=f\circ (P\circ Q), \quad (f\circ P)\circ g=f\circ (P\circ g),
\end{equation}
for arbitrary fuzzy subsets $f$ and $g$ of $A$, and fuzzy
relations $P$ and $Q$ on $A$, and hence, the parentheses~in
(\ref{eq:comp.as}) can be omitted.~For $n\in \Bbb N$, an $n$-th
power of a fuzzy relation $R$ on $A$ is a fuzzy relation $R^n$ on
$A$ defined inductively by $R^1=R$ and $R^{n+1}=R^n\circ R$.~We
also define $R^0$ to be the equality relation~on~$A$.

Note also that if $A$ is a finite set with $n$ elements, then $P$
and $Q$ can be treated~as $n\times n$ fuzzy matrices over $\cal L$
and $P\circ Q$ is the matrix~pro\-duct, whereas $f\circ P$ can be
treated as the product of a $1\times n$ matrix $f$ and an $n\times
n$ matrix $P$, and $P\circ f$ as the product of an $n\times n$
matrix $P$ and an $n\times 1$ matrix $f^t$ (the transpose~of~$f$).

A fuzzy relation $R$ on $A$ is said to be
\begin{itemize}\parskip=-3pt
\parskip=-2pt
\item[(R)] {\it reflexive\/} (or {\it fuzzy reflexive\/}) if $R (a,a)=1$, for every $a\in A$;
\item[(S)] {\it symmetric\/} (or {\it fuzzy symmetric\/}) if $R (a,b)=R (b,a)$, for all $a,b\in A$;
\item[(T)] {\it transitive\/} (or {\it fuzzy transitive\/}) if for all $a,b,c\in A$ we have
\[
R (a,b)\otimes R (b,c)\le R (a,c).
\]
\end{itemize}

For a fuzzy relation $R$ on a set $A$, a fuzzy relation $R^\infty $ on $A$ defined
by
\[
R^\infty = \bigvee_{n\in \Bbb N} R^n
\]
is the least transitive fuzzy relation on $A$ containing $R$, and it is called the
{\it transitive closure\/} of $R$.

A fuzzy relation on $A$ which is reflexive, symmetric and
transitive is called a
{\it fuzzy equivalence\/}.~With respect to the ordering of fuzzy
relations, the set ${\cal E}(A)$ of all fuzzy equivalences on $A$ is a complete
lattice,
in which the meet coincide with the ordinary intersection of fuzzy relations, but in the
general case, the join in ${\cal E}(A)$ does not coincide with the ordinary union of fuzzy
relations.

For a fuzzy~equi\-valence $E  $ on $A$ and $a\in A$ we define a
fuzzy subset $E _a$ of $A$ by $E_a(x)=E(a,x)$,~for~every $x\in
A$.~We call $E _a$ an {\it equivalence class\/} of $E$
deter\-mined by $a$.~The set $A/E  =\{E _a\,|\, a\in A\}$ is
called~the {\it factor set\/} of $A$ w.r.t. $E$ (cf.
\cite{Bel.02,CIB.07}). For an equivalence $\pi $ on $A$, the
related factor set will be denoted by $A/\pi $ and the equivalence
class of an element $a\in A$ by $\pi_a$.~A fuzzy equivalence $E$
on a set $A$ is called a {\it fuzzy equality\/} if for all $x,y\in
A$, $E(x,y)=1$ implies $x=y$. In other words, $E$ is a fuzzy
equality if and only if its crisp part $\widehat E$ is a crisp
equality.

A fuzzy relation on a set $A$ which is reflexive and transitive is
called a {\it fuzzy quasi-order\/}, and a~reflexive and transitive
crisp relation on $A$ is called a {\it quasi-order\/}.~In some
sources quasi-orders and fuzzy quasi-orders are called preorders
and fuzzy preorders (for example, see
\cite{CC.03,CC.04,INY.04,ISY.05}).~Note that a reflexive fuzzy
relation $R$ is a fuzzy quasi-order if and only if $R^2=R$.~With
respect to the ordering~of fuzzy relations,~the set ${\cal Q}(A)$
of all fuzzy quasi-orders on $A$ is a complete lattice, in
which~the~meet~coin\-cide with the ordinary intersection of fuzzy
relations. Nevertheless, in the general case, the join in ${\cal
Q}(A)$ does not coincide with the ordinary union of fuzzy
relations.~Namely, if $R$ is the join in ${\cal Q}(A)$ of a family
$\{R_i\}_{i\in I}$~of~fuzzy quasi-orders on $A$, then $R$ can be
represented by
\begin{equation}\label{eq:fqo.join}
R=\biggl( \bigvee_{i\in I}R_i \biggr)^{\!\infty} = \bigvee_{n\in \Bbb N}\biggl( \bigvee_{i\in I}R_i \biggr)^{\!n} .
\end{equation}
If $R$ is a fuzzy quasi-order on a set $A$, then a fuzzy relation
$E_R$  defined by $E_R=R\land R^{-1}$ is a fuzzy equivalence on
$A$, and is called a {\it natural fuzzy equivalence\/} of $R$.~A
fuzzy quasi-order $R$ on a set  $A$ is a {\it fuzzy order\/} if
for all $a,b\in A$, $R(a,b)=R(b,a)=1$ implies~$a=b$, i.e., if the
natural fuzzy equivalence~$E_R$ of $R$ is a fuzzy
equality.~Clearly, a fuzzy quasi-order $R$ is a fuzzy order if and
only if its crisp part $\widehat R$ is a crisp order.

It is worth noting that different concepts of a fuzzy order have
been discussed in literature~concerning fuzzy relations (for
example, see \cite{Bod.00,Bod.03,BDeBF.07,BK.04} and other sources
cited there). In particular, fuzzy~orders defined here differ from
fuzzy order\-ings defined in \cite{Bod.00,Bod.03,BK.04}.~

For more information about fuzzy sets and fuzzy logic we refer to the books \cite{Bel.02,BV.05,KY.95}, as
well as to recent papers \cite{Z.05,Z.08}, which review fuzzy
logic and uncertainty in a much broader perspective.

\subsection{Fuzzy automata and languages}\label{ssec:FAL}
\medskip

By a {\it fuzzy automaton over\/} $\cal L$, or simply a {\it fuzzy
automaton\/}, is defined as a triple ${\cal A}=(A,X,\delta^A)$,
where~$A$ and $X$ are the {\it set of states\/} and the {\it input
alphabet\/}, and $\delta^A:A\times X\times A\to L$~is~a fuzzy
subset~of $A\times X\times A$, called the {\it fuzzy transition
function\/}.~We can interpret $\delta ^{A}(a,x,b)$ as the
degree~to~which an~input letter $x\in X$~causes~a~transition from
a state $a\in A$ into a state $b\in A$.~The input alphabet $X$
will be always finite, but for methodological reasons we will
allow the set of states $A$ to be infinite.~A~fuzzy auto\-maton
whose set of states is finite is called a {\it fuzzy finite
automaton\/}.~Cardinality of a fuzzy automaton ${\cal
A}=(A,X,\delta^{A} )$, denoted as $|{\cal A}|$, is defined as the
cardinality of its set of states $A$.

Let $X^*$ denote the free monoid over the alphabet $X$, and  let $e\in X^*$ be the
empty word.~The mapping~$\delta^{A} $ can be extended~up to a mapping
$\delta_*^A:A\times X^*\times A\to L$ as follows: If $a,b\in A$, then
\begin{equation}\label{eq:delta.e}
\delta_*^A(a,e ,b)=\begin{cases}\ 1 & \text{if}\ a=b \\ \ 0 & \mbox{otherwise}
\end{cases},
\end{equation}
and if $a,b\in A$, $u\in X^*$ and $x\in X$, then
\begin{equation}\label{eq:delta.x}
\delta_*^A(a,ux ,b)= \bigvee _{c\in A} \delta_*^A(a,u,c)\otimes \delta^A (c,x,b).
\end{equation}

By (\ref{eq:inf.dist}) and Theorem 3.1 \cite{LP.05} (see also \cite{Qiu.01,Qiu.02,Qiu.06}), we have that
\begin{equation}\label{eq:delta.uv}
\delta_*^A(a,uv,b)= \bigvee _{c\in A} \delta_*^A(a,u,c)\otimes \delta_*^A(c,v,b),
\end{equation}
for all $a,b\in A$ and $u,v\in X^*$, i.e., if $w=x_1\cdots x_n$, for $x_1,\ldots ,x_n\in X$,
then
\begin{equation}\label{eq:delta.x1xn}
\delta_*^A(a,w,b)= \bigvee _{(c_1,\ldots ,c_{n-1})\in A^{n-1}}
\delta^{A}(a,x_1,c_1)\otimes \delta^{A}(c_1,x_2,c_2) \otimes\cdots \otimes \delta^{A}(c_{n-1},x_n,b).
\end{equation}
Intuitively, the product $\delta^{A}(a,x_1,c_1)\otimes \delta^{A}(c_1,x_2,c_2) \otimes\cdots \otimes
\delta^{A}(c_{n-1},x_n,b)$ represents the degree to which the input word $w$ causes a transition from a
state $a$ into a state $b$ through the sequence of intermediate states $c_1,\ldots ,c_{n-1}\in A$,
and $\delta_*^A(a,w,b)$ represents the supremum of degrees of all possible transitions from $a$ into
$b$ caused by $w$.

For any $u\in X^*$, and any $a,b\in A$ define a fuzzy relation
$\delta_u^A$ on $A$ by
\begin{equation}\label{eq:trans.rel}
\delta_u^A (a,b) = \delta_*^A (a,u,b),
\end{equation}
called the {\it fuzzy transition relation\/} determined by $u$.
Then (\ref{eq:delta.uv}) can be written as
\begin{equation}\label{eq:delta.uv.2}
\delta_{uv}^A= \delta_u^A\circ \delta_v^A,
\end{equation}
for all $u,v\in X^*$.

An {\it initial fuzzy automaton\/} is defined as a quadruple
${\cal A}=(A,X,\delta ^A,\sigma^A )$, where $(A,X,\delta^A )$ is a
fuzzy automaton and $\sigma^A \in L^A$ is the fuzzy set of {\it
initial states\/}, and a {\it fuzzy recognizer\/} is defined as a
five-tuple ${\cal A}=(A,X,\delta^A ,\sigma^A ,\tau^A )$, where
$(A,X,\delta^A,\sigma^A)$ is as above, and $\tau^A \in L^A$ is the
fuzzy set of~{\it terminal states\/}. We also say that $\cal A$ is
a {\it fuzzy recognizer belonging to the fuzzy automaton\/}
$(A,X,\delta^A)$.

A {\it fuzzy language\/} in $X^*$ over $\cal L$, or briefly a {\it
fuzzy language\/}, is any fuzzy subset of~$X^*$, i.e., any mapping
from $X^*$ into $L$.~A {\it fuzzy language recognized by a fuzzy
recognizer\/} ${\cal A}=(A,X,\delta^{A} ,\sigma^{A} ,\tau^{A} )$,
denoted as $L({\cal A})$, is a fuzzy language in $L^{X^*}$ defined
by
\begin{equation}\label{eq:recog}
L({\cal A})(u) = \bigvee_{a,b\in A} \sigma^{A} (a)\otimes \delta_*^A(a,u,b)\otimes \tau^A(b) ,
\end{equation}
or equivalently,
\begin{equation}\label{eq:recog.comp}
\begin{aligned}
L({\cal A})(e) &= \sigma^A\circ \tau^A ,\\
L({\cal A})(u) &= \sigma^A\circ \delta_{x_1}^A\circ \delta_{x_2}^A\circ \cdots
\circ \delta_{x_n}^A\circ \tau^A ,
\end{aligned}
\end{equation}
for any $u=x_1x_2\dots x_n\in X^+$, where $x_1,x_2,\ldots ,x_n\in
X$.~In other words, the equality (\ref{eq:recog}) means that~the
membership degree of the word $u$~to~the fuzzy language $L({\cal
A})$ is equal to the degree to which $\cal A$ recognizes or
accepts the word $u$.~

The~{\it reverse fuzzy auto\-maton\/} of a fuzzy automaton ${\cal
A}=(A,X,\delta ^{A})$ denoted as $\bar {\cal
A}=(A,X,\bar\delta^{A} )$, is a fuzzy~automa\-ton  with the fuzzy
transition function defined by $\bar\delta^{A} (a,x,b)=\delta
^{A}(b,x,a)$, for all $a,b\in A$ and $x\in X$. A {\it reverse
fuzzy recognizer\/} of a fuzzy recognizer ${\cal A}=(A,X,\delta
^{A},\sigma^A,\tau^A)$ is  $\bar {\cal
A}=(A,X,\bar\delta^{A},\bar\sigma^A,\bar\tau^A )$,~a~fuzzy
recognizer  with the fuzzy transition function $\bar\delta^{A}$
defined as above, and fuzzy sets of initial and terminal states
defined by $\bar\sigma^A =\tau^A$ and $\bar\tau^A =\sigma^A$.

Fuzzy~auto\-mata ${\cal A}=(A,X,\delta^{A} )$ and ${\cal
A}'=(A',X,\delta^{A'} )$ are {\it isomorphic\/} if there is a
bijective mapping $\phi :A\to A'$ such that $\delta^{A}
(a,x,b)=\delta^{A'}(\phi(a),x,\phi(b))$, for all $a,b\in A$ and
$x\in X$.~It is easy to check that in this case we also have that
$\delta_*^A(a,u,b)=\delta_*^{A'}(\phi(a),u,\phi(b))$, for all
$a,b\in A$ and $u\in X^*$.~Similarly, fuzzy recognizers ${\cal
A}=(A,X,\delta^{A},\sigma^A,\tau^A)$ and ${\cal
A}'=(A',X,\delta^{A'}, \sigma^{A'},\tau^{A'} )$ are {\it
isomorphic\/} if there exists a bijective mapping $\phi :A\to A'$
such that $\delta^{A} (a,x,b)=\delta^{A'}(\phi(a),x,\phi(b))$, for
all $a,b\in A$ and $x\in X$, and also,
$\sigma^A(a)=\sigma^{A'}(\phi (a))$ and $\tau^A(a)=\tau^{A'}(\phi
(a))$, for every $a\in A$.

If ${\cal A}=(A,X,\delta^{A} )$ is a fuzzy automaton such that
$\delta^A$ is a crisp relation, then $\cal A$ is an ordinary crisp
{\it non-deterministic automaton\/}, while if $\delta^{A} $ is a
mapping of $A\times X$ into~$A$, then $\cal A$ is an ordinary {\it
deterministic automaton\/}.~Evidently, in these two cases we have
that $\delta_*^A$ is also a crisp subset of $A\times X^*\times A$,
and a mapping of $A\times X^*$ into $A$, respectively.~In other
words, non-deterministic automata are fuzzy automata over the
Boolean structure.~If  ${\cal A}=(A,X,\delta^{A},\sigma^A,\tau^A)$
such that $\delta^A$ is a crisp relation and $\sigma^A$ and
$\tau^A$ are crisp subsets of $A$, then $\cal A$ is called a {\it
non-deterministic recognizer\/}.

\medskip

For undefined notions and notation one can refer to
\cite{Bel.02,BV.05,MM.02}.

\section{Afterset and foreset fuzzy automata}

Let $R$ be a fuzzy quasi-order on a set $A$.~For each $a\in A$,
the {\it $R$-afterset\/} of $a$ is the fuzzy set $R_a\in L^{A}$
defined by $R_a(b)=R(a,b)$, for any $b\in A$,~while the {\it
$R$-foreset\/} of $a$ is the fuzzy set $R^a\in L^{A}$ defined by
$R^a(b)=R(b,a)$, for any $b\in A$ (see
\cite{BK.80,DeBK.94,DeCK.04}).~The set of all $R$-aftersets will
be denoted~by~$A/R$, and the set of all $R$-foresets will be
denoted by $A\backslash R$.~Clearly, if $R$ is~a~fuzzy
equivalence, then $A/R=A\backslash R$ is the set of all
equivalence classes of $R$.

If $f$ is an arbitrary fuzzy subset of $A$, then fuzzy relations $R_f$ and $R^f$
on $A$ defined by
\begin{equation}\label{eq:Rf}
R_f(a,b)=f(a)\to f(b), \ \ \ R^f(a,b)=f(b)\to f(a),
\end{equation}
for all $a,b\in A$, are fuzzy quasi-orders on $A$.~In particular, if $f$ is a normalized
fuzzy subset of $A$, then it is an afterset of $R_f$ and a foreset of $R^f$.

\begin{theorem}\label{th:ER} Let $R$ be a fuzzy quasi-order on a set
$A$ and $E$ the natural fuzzy equivalence of $R$.~Then
\begin{itemize}\itemindent10pt\parskip=-2pt
\item[{\rm (a)}] For arbitrary $a, b\in A$ the following conditions are equivalent:
\begin{itemize}\itemindent10pt\parskip=0pt
\item[\rm{(i)}] $E(a,b)=1$;
\item[\rm{(ii)}] $E_a = E_b$;
\item[\rm{(iii)}] $R_a = R_b$;
\item[\rm{(iv)}] $R^a = R^b$.
\end{itemize}
\item[{\rm (b)}] Functions $R_a\mapsto E_a$ of $A/R$ to $A/E$, and $R_a\mapsto R^a$ of $A/R$ to $A\backslash R$, are bijective functions.
\end{itemize}
\end{theorem}

\begin{proof} (a) Consider arbitrary $a,b\in A$.

(i)$\implies $(ii). Let $E(a,b)=1$, that is $R(a,b)=R(b,a)=1$.Then
for every $c\in A$ we have that
\[
R_b(c)=R(b,c)=R(a,b)\otimes R(b,c)\le R(a,c)=R_a(c),
\]
whence $R_b\le R_a$. Analogously we prove that $R_a\le R_b$, and
therefore, $R_a=R_b$.

(ii)$\implies $(i). Let $R_a = R_b$. Then
\[
R(a,b)=R_a(b)\ge R_b(b)=R(b,b)=1,
\]
which yields $R(a,b)=1$. Analogously we prove that $R(b,a)=1$, and hence, $E(a,b)=1$.

Equivalence (i)$\iff$(iii) can be proved similarly as (i)$\iff $(ii).

The assertion (b) follows immediately by (a).
\end{proof}

Let us consider the G\"odel structure and a fuzzy quasi-order $R$ on a set $A$  given by
\[
R=\begin{bmatrix}
1 & 0.3 & 0.3 \\
0 &  1  & 0.2 \\
0 &  1  &  1
\end{bmatrix}.
\]
The natural fuzzy equivalence $E_R$ of $R$ is calculated by
$E_R(a,b)=R(a,b)\wedge R^{-1}(a,b)=R(a,b)\wedge R(b,a)$, i.e.
\[
E_R=\begin{bmatrix}
1 &  0  &  0 \\
0 &  1  & 0.2 \\
0 & 0.2 &  1
\end{bmatrix}.
\]

If $A$ is a finite set with $n$ elements and a fuzzy quasi-order
$R$ on $A$ is treated as an $n\times n$ fuzzy matrix over $\cal
L$, then $R$-aftersets are row vectors, whereas $R$-foresets are
column vectors of this matrix.~The~previous theorem says that
$i$-th and $j$-th row vectors of this matrix  are equal if and
only if its $i$-th and $j$-th column vectors are equal, and vice
versa.~Moreover, we have that $R$ is a fuzzy order if and only if
all~its~row~vectors are different, or equivalently, if and only if
all its column vectors are different.

Let ${\cal A}=(A, X, \delta^A)$ be a fuzzy automaton and let $R$ be a
fuzzy quasi-order on $A$.~We can define~the~fuzzy transition function
$\delta^{A/R}:A/R\times X\times A/R\to L$ by
\begin{equation}\label{eq:aft.aut}
\delta^{A/R}(R_a,x,R_b)=\bigvee_{a',b'\in
A}R(a,a')\otimes\delta^A(a',x,b')\otimes R(b',b),
\end{equation}
or equivalently
\begin{equation}\label{eq:aft.aut.1}
\delta^{A/R}(R_a,x,R_b)=(R\circ\delta_x^A\circ R)(a,b)=R_a\circ \delta_x^A\circ R^b,
\end{equation}
for all $a,b\in A$ and $x\in X.$ According to the statement (a) of
Theorem \ref{th:ER}, $\delta^{A/R}$ is well-defined, and we~have that ${\cal
A}/R=(A/R,X,\delta^{A/R})$ is a fuzzy automaton, called the {\it
afterset fuzzy automaton\/} of $\cal A$ w.r.t.~$R$.

In addition, if ${\cal A}=(A,X,\delta^A,\sigma^A,\tau^A)$ is a
fuzzy recognizer, then we define the fuzzy transition function
$\delta^{A/R}$ as in (\ref{eq:aft.aut}), and we also define a
fuzzy set $\sigma^{A/R}\in L^{A/R}$ of initial~states~and a fuzzy
set $\tau^{A/R}\in L^{A/R}$ of terminal states by
\begin{gather}
\sigma^{A/R}(R_a) = \bigvee_{a'\in A}\sigma^{A} (a')\otimes R(a',a) = (\sigma^A\circ
R)(a) = \sigma^A\circ R^a , \label{eq:sE} \\
\tau^{A/R}(R_a) = \bigvee_{a'\in A}R(a,a')\otimes \tau^A(a') = (R\circ\tau^A)(a) = R_a\circ\tau^A , \label{eq:tE}
\end{gather}
for any $a\in A$.~According to (a) of Theorem \ref{th:ER},
$\sigma^{A/R}$ and $\tau^{A/R}$ are well-defined functions, and~we~have that
${\cal A}/R=(A/R,X,\delta^{A/R},\sigma^{A/R},\tau^{A/R})$ is a~fuzzy recognizer, which is called the
{\it afterset fuzzy recognizer\/} of $\cal A$ w.r.t.~$R$.~

Analogously, for a fuzzy automaton ${\cal A}=(A, X, \delta^A)$, the {\it  foreset fuzzy automaton\/} of $\cal A$ w.r.t.~$R$ is a fuzzy
automaton ${\cal A}\backslash R=(A\backslash R,X,\delta^{A\backslash R})$ with the
fuzzy transition function $\delta^{A\backslash R}$ defined by
\begin{equation}\label{eq:for.aut}
\delta^{A\backslash R}(R^a,x,R^b)=\bigvee_{a',b'\in
A}R(a,a')\otimes\delta^A(a',x,b')\otimes R(b',b) =(R\circ\delta_x^A\circ R)(a,b)=R_a\circ \delta_x^A\circ R^b,
\end{equation}
for all $a,b\in A$ and $x\in X$. In addition, for a fuzzy
recognizer ${\cal A}=(A,X,\delta^A,\sigma^A,\tau^A)$, the {\it
foreset fuzzy recognizer\/} of $\cal A$ w.r.t.~$R$ is a fuzzy
recognizer ${\cal A}\backslash R=(A\backslash
R,X,\delta^{A\backslash R},\sigma^{A\backslash
R},\tau^{A\backslash R})$ with a  a fuzzy set $\sigma^{A\backslash
R}\in L^{A\backslash R}$ of initial~states~and a fuzzy set
$\tau^{A\backslash R}\in L^{A\backslash R}$ of terminal states by
\begin{gather}
\sigma^{A\backslash R}(R^a) = \bigvee_{a'\in A}\sigma^{A} (a')\otimes R(a',a) = (\sigma^A\circ
R)(a) = \sigma^A\circ R^a , \label{eq:sE2} \\
\tau^{A\backslash R}(R^a) = \bigvee_{a'\in A}R(a,a')\otimes \tau^A(a') = (R\circ\tau^A)(a) = R_a\circ\tau^A , \label{eq:tE2}
\end{gather}
for any $a\in A$.~

We can easily prove the following:

\begin{theorem}\label{th:aft.for}
For any fuzzy quasi-order $R$ on a fuzzy recognizer {\rm ({\it automaton\/})} ${\cal A}$ the afterset fuzzy recog\-nizer {\rm ({\it automaton\/})} ${\cal A}/R$ and the foreset fuzzy recognizer {\rm ({\it automaton\/})} ${\cal A}\backslash R$ are isomorphic.
\end{theorem}

\begin{proof} This follows immediately by (\ref{eq:aft.aut}), (\ref{eq:for.aut})
and (b) of Theorem \ref{th:ER}.
\end{proof}

In view of Theorem \ref{th:aft.for}, in the remainder of this
paper we will consider only afterset fuzzy recognizers and
automata. We will~see~in Example \ref{ex:R.ri} that
the~factor~fuzzy recognizer (automaton) ${\cal A}/E_R$ of $\cal
A$, w.r.t.~the natural fuzzy equivalence $E_R$ of~$R$, is not
necessary isomorphic to~fuzzy recognizers  ${\cal A}/R$ and ${\cal
A}\backslash R$, but by (b) of Theorem \ref{th:ER},~it has the
same cardinality as ${\cal A}/R$ and ${\cal A}\backslash R$,~and
if  $L({\cal A})=L({\cal A}/R)$ $(=L({\cal A}\backslash R)$, then
we also have that $L({\cal A})=L({\cal A}/E_R)$.

If ${\cal A}=(A,X,\delta^A)$ is a fuzzy automaton and $R$ is a fuzzy quasi-order
on $A$, then we also define a new fuzzy transition function $\delta^{A|R}:A\times
X\times A\to L$ by
\[
\delta^{A|R}(a,x,b)=(R\circ \delta^A_x\circ R )(a,b), \ \ \ \text{for all
$a,b\in A$ and $x\in X$,}
\]
i.e., $\delta^{A|R}_x=R\circ \delta^A_x\circ R$, for each $x\in X$, and we
obtain a new fuzzy automaton ${\cal A}|R = (A,X,\delta^{A|R})$~with~the same
set of states and input alphabet as the original one.~Furthermore, if ${\cal A}=(A,X,\delta^A,\sigma^A,\tau^A)$ is a~fuzzy recognizer, then we also set  $\sigma^{A|R}=\sigma^{A}$ and $\tau^{A|R}=\tau^A$,
and we have that ${\cal A}|R = (A,X,\delta^{A|R},\sigma^{A|R},\tau^{A|R})$ is a fuzzy
recognizer.

The following theorem can be conceived as a version of the
well-known Second Isomor\-phism Theorem, concerning fuzzy automata
and fuzzy quasi-orders~on~them. (cf. \cite{BS.81}, \S 2.6).

\begin{theorem}\label{th:SIT}
Let ${\cal A}=(A,X,\delta^{A},\sigma^A,\tau^A)$ be a fuzzy recognizer and let $R$ and $S$ be fuzzy quasi-orders on $\cal A$ such that $R\le S$.~Then a fuzzy relation $S/R$ on $A/R$ defined~by
\begin{equation}\label{eq:EdE=d}
S/R(R_a,R_b)=S(a,b), \qquad\text{for all $a,b\in A$},
\end{equation}
is a fuzzy quasi-order on ${A}/R$ and  fuzzy recognizers ${\cal A}/S$,
$({\cal A}/R)/(S/R)$ and $({\cal A}|R)/S$ are isomorphic.
\end{theorem}

\begin{proof}
Let $a,a',b,b'\in A$ such that $R_a=R_{a'}$
and $R_b=R_{b'}$, i.e., $E_{R}(a,a')=E_{R}(b,b')=1$.~Since~$R\le S$, we also have
that $R^{-1}\le S^{-1}$, whence $E_R\le E_S$, and by this it follows that
$E_{S}(a,a')=E_{S}(b,b')=1$, so $S(a,b)=S(a',b')$.~Therefore, $S/R$ is~a~well-defined fuzzy relation, and clearly, $S/R$ is a fuzzy quasi-order.

For the sake of simplicity set $S/R=Q$.
Define a mapping $\phi :A/S\to (A/R)/Q$ by
\[
\phi (S_a)=Q_{R_a}, \qquad \text{for every $a\in A$}.
\]
According to Theorem \ref{th:ER}, for arbitrary $a,b\in A$  we have that
\[
S_a=S_b \ \Leftrightarrow\ \ S(a,b)=S(b,a)=1 \ \Leftrightarrow \ Q(R_a,R_b)=Q(R_b,R_a)=1 \ \Leftrightarrow \  Q_{R_a} = Q_{R_b} \ \Leftrightarrow \ \phi (S_a)=\phi
(S_b),
\]
and hence, $\phi $ is a well-defined and injective function.~It is clear that $\phi
$ is also a surjective function. Thus, $\phi $ is a bijective function of $A/S$
onto $(A/R)/Q$.

Since $R\le S$ implies $R\circ S=S\circ R=S$, for arbitrary $a,b\in
A$ and $x\in X$ we have that
\[
\begin{aligned}
\delta_x^{(A/R)/Q}(\phi (S_a),\phi (S_b)) &= \delta_x^{(A/R)/Q}(Q_{R_a},Q_{R_b})= (Q\circ \delta_x^{A/R}\circ Q)(R_a,R_b) \\
&=  \bigvee_{c,d\in A} Q(R_a,R_c)\otimes \delta_x^{A/R}(R_c,R_d)\otimes Q(R_d,R_b) \\
&=  \bigvee_{c,d\in A} S(a,c)\otimes (R\circ \delta_x^A\circ R)(c,d)\otimes S(d,b)\\
&= (S\circ R\circ \delta_x^A\circ R\circ S)(a,b) = (S\circ \delta_x^A\circ S)(a,b) = \delta_x^{A/S}(S_a,S_b) .
\end{aligned}
\]
Moreover, for any $a\in A$ we have that
\[
\sigma^{(A/R)/Q}(\phi(S_a))= \sigma^{(A/R)/Q}(Q_{R_a})=\sigma^{A/R}({R_a})= \sigma^{A}(a)=
\sigma^{A/S}({S_a}),
\]
and similarly, $\tau^{(A/R)/Q}(\phi(S_a))= \tau^{A/S}({S_a})$.~Therefore, $\phi $ is an isomorphism of the fuzzy recognizer ${\cal A}/S$ onto the
fuzzy recognizer $({\cal A}/R)/(S/R)$.

Next, for all $a,b\in A$ and $x\in X$ we have that
\[
\begin{aligned}
\delta^{(A|R)/S}(S_a,x,S_b) & = (S\circ \delta_x^{A|R}\circ
S)(a,b) =(S\circ R\circ \delta_x^A\circ R\circ S)(a,b) \\
&= (S\circ \delta_x^A\circ S)(a,b) =  \delta^{A/S}(S_a,x,S_b),
\end{aligned}
\]
and $\sigma^{(A|R)/S}=\sigma^{A/S}$,  $\tau^{(A|R)/S}=\tau^{A/S}$,
so fuzzy recognizers $({\cal A}|R)/S$ and ${\cal A}/S$ are isomorphic.
\end{proof}

If in the proof of the previous theorem we disregard fuzzy sets of initial and terminal
states, we see~that the theorem also hold for fuzzy automata.

\begin{remark}\label{re:afrcardinality}\rm
For any given fuzzy quasi-order $R$ on a fuzzy recognizer ${\cal A}=(A,X,\delta^{A},\sigma^A,\tau^A)$, the rule $a\mapsto R_a$ defines a surjective function of $A$ onto $A/R$.~This means that the afterset fuzzy recognizer ${\cal A}/R$ has smaller or equal cardinality than the fuzzy recognizer ${\cal A}$.

Now, if $R$ and $S$ are fuzzy quasi-orders on ${\cal A}$ such that $R\le S$, according to Theorem \ref{th:SIT}, the afterset fuzzy recognizer ${\cal A}/S$ has smaller or equal cardinality than ${\cal A}/R$.~This fact will be frequently used in the rest of the paper.
\end{remark}

Let us note that if ${\cal A}$ is a fuzzy recognizer or a fuzzy automaton, $A$ is
its set of states, and $R$, $S$ and~$T$ are fuzzy quasi-orders~on~$A$~such that $R\le S$ and $R\le
T$, then
\begin{equation}\label{eq:FE.GE}
S\le T \ \Leftrightarrow \  S/R\le T/R ,
\end{equation}
and hence, a mapping $\Phi:{\cal Q}_R(A)=\{S\in {\cal Q}(A)\mid R\le S\}\to {\cal Q}(A/R)$, given by $\Phi :S\mapsto S/R$, is injective (in fact, it
is an order isomorphism of ${\cal Q}_R(A)$ onto a subset of ${\cal Q}(A/R)$).~In particular, for~a~fuzzy~quasi-order $R$ on $A$, the fuzzy relation $R/R$ on $A/R$
will be denoted by $\widetilde R$.~It can be easily verified that $\widetilde R$
is a fuzzy order on $A/R$, and if $E$ is a fuzzy equivalence on $A$, then $\widetilde E$ is a fuzzy equality on $A/E$.

For a fuzzy recognizer ${\cal A}=(A,X,\delta^A,\sigma^A,\tau^A)$ and a fuzzy quasi-order
$R$ on $A$ we have that
the fuzzy language $L({\cal A}/R)$ recognized~by the afterset fuzzy recognizer ${\cal A}/R$ is~given~by
\begin{equation}\label{eq:LAR}
\begin{aligned}
L({\cal A}/R)(e) &= \sigma^A\circ R\circ \tau^A ,\\
L({\cal A}/R)(u) &= \sigma^A\circ R\circ \delta_{x_1}^A\circ R\circ \delta_{x_2}^A\circ R \circ \cdots \circ R\circ \delta_{x_n}^A\circ R\circ \tau^A ,
\end{aligned}
\end{equation}
whereas the fuzzy language $L({\cal A})$ recognized by ${\cal A}$ is given by
\[
\begin{aligned}\label{eq:LA}
L({\cal A})(e) &= \sigma^A\circ \tau^A ,\\
L({\cal A})(u) &= \sigma^A\circ \delta_{x_1}^A\circ \delta_{x_2}^A\circ \cdots
\circ \delta_{x_n}^A\circ \tau^A ,
\end{aligned}
\]
for any $u=x_1x_2\dots x_n\in X^+$, where $x_1,x_2,\ldots ,x_n\in X$.~Let us note
that the equation
(\ref{eq:LAR}) follows immediately by definition of the afterset fuzzy recognizer
${\cal A}/R$ (the equations (\ref{eq:aft.aut.1}), (\ref{eq:sE}) and (\ref{eq:tE})),
by the equations (\ref{eq:comp.as}) and (\ref{eq:comp.as2}), and the fact that $R\circ
R=R$, for every fuzzy quasi-order $R$.~Hence, the fuzzy recognizer $\cal A$~and
the~afterset fuzzy recognizer ${\cal A}/R$ are {\it equivalent\/}, i.e., they recognize the same fuzzy language, if and only~if the fuzzy quasi-order $R$ is a solution to a system of fuzzy relation equations
\begin{equation}\label{eq:gen.syst}
\begin{aligned}
&\sigma^A\circ\tau^A=\sigma^A\circ R\circ\tau^A, \\
&\sigma^A\circ \delta_{x_1}^A\circ \delta_{x_2}^A\circ \cdots \circ
\delta_{x_n}^A\circ \tau^A = \sigma^A\circ R\circ \delta_{x_1}^A\circ R\circ \delta_{x_2}^A\circ R \circ \cdots
\circ R\circ \delta_{x_n}^A\circ R\circ \tau^A,
\end{aligned}
\end{equation}
for all $n\in \Bbb N$ and $x_1,x_2,\ldots ,x_n\in X$.~We will
call (\ref{eq:gen.syst})  the {\it general~system\/}.

The general system has at least one solution in ${\cal Q}(A)$,~the equality relation on $A$.~It will be called
the {\it trivial solution\/}.~To attain the best possible reduction of $\cal A$, we have to find the greatest
solu\-tion to the general system in ${\cal Q}(A)$, if it exists, or to find as big a solu\-tion as
possible.~However, the general system~does not necessary have the greatest solution (see Example 3.2), and also, it
may consist of infinitely many equations, and finding its nontrivial solutions may be a very difficult task.~For that reason we will aim our attention to some instances of
the general system.~These instances have to be as general as possible, but they have to be easier to
solve. From a practical point of view, these instances have to consist of finitely many equations.

The following theorem describes some properties of the set of all solutions to the
general system.

\begin{theorem}\label{th:lang.ER}
Let ${\cal A}=(A,X,\delta^A,\sigma^A,\tau^A)$ be a fuzzy recognizer.

The set of all solutions to the general system in ${\cal Q}(A)$
is an order ideal of the lattice ${\cal Q}(A)$.

Consequently,~if a fuzzy quasi-order $R$ on $A$ is a solution to the general system,~then its natural~fuzzy equivalence $E_R$ is also a solution to the general system.
\end{theorem}

\begin{proof}
Consider arbitrary $n\in \Bbb N$, $x_1,x_2,\ldots ,x_n\in X$, and fuzzy quasi-orders $R$ and $S$ on $A$ such that $S$ is a solution to the general system and $R\le S$.~By the facts that $S$ is a solution to the general system~and $R\le S$,
by reflexivity of $R$, and by (\ref{eq:comp.mon}) we obtain that
\[
\begin{aligned}
\sigma^A\circ \delta_{x_1}^A\circ \delta_{x_2}^A\circ \cdots \circ
\delta_{x_n}^A\circ \tau^A &\le \sigma^A\circ R\circ \delta_{x_1}^A\circ R\circ \delta_{x_2}^A\circ R \circ \cdots \circ R\circ \delta_{x_n}^A\circ R\circ \tau^A\\
&\le \sigma^A\circ S\circ \delta_{x_1}^A\circ S\circ \delta_{x_2}^A\circ S \circ \cdots
\circ S\circ \delta_{x_n}^A\circ S\circ \tau^A \\
&= \sigma^A\circ \delta_{x_1}^A\circ \delta_{x_2}^A\circ \cdots \circ
\delta_{x_n}^A\circ \tau^A ,
\end{aligned}
\]
and hence, $R$ is a solution to the general system.~By this it follows that solutions
to the general system in ${\cal Q}(A)$ form an order ideal of the lattice ${\cal Q}(A)$.

The second part of the theorem follows immediately by the fact that $E_R=R\land
R^{-1}\le R$.
\end{proof}

The following example shows that there are fuzzy quasi-orders which are not solutions
to the general system, but their natural fuzzy equivalences are solutions to this
system.
\begin{example}\rm
Let $\cal L$ be the Boolean structure, let ${\cal A}=(A,X,\delta^A,\sigma^A,\tau^A)$ be a fuzzy recognizer~over~$\cal L$, where $A=\{1,2,3\}$, $X=\{x,y\}$, and  $\delta_x^A$,
$\delta_y^A$, $\sigma^A$ and  $\tau^A$ are given by
\[
\delta_x^A=\begin{bmatrix}
1 & 0 & 0 \\
0 & 0 & 0 \\
0 & 0 & 0
\end{bmatrix},\ \ \ \
\delta_y^A=\begin{bmatrix}
0 & 1 & 0 \\
1 & 1 & 1 \\
1 & 0 & 0
\end{bmatrix},\ \ \ \
\sigma^A=\begin{bmatrix}
1 & 1 & 1\end{bmatrix},
\ \ \ \ \tau^A=\begin{bmatrix}
1 \\
0 \\
1
\end{bmatrix},
\]
and consider a fuzzy quasi-order $R$ on $A$  given by
\[
R=\begin{bmatrix}
1 & 1 & 1 \\
0 & 1 & 1 \\
0 & 0 & 1
\end{bmatrix}.
\]
Then we have that
\[
\sigma^A\circ R\circ \delta_x^A\circ R\circ \delta_y^A\circ R\circ \tau^A=
1\ne 0 =\sigma^A\circ  \delta_x^A\circ  \delta_y^A\circ  \tau^A ,
\]
so $R$ is not a solution to the general system, but its natural fuzzy
equivalence $E_R$ is the equality relation on $A$, and hence, it is a solution
to the general system.
\end{example}

The next example shows that the general system does not necessary have the greatest
solution.

\begin{example}\label{ex:not.great}\rm
Let $\cal L$ be the Boolean structure, let ${\cal A}=(A,X,\delta^A,\sigma^A,\tau^A)$ be a fuzzy recognizer~over~$\cal L$, where $A=\{1,2,3\}$, $X=\{x\}$, and  $\delta_x^A$,
$\sigma^A$ and  $\tau^A$ are given by
\[
\delta_x^A=\begin{bmatrix}
0 & 1 & 0 \\
0 & 0 & 0 \\
0 & 0 & 0
\end{bmatrix},\ \ \ \
\sigma^A=\begin{bmatrix}
1 & 1 & 1\end{bmatrix},
\ \ \ \ \tau^A=\begin{bmatrix}
1 \\
1 \\
1
\end{bmatrix},
\]
and consider fuzzy quasi-orders (in fact, fuzzy equivalences) $E$ and $F$ on $A$ given by
\[
E=\begin{bmatrix}
1 & 0 & 0 \\
0 & 1 & 1 \\
0 & 1 & 1
\end{bmatrix},\ \ \ \
F=\begin{bmatrix}
1 & 0 & 1 \\
0 & 1 & 0 \\
1& 0 & 1
\end{bmatrix}.
\]
We have that both $E$ and $F$ are solutions to the general system (since $E$ is right
invariant and $F$ is left invariant, see the next section for details).~On the other
hand, the join of $E$ and $F$ in the lattice ${\cal Q}(A)$ is a fuzzy quasi-order
$U$ given by
\[
U=\begin{bmatrix}
1 & 1 & 1 \\
1 & 1 & 1 \\
1 & 1 & 1
\end{bmatrix},
\]
and it is not a solution to the general system, since
\[
\sigma^A\circ U\circ \delta_x^A\circ U\circ \delta_x^A\circ U\circ \tau^A=1\ne 0=
\sigma^A\circ \delta_x^A\circ \delta_x^A\circ \tau^A .
\]
If the general system would have the greatest solution $R$ in ${\cal Q}(A)$,
then $E\le R$ and $F\le R$ would imply $U\le R$, and by
Theorem \ref{th:lang.ER} we would obtain that $U$ is a solution to the general system.~Hence,
we conclude that the general system does not have the greatest solution in ${\cal Q}(A)$.
\end{example}

The next theorem demonstrates one shortcoming of state reductions by means of fuzzy quasi-orders and fuzzy equivalences.~Namely, we show that for some fuzzy recognizers no reduction will result in its minimal automaton.

\begin{theorem}
There exists a fuzzy automaton $\cal A$ such that no reduction of $\cal A$ by means of fuzzy quasi-orders provide a minimal fuzzy recognizer.
\end{theorem}

\begin{proof}
Let $\cal L$ be the Boolean structure and ${\cal A}=(A,X,\delta^A,\sigma^A,\tau^A)$ a fuzzy recognizer over $\cal L$, where $|A|=4$, $X=\{x\}$, and $\delta^A$,$\sigma^A$, and $\tau^A $ are given by
\[
\delta_x^A=\begin{bmatrix} 1 & 0 & 0 & 0 \\ 0 & 0 & 0 & 1 \\ 0 & 0 & 0 & 0 \\ 0 & 0 & 0 & 0
\end{bmatrix}, \qquad \sigma^A=\begin{bmatrix} 0 & 1 & 0 & 0 \end{bmatrix}, \qquad
\tau^A=\begin{bmatrix} 0 \\ 0 \\ 1 \\ 1 \end{bmatrix}.
\]
It is easy to check that for each $u\in X^*$ the following is true:
\[
L({\cal A})(u)=\begin{cases}\ 0 & \text{if}\ u=e\ \text{or}\ u=x^n, \text{for}\
 n\ge 2 , \\
\ 1 & \text{if}\ u=x ,
\end{cases}
\]
(in fact, $\cal A$ is a nondeterministic recognizer and $L({\cal A})$ is an ordinary crisp language consisting only of the letter $x$). If ${\cal B}=(B,X,\delta^B,\sigma^B,\tau^B)$ is a fuzzy recognizer over $\cal L$ with $|B|=2$, and
\[
\delta_x^B=\begin{bmatrix} 0 & 1 \\ 0 & 0 \end{bmatrix}, \qquad \sigma^B=\begin{bmatrix} 1 & 0 \end{bmatrix}, \qquad
\tau^B=\begin{bmatrix} 0 \\ 1 \end{bmatrix},
\]
then $\cal B$ recognizes $L({\cal A})$, and it is a minimal fuzzy recognizer of $L({\cal A})$, since $L({\cal A})$ can not be recognized by a fuzzy recognizer with only one state.

Consider now an arbitrary fuzzy equivalence
\[
E=\begin{bmatrix}
1 & a_{12} & a_{13} & a_{14} \\
a_{12} & 1 & a_{23} & a_{24} \\
a_{13} & a_{23} & 1 & a_{34} \\
a_{14} & a_{24} & a_{34} & 1
\end{bmatrix}
\]
on $A$, and suppose that $E$ is a solution to the general system corresponding to the fuzzy automaton $\cal A$. We will show that $E$ can not reduce $\cal A$ to a fuzzy recognizer with two states.

First, by $\sigma^A\circ E\circ \tau^A=a_{23}\lor a_{24}$ and $\sigma^A\circ E\circ \tau^A=\sigma^A\circ \tau^A=L({\cal A})(e)=0$ it follows $a_{23}=a_{24}=0$. Next, reflexivity and transitivity of $E$ yield $E\circ E=E$, what implies
\begin{eqnarray}
a_{12}\land a_{13}=0 ,& & a_{12}=0\ \ \text{or}\ \ a_{13}=0 \label{eq:coef1}\\
a_{12}\land a_{14}=0 ,& & a_{12}=0\ \ \text{or}\ \ a_{14}=0 \label{eq:coef2}\\
a_{13}\lor (a_{14}\land a_{34})=a_{13} ,&\qquad\text{i.e.,}\qquad & a_{13}=0\ \ \text{implies}\ \ a_{14}=0\ \ \text{or}\ \ a_{34}=0, \label{eq:coef3}\\
a_{14}\lor (a_{13}\land a_{34})=a_{14} ,& & a_{14}=0\ \ \text{implies}\ \ a_{13}=0\ \ \text{or}\ \ a_{34}=0, \label{eq:coef4}\\
a_{34}\lor (a_{13}\land a_{14})=a_{34} ,& & a_{34}=0\ \ \text{implies}\ \ a_{13}=0\ \ \text{or}\ \ a_{14}=0. \label{eq:coef5}
\end{eqnarray}
If $a_{12}=1$, then by (\ref{eq:coef1}) and (\ref{eq:coef2}) we obtain $a_{13}=a_{14}=0$, and hence
\[
E=\begin{bmatrix}
1 & 1 & 0 & 0 \\
1 & 1 & 0 & 0 \\
0 & 0 & 1 & 0 \\
0 & 0 & 0 & 1
\end{bmatrix} \qquad\text{or}\qquad
E=\begin{bmatrix}
1 & 1 & 0 & 0 \\
1 & 1 & 0 & 0 \\
0 & 0 & 1 & 1 \\
0 & 0 & 1 & 1
\end{bmatrix}.
\]
However, none of these two matrices is a solution to the general system.~Therefore, we conclude that $a_{12}=0$. According to
(\ref{eq:coef3}), (\ref{eq:coef4}) and (\ref{eq:coef5}), we distinguish the following five cases
\begin{eqnarray*}
&&a_{13}=a_{14}=a_{34}=0 ,\\
&&a_{13}=a_{14}=0, \quad a_{34}=1 ,\\
&&a_{13}=a_{34}=0, \quad a_{14}=1 ,\\
&&a_{14}=a_{34}=0, \quad a_{13}=1 ,\\
&&a_{13}=a_{14}=a_{34}=1 ,
\end{eqnarray*}
and we obtain that $E$ has one of the following forms
\begin{equation}\label{eq:matrices}
E=\begin{bmatrix}
1 & 0 & 0 & 0 \\
0 & 1 & 0 & 0 \\
0 & 0 & 1 & 0 \\
0 & 0 & 0 & 1
\end{bmatrix}, \
E=\begin{bmatrix}
1 & 0 & 0 & 0 \\
0 & 1 & 0 & 0 \\
0 & 0 & 1 & 1 \\
0 & 0 & 1 & 1
\end{bmatrix}, \
E=\begin{bmatrix}
1 & 0 & 0 & 1 \\
0 & 1 & 0 & 0 \\
0 & 0 & 1 & 0 \\
1 & 0 & 0 & 1
\end{bmatrix}, \
E=\begin{bmatrix}
1 & 0 & 1 & 0 \\
0 & 1 & 0 & 0 \\
1 & 0 & 1 & 0 \\
0 & 0 & 0 & 1
\end{bmatrix}, \
E=\begin{bmatrix}
1 & 0 & 1 & 1 \\
0 & 1 & 0 & 0 \\
1 & 0 & 1 & 1 \\
1 & 0 & 1 & 1
\end{bmatrix}.
\end{equation}
In the first case, $E$ is the equality relation, and it does not provide any reduction of $\cal A$, and in the second and fourth case, it can be easily verified that $E$ is a solution to the general system, but it reduces $\cal A$ to a fuzzy recognizer with three states.~Finally, in the third and fifth case, $E$ is not a solution to the general system, since
\[
\sigma^A\circ E\circ \delta_x^A\circ E\circ \delta_x^A\circ E\circ \tau_x^A = 1 \ne 0 =
\sigma^A\circ \delta_x^A\circ \delta_x^A\circ \tau_x^A .
\]
Therefore, any state reduction of $\cal A$ by means of fuzzy equivalences does not provide fuzzy recognizer with less than three states. According to (b) of Theorem \ref{th:ER}, the same conclusion also holds for fuzzy quasi-orders. This completes the proof of the theorem.
\end{proof}

\section{Right and left invariant fuzzy quasi-orders}\label{sec:RLIFQO}

As in \cite{CSIP.07,CSIP.09}, where similar questions concerning
fuzzy equivalences have been considered, here we~study the following two
instances of the general system.~Let ${\cal A}=(A,X,\delta^A)$ be a fuzzy automaton.~If a fuzzy quasi-order $R$ on $A$ is a solution
to  a system
\begin{equation}\label{eq:RIFQO}
R\circ \delta_{x}^{A}\circ R = \delta_{x}^{A}\circ R, \ \ \text{for every}\ x\in X,
\end{equation}
then it will be called a {\it right invariant fuzzy quasi-order\/} on $\cal A$, and if it is a solution
to a system
\begin{equation}\label{eq:LIFQO}
R\circ \delta_{x}^A\circ R = R\circ \delta_{x}^A, \ \ \text{for every}\ x\in X,
\end{equation}
then it will be called a {\it left invariant fuzzy quasi-order\/} on $\cal A$.~A crisp quasi-order on $A$ which is a
solution to (\ref{eq:RIFQO}) is called a {\it right invariant quasi-order\/} on $\cal A$, and a crisp quasi-order
which is a solution to (\ref{eq:LIFQO}) is called a {\it left invariant quasi-order\/} on $\cal A$.~Let us note that a fuzzy quasi-order on $A$ is both right and left invariant if and only if it is a solution to system
\begin{equation}\label{eq:IFQO}
R\circ \delta_x^{A} = \delta_x^{A}\circ R, \ \ \text{for every}\ x\in X,
\end{equation}
and then it is called an {\it  invariant fuzzy quasi-order\/}.

If ${\cal A}=(A,X,\delta^A,\sigma^A,\tau^A)$ is a fuzzy recognizer, then
by a {\it right invariant fuzzy quasi-order\/} on $\cal A$~we~mean a fuzzy quasi-order
$R$ on $A$ which is a solution to (\ref{eq:RIFQO})
and

\begin{equation}\label{eq:RIFQO.2}
R\circ \tau^A  = \tau^A ,
\end{equation}
and a {\it left invariant fuzzy quasi-order\/} on $\cal A$ is a fuzzy quasi-order
$R$ on $A$ which is a solution to (\ref{eq:LIFQO})
and
\begin{equation}\label{eq:LIFQO.2}
\sigma^A\circ R  = \sigma^A .
\end{equation}
It is clear that all right and left invariant fuzzy quasi-orders
on a fuzzy recognizer $\cal A$ are solutions of the general system (\ref{eq:gen.syst}), and hence, the corresponding afterset fuzzy automata are equivalent to $\cal
A$.

In other words, right (resp.~left) invariant fuzzy quasi-orders on the fuzzy recognizer $\cal A$ are exactly~those right (resp.~left) invariant fuzzy quasi-orders on the fuzzy automaton $(A,X,\delta^A)$ which are solutions to the\break fuzzy relation equation (\ref{eq:RIFQO.2}) (resp.~(\ref{eq:LIFQO.2})).~It is well-known
(see \cite{CIB.09,Perf.04,PG.03,PN.07,San.76}) that solutions to (\ref{eq:RIFQO.2}) (resp.~(\ref{eq:LIFQO.2})) in ${\cal Q}(A)$ form a principal ideal of ${\cal Q}(A)$
whose greatest element is a fuzzy quasi-order $R^\tau $ (resp.~$R_\sigma$) defined by (\ref{eq:Rf}) (here we write $\tau^A=\tau $ and $\sigma^A=\sigma $).~This means that right (resp.~left) invariant fuzzy quasi-orders on the fuzzy recognizer $\cal A$ are those right (resp.~left) invariant fuzzy
quasi-orders on the fuzzy automaton $(A,X,\delta^A)$ which are contained in $R^\tau
$ (resp.~$R_\sigma $).

Let us note that fuzzy equivalences satisfying (\ref{eq:RIFQO})
and (\ref{eq:LIFQO}) have been studied in
\cite{CSIP.07,CSIP.09}.~They are~respect\-ively called right and
left invariant fuzzy equivalences.~Right and left invariant
quasi-orders have~been used for the state reduction of
non-deterministic automata by~Champarnaud and Coulon
\cite{CC.03,CC.04}, Ilie, Navarro and Yu \cite{INY.04}, and Ilie,
Solis-Oba and Yu \cite{ISY.05} (see also \cite{IY.02,IY.03}).

By the following theorem we give a characterization of right invariant fuzzy quasi-orders:

\begin{theorem} \label{th:rifqo}
Let ${\cal A}=(A,X,\delta^A)$ be a fuzzy automaton and $R$ a fuzzy quasi-order on $A$.~Then the~follow\-ing conditions are equivalent:
\begin{itemize}\parskip=0pt
\item[{\rm (i)}] $R$ is a right invariant fuzzy quasi-order;
\item[{\rm (ii)}] $R\circ \delta_x^A\le \delta_x^A\circ R$, for every $x\in X$;
\item[{\rm (iii)}] for all $a,b\in A$ we have
\begin{equation}\label{eq:rifqo1}
R(a,b)\le \bigwedge_{x\in X}\bigwedge_{c\in A}(\delta_x^A\circ R)(b,c)\to (\delta_x^A\circ R)(a,c).
\end{equation}
\end{itemize}
\end{theorem}

\begin{proof}
(i)$\iff $(ii). Consider an arbitrary $x\in X$. If $R\circ \delta_{x}^A\circ R = \delta_{x}^A\circ R$, then by reflexivity of $R$ it follows
\[
R\circ \delta_x^A\le R\circ \delta_{x}^A\circ R =
\delta_x^A\circ R .
\]
Conversely, if $R\circ \delta_x^A\le \delta_x^A\circ R$ then
$R\circ \delta_x^{A}\circ R\le \delta_x^A\circ R\circ R = \delta_x^A\circ R$, and since the opposite
inequality follows by reflexivity of $R$, we conclude that $R\circ \delta_{x}^A\circ R = \delta_{x}^A\circ R$.

(i)$\implies $(iii). Let $R$ be a right invariant fuzzy equivalence.~Then for all
$x\in X$ and $a,b,c\in A$~we have that
\[
R(a,b)\otimes (\delta_x^A\circ R)(b,c) \le (R\circ\delta_x^A\circ R)(a,c) = (\delta_x^A\circ R)(a,c),
\]
and by the adjunction property we obtain that $R(a,b)\le (\delta_x^A\circ R)(b,c) \to (\delta_x^A\circ R)(a,c)$.
Hence,
\begin{equation}\label{eq:rife2}
R(a,b) \le (\delta_x^A\circ R)(b,c) \to (\delta_x^{A}\circ R)(a,c).
\end{equation}
Since (\ref{eq:rife2}) is satisfied for all $c\in A$ and $x\in X$, we conclude that (\ref{eq:rifqo1}) holds.

(iii)$\implies $(i). If (iii) holds, then for arbitrary $x\in X$ and $a,b,c\in A$ we have that
\[
R(a,b) \le (\delta_x^A\circ R)(b,c) \to (\delta_x^A\circ R)(a,c),
\]
and by the adjunction~prop\-erty we obtain that $R(a,b)\otimes (\delta_x^A\circ R)(b,c)\le (\delta_x^A\circ R)(a,c)$. Now,
\[
(R\circ\delta_x^A\circ R)(a,c)= \bigvee_{b\in A} R(a,b)\otimes (\delta_x^A\circ R)(b,c)\le (\delta_x^A\circ R)(a,c),
\]
whence $R\circ\delta_x^A\circ R\le \delta_x^A\circ R$, and since the opposite inequality follows immediately by reflexivity of $R$,~we conclude that $R\circ\delta_x^A\circ R= \delta_x^A\circ R$, for every $x\in X$, i.e., $R$ is a right invariant fuzzy quasi-order.
\end{proof}

Let ${\cal A}=(A,X,\delta^A)$ be a fuzzy automaton and $R$ a fuzzy quasi-order on~$A$.~Let
us define a fuzzy~relation $R^r$ on $A$ by
\begin{equation}\label{eq:R.r}
R^r(a,b) = \bigwedge_{x\in X}\bigwedge_{c\in A}(\delta_x^{A}\circ R)(b,c)\to
(\delta_x^{A}\circ R)(a,c), \qquad
\end{equation}
for all $a,b\in A$. Since $R^r$ is an intersection of a family of fuzzy quasi-orders defined as in (\ref{eq:Rf}), we~have~that $R^r$ is also a fuzzy quasi-order.~According
to Theorem \ref{th:rifqo}, $R$ is a right invariant fuzzy quasi-order if and only if $R\le R^r$.

Moreover, we have the following:

\begin{lemma}\label{le:R.r}
Let ${\cal A}=(A,X,\delta^{A})$ be a fuzzy automaton, and let $R$ and $S$ be fuzzy quasi-orders on
$A$.

If $R\le S$, then $R^r\le S^r$.
\end{lemma}

\begin{proof}
Consider arbitrary $a,b\in A$ and $x\in X$. By $R\le S$ it follows $R\circ S=S$, and
by (\ref{eq:res.mult}), for arbitrary $c,d\in A$ we have that
\[
(\delta_x^{A}\circ R)(b,c)\to (\delta_x^{A}\circ R)(a,c) \le
(\delta_x\circ R)(b,c)\otimes S(c,d)\to (\delta_x^{A}\circ R)(a,c)\otimes S(c,d).
\]
Now, by (\ref{eq:res.inf.sup}) we obtain that
\[
\begin{aligned}
R^r(a,b)&\le \bigwedge_{c\in A}(\delta_x^{A}\circ R)(b,c)\to (\delta_x^A\circ R)(a,c) \\
&\le \bigwedge_{c\in A}\Bigl[(\delta_x^{A}\circ R)(b,c)\otimes S(c,d)\to
(\delta_x^{A}\circ R)(b,c)\otimes S(c,d)\Bigr] \\
&\le \Bigl[\bigvee_{c\in A}(\delta_x^{A}\circ R)(b,c)\otimes S(c,d)\Bigr]\to
\Bigl[\bigvee_{c\in A}(\delta_x^{A}\circ R)(a,c)\otimes S(c,d)\Bigr] \\
&= (\delta_x^{A}\circ R\circ S)(b,d)\to (\delta_x^{A}\circ R\circ S)(a,d) =
(\delta_x^{A}\circ S)(b,d)\to (\delta_x^{A}\circ S)(a,d) .
\end{aligned}
\]
Since this holds for all $x\in X$ and $d\in A$, we conclude that
\[
R^r(a,b)\le \bigwedge_{x\in X}\bigwedge_{d\in A}(\delta_x^{A}\circ S)(b,d)\to (\delta_x^{A}\circ S)(a,d)=S^r(a,b),
\]
and hence, $R^r\le S^r$.
\end{proof}

Now we prove the following:
\begin{theorem} \label{th:lrifqo}
Let ${\cal A}=(A,X,\delta^A)$ be a fuzzy automaton and let ${\cal A}'=(A,X,\delta^A,\sigma^A,\tau^A)$ be a fuzzy~recog\-nizer belonging to $\cal A$.
Then
\begin{itemize}\parskip0pt
\item[{\rm (a)}] The set ${\cal Q}^{\mathrm{ri}}({\cal A})$ of all right invariant fuzzy quasi-orders on $\cal A$ forms a complete lattice, which is a complete join-subsemilattice of the lattice ${\cal Q}(A)$ of all fuzzy quasi-orders on $A$.
\item[{\rm (b)}] The set ${\cal Q}^{\mathrm{cri}}({\cal A})$ of all right invariant crisp quasi-orders on $\cal A$ forms a complete lattice, which is a complete join-subsemilattice of the lattice ${\cal Q}^{\mathrm{ri}}(A)$.
\item[{\rm (c)}] The set ${\cal Q}^{\mathrm{ri}}({\cal A}')$ of all right invariant fuzzy quasi-orders on ${\cal A}'$ is a principal ideal of
the lattice ${\cal Q}^{\mathrm{ri}}({\cal A})$.
\end{itemize}
\end{theorem}

\begin{proof}
(a) Let $\{R_i\}_{i\in I}\subseteq {\cal Q}^{\mathrm{ri}}({\cal A})$, and let $R$ be the join of this family in ${\cal Q}(A)$.~Then for each $i\in
I$, by $R_i\le R$ and Lemma \ref{le:R.r} we obtain that $R_i\le R_i^r\le
R^r$, whence $R\le R^r$. Now, by Theorem \ref{th:rifqo} it follows that $R\in
{\cal Q}^{\mathrm{ri}}({\cal A})$, and hence, ${\cal Q}^{\mathrm{ri}}({\cal A})$ is a complete join-subsemilattice of ${\cal Q}({\cal A})$.~Since ${\cal Q}^{\mathrm{ri}}({\cal A})$ contains the
least element of ${\cal Q}(A)$, the equality relation on~$A$, we conclude
that ${\cal Q}^{\mathrm{ri}}({\cal A})$ is a complete lattice.

(b) This follows immediately by (a) and (\ref{eq:fqo.join}), since union and composition of fuzzy relations,
applied~to crisp relations, as results give crisp relations.

(c) By definition, ${\cal Q}^{\mathrm{ri}}({\cal A}')$ consists of all $R\in {\cal
Q}^{\mathrm{ri}}({\cal A})$ which satisfy $R\circ \tau = \tau $.~It~is~well-known that $R\circ \tau
= \tau $ is equivalent to $R\le R_\tau $, what implies that ${\cal Q}^{\mathrm{ri}}({\cal A}')$ is
an ideal of ${\cal Q}^{\mathrm{ri}}({\cal A})$. Next, let $\{R_i\}_{i\in I}$ be an arbitrary family
of elements of ${\cal Q}^{\mathrm{ri}}({\cal A}')$ and let $R$ be the join of this family in ${\cal
Q}^{\mathrm{ri}}({\cal A})$.~According~to~(a) of this theorem, $R$ is also the join of the family
$\{R_i\}_{i\in I}$ in ${\cal Q}({\cal A})$, and since $R_i\le R_\tau $, for every $i\in I$, we
conclude that $R\le R_\tau $.~By this it follows that, ${\cal Q}^{\mathrm{ri}}({\cal A}')$ is a
complete join-subsemilattice of ${\cal Q}^{\mathrm{ri}}({\cal A})$, and hence, ${\cal
Q}^{\mathrm{ri}}({\cal A}')$ is an ideal of ${\cal Q}^{\mathrm{ri}}({\cal A})$ having the greatest
element, what means that ${\cal Q}^{\mathrm{ri}}({\cal A}')$ is a~principal ideal of ${\cal
Q}^{\mathrm{ri}}({\cal A})$.
\end{proof}

As we have noted before, the problem of computing the greatest right invariant fuzzy quasi-order~on~a fuzzy recognizer ${\cal A}=(A,X,\delta^A,\sigma^A,\tau^A)$
one reduces to the problem of
computing the greatest~right~invar\-iant fuzzy quasi-order~on a fuzzy automaton
$(A,X,\delta^A)$ contained in the fuzzy quasi-order $R^\tau $ ($\tau=\tau^A$). For
that reason,
in the sequel we consider the problem how to construct
the greatest right~invariant fuzzy quasi-order
$R^{\textrm{ri}}$ contained in a given fuzzy quasi-order $R$ on a fuzzy automaton.

\begin{theorem}\label{th:alg}
Let ${\cal A}=(A,X,\delta^{A})$ be a fuzzy automaton, let $R$ be a fuzzy quasi-order on $\cal A$ and let $R^{\mathrm{ri}}$ be the greatest right invariant fuzzy quasi-order
on $\cal A$ contained in $R$.

Define inductively a sequence
$\{R_k\}_{k\in\Bbb N}$ of fuzzy quasi-orders on $\cal A$ as follows:
\begin{equation}\label{eq:Rk}
R_1=R, \ \ R_{k+1}=R_k\land R_k^r, \ \ \text{for each}\ k\in \Bbb N .
\end{equation}
Then
\begin{itemize}\parskip=0pt
\item[{\rm (a)}] $R^{\mathrm{ri}} \le \cdots \le R_{k+1}\le R_k\le \cdots \le R_1=R $;
\item[{\rm (b)}] If $R_k=R_{k+m}$, for some $k,m\in \Bbb N$, then $R_k=R_{k+1}=R^{\mathrm{ri}}$;
\item[{\rm (c)}] If $\cal A$ is finite and $\cal L$ is locally finite,
then $R_k=R^{\mathrm{ri}}$ for some $k\in \Bbb N$.
\end{itemize}
\end{theorem}

\begin{proof}
(a) Clearly, $R_{k+1}\le R_k$, for each $k\in \Bbb N$, and $R^{\mathrm{ri}}\le R_1$.~Suppose~that
$R^{\mathrm{ri}}\le R_k$, for some~$k\in \Bbb N$. Then $R^{\mathrm{ri}} \le (R^{\mathrm{ri}})^r\le R_k^r$,
so $R^{\mathrm{ri}} \le R_k\land R_k^r=R_{k+1}$.~Therefore, by
induction~we~obtain that  $R^{\mathrm{ri}} \le R_k$, for~every~$k\in \Bbb N$.

(b) Let $R_k=R_{k+m}$, for some $k,m\in \Bbb N$. Then $R_k=R_{k+m}\le R_{k+1}=R_k\land R_k^r\le R_k^r$, what means~that $R_k$ is a right invariant fuzzy quasi-order.~Since $R^{\mathrm{ri}}$ is the greatest right invariant
fuzzy~quasi-order~con\-tained in $R$, we conclude that $R_k=R_{k+1}=R^{\mathrm{ri}}$.

(c) Let $\cal A$ be a finite fuzzy automaton and $\cal L$ be a locally finite algebra. Let
the carrier of a subalgebra of $\cal L$ generated by the set
$\delta^{A}(A\times X\times A)\cup R(A\times A)$ be denoted by $L_{\cal A}$. This generating set is finite, so
$L_{\cal A}$ is also finite, and hence, the set $L_{\cal A}^{A\times A}$ of all fuzzy
relations on $A$ with values in $L_{\cal A}$ is finite.~By definitions of fuzzy relations
$R_k$ and $R_k^r$ we have that $R_k\in L_{\cal A}^{A\times A}$, what implies that there
exist $k,n\in \Bbb N$ such that $R_k=R_{k+m}$, and by (b) we conclude that $R_k=R^{\mathrm{ri}}$.
\end{proof}

According to (c) of Theorem \ref{th:alg}, if the structure $\cal L$ is locally
finite, then for every fuzzy automaton~$\cal A$ over $\cal L$ we have that  every sequence of fuzzy quasi-orders defined by (\ref{eq:Rk}) is finite.~However, this does~not necessary hold if $\cal L$ is not locally finite, as the following
example shows:

\begin{example}\label{ex:prod}\rm
Let $\cal L$ be the Goguen (product) structure and ${\cal A}=(A,X,\delta^A)$ a fuzzy
automaton over~$\cal L$, where $A=\{1,2\}$, $X=\{x\}$, and $\delta_x^A$ is given by
\[
\delta_x^A=\begin{bmatrix}
0.2 &  0  \\
 0  & 0.1
\end{bmatrix},
\]
and let $R$ be the universal relation on $A$.~Applying to $R$ the procedure from Theorem \ref{th:alg}, we obtain~a~sequence $\{R_k\}_{k\in \Bbb N}$ of fuzzy quasi-orders given by
\[
R_k=\begin{bmatrix} 1 & 1\ \\ \textstyle\frac1{2^{k-1}} & 1\  \end{bmatrix},\ \
k\in \Bbb N ,
\]
whose all members are different, i.e., this sequence is infinite.~We also have that the greatest right~invar\-iant fuzzy quasi-order contained in $R$
is given by
\[
R^{\mathrm{ri}}=\begin{bmatrix} \,1\ &\ 1\, \\ \,0\ &\ 1\, \end{bmatrix}.
\]
\end{example}

For a fuzzy  automaton ${\cal A}=(A,X,\delta^A)$ over a complete residuated lattice $\cal L$,
the~greatest left invar\-iant fuzzy quasi-order $R^{\mathrm{li}}$ contained in a given fuzzy
quasi-order $R$ on $A$ can be computed in a similar way as~$R^{\mathrm{ri}}$. Indeed,
inductively~we define~a sequence $\{R_k\}_{k\in \Bbb N}$ of fuzzy quasi-orders on $A$ by
\begin{equation}\label{eq:Lk}
R_1=R, \ \ R_{k+1}=R_k\land R_k^{l}, \ \ \text{for each}\ k\in \Bbb N ,
\end{equation}
where $R_k^{l}$ is a fuzzy quasi-order on $A$ defined by
\[
R_k^{l}(a,b) = \bigwedge_{x\in X}\bigwedge_{c\in A}(R_k\circ \delta_x^{A})(c,a)\to
(R_k\circ \delta_x^{A})(c,b), \qquad \text{for all $a,b\in A$}.
\]
If $\cal L$ is locally finite,
then this sequence is necessary finite and $R^{\mathrm{li}}$ equals~the least element of this~sequence.

It is worth~noting~that
the greatest right and left invariant fuzzy quasi-orders are calculated using iterative procedures, but these calculations are not approximative.~Whenever these procedures
terminate in a finite number of steps, exact solutions to the considered systems of fuzzy relation equations are obtained.

Note also that for a fuzzy automaton ${\cal A}=(A,X,\delta^A)$ over a complete residuated lattice $\cal
L$, in \cite{CSIP.07,CSIP.09}~we gave a procedure~for~compu\-ting the greatest right invariant fuzzy equivalence
$E^{\mathrm{rie}}$   contained in a given fuzzy
equivalence $E$ on $A$. This procedure is similar to the proce\-dure given in Theorem~\ref{th:alg}~for~fuzzy~quasi-orders, and it also~works~for all fuzzy finite automata over a locally
finite complete residuated lattice. Namely,
inductively~we define~a sequence $\{E_k\}_{k\in \Bbb N}$ of fuzzy equivalences on $A$ by
\begin{equation}\label{eq:Ek}
E_1=E, \ \ E_{k+1}=E_k\land E_k^{req}, \ \ \text{for each}\ k\in \Bbb N ,
\end{equation}
where $E_k^{req}$ is a fuzzy equivalence defined by
\[
E_k^{req}(a,b) = \bigwedge_{x\in X}\bigwedge_{c\in A}(\delta_x^{A}\circ E_k)(a,c)\lra
(\delta_x^{A}\circ E_k)(b,c), \qquad \text{for all $a,b\in A$}.
\]
It was proved in \cite{CSIP.07,CSIP.09} that if $\cal L$ is locally finite,
then this sequence is necessary finite and $E^{\mathrm{rie}}$ equals~the least element of this sequence.~

By the next example we show that it is possible that the sequence of fuzzy equivalences defined by~(\ref{eq:Ek}) is infinite, but the sequence of fuzzy quasi-orders defined by (\ref{eq:Rk})
is finite.

\begin{example}\label{ex:prod2}\rm
Let $\cal L$ be the Goguen (product) structure and   ${\cal A}=(A,X,\delta^A)$  a fuzzy
automaton over~$\cal L$, where $A=\{1,2,3\}$, $X=\{x\}$, and $\delta_x^A$ is given by
\[
\delta_x^A=
\begin{bmatrix}
\,0\ &\ 1\, & \ 1\, \\
\,0\ &\ 1\, & \ 1\, \\
\,\frac12\ &\ 0\, & \ 0\,
\end{bmatrix}.
\]
If we start from the universal relation on $A$, applying the rule (\ref{eq:Ek}) we obtain
an infinite sequence~$\{E_k\}_{k\in \Bbb N}$ of fuzzy equivalences
on $A$, where
\[
E_k=\begin{bmatrix}
1 & 1 & \frac1{2^{k-1}} \\
1 & 1 & \frac1{2^{k-1}} \\
\textstyle\frac1{2^{k-1}} & \textstyle\frac1{2^{k-1}} & 1
\end{bmatrix},\ \ k\in \Bbb N .
\]
On the other hand, if we also start from the universal relation, the rule (\ref{eq:Rk}) gives a finite sequence~$\{R_k\}_{k\in \Bbb N}$ of fuzzy quasi-orders on $A$, where
\[
R_1=\begin{bmatrix}
1 & 1 & 1 \\
1 & 1 & 1 \\
1 & 1 & 1
\end{bmatrix}, \ \ \
R_2=\begin{bmatrix}
1 & 1 & 1 \\
1 & 1 & 1 \\
\frac 12 & \frac 12 &1
\end{bmatrix}, \ \ \
R_k=R_2, \ \text{for each $k\in \Bbb N$, $k\ge 3$}.
\]

\end{example}

Reduction of fuzzy automata by means of right and left invariant fuzzy equivalences
has been studied in \cite{CSIP.07,CSIP.09}.~Since the set of all right invariant fuzzy equivalences is a subset of the set of all right invariant fuzzy quasi-orders, the greatest element of this subset (the greatest right invariant fuzzy equivalence) is~less or equal than the greatest element of the whole set (the greatest right invariant fuzzy quasi-order).~The next example shows that this inequality can be strict.~Thus, reduction of a fuzzy automaton by using the greatest right invariant fuzzy quasi-order gives better results than its reduction by using the greatest right invariant fuzzy equivalence, according to Remark \ref{re:afrcardinality}.

Furthermore, as we have shown by Theorem \ref{th:lang.ER}, if a fuzzy quasi-order $R$ on a fuzzy automaton $\cal A$ is a solution to the general system, then its natural fuzzy equivalence $E_R$ is also a solution to the general system.~But, the next example also shows that
if $R$ is a right invariant fuzzy quasi-order, then $E_R$ is not necessary a right
invariant fuzzy equivalence.

\begin{example}\label{ex:R.ri}\rm
Let $\cal L$ be the Boolean structure, and let ${\cal A}=(A,X,\delta^A)$ be a fuzzy
automaton over $\cal L$, where $A=\{1,2,3\}$, $X=\{x,y\}$, and fuzzy transition relations
$\delta_x^A$ and $\delta_y^A$
are given by
\[
\delta_x^A=\begin{bmatrix}
1 & 0 & 0  \\
0 &  0  & 0 \\
0 & 0 & 0
\end{bmatrix}, \ \ \ \
\delta_y^A=\begin{bmatrix}
1 & 0 & 0  \\
1 &  1  & 0 \\
1 & 0 & 0
\end{bmatrix}.
\]
The greatest right invariant
fuzzy quasi order $R^{\mathrm{ri}}$~on~$\cal A$, its natural fuzzy equivalence $E_{R^{\mathrm{ri}}}$,
and the greatest right invariant fuzzy equivalence $E^{\mathrm{ri}}$ on $\cal A$
are given~by
\[
R^{\mathrm{ri}}=\begin{bmatrix}
1 & 1 & 1  \\
0 & 1 & 1 \\
0 & 1 & 1
\end{bmatrix},\ \ \
E_{R^{\mathrm{ri}}}=\begin{bmatrix}
1 & 0 & 0  \\
0 & 1 & 1 \\
0 & 1 & 1
\end{bmatrix},\ \ \
E^{\mathrm{ri}}=\begin{bmatrix}
1 & 0 & 0  \\
0 & 1 & 0 \\
0 & 0 & 1
\end{bmatrix}.
\]
Thus, $E^{\mathrm{ri}}$ do not reduce the number of states of $\cal A$, but
$R^{\mathrm{ri}}$ reduces $\cal A$ to a fuzzy automaton with~two~states.

Moreover, $R^{\mathrm{ri}}$ is a right invariant fuzzy quasi-order, but its
natural~fuzzy~equiva\-lence $E_{R^{\mathrm{ri}}}$ is not a right invariant fuzzy equivalence,
because $E^{\mathrm{ri}}<E_{R^{\mathrm{ri}}}$.~We also have that the afterset fuzzy automaton
${\cal A}/R^{\mathrm{ri}}$ is not isomorphic to the factor fuzzy automaton ${\cal
A}/E_{R^{\mathrm{ri}}}$, since
\[
R^{\mathrm{ri}}\circ \delta_y^A\circ R^{\mathrm{ri}}=\begin{bmatrix}
1 & 1 & 1  \\
1 & 1 & 1 \\
1 & 1 & 1
\end{bmatrix}, \ \ \ \
E_{R^{\mathrm{ri}}}\circ \delta_y^A\circ E_{R^{\mathrm{ri}}}=\begin{bmatrix}
1 & 0 & 0  \\
1 &  1  & 1 \\
1 & 1 & 1
\end{bmatrix}.
\]
\end{example}

Next we consider the case when $\cal L$ is a complete residuated lattice satisfying
the following conditions:
\begin{eqnarray}\label{eq:infd}
x\lor \bigl(\bigwedge_{i\in I}y_i\bigr) = \bigwedge_{i\in I}(x\lor y_i) ,\\
\label{eq:infdm}
x\otimes \bigl(\bigwedge_{i\in I}y_i\bigr) = \bigwedge_{i\in I}(x\otimes y_i) ,
\end{eqnarray}
for all $x\in L$ and $\{y_i\}_{i\in I}\subseteq L$. Let us note that if ${\cal L}=([0,1],\land ,\lor , \otimes ,\to , 0, 1)$, where $[0,1]$ is the~real~unit interval
and $\otimes $ is a left-con\-tin\-uous t-norm on $[0,1]$, then (\ref{eq:infd}) follows
immediately by linearity of $\cal L$, and $\cal L$ satisfies (\ref{eq:infdm}) if and only if $\otimes $ is
a continuous t-norm, i.e., if and only if $\cal L$ is a $BL$-algebra~(cf.~\cite{Bel.02,BV.05}).
 Therefore, conditions {\rm (\ref{eq:infd})} and {\rm (\ref{eq:infdm})} hold for~every $BL$-algebra on the real unit interval.~In particular,
the {\L}ukasiewicz, Goguen (product)~and~G\"odel structures fulfill {\rm (\ref{eq:infd})} and {\rm (\ref{eq:infdm})}.

We have that the following is true:

\begin{theorem}\label{th:inf}
Let $\cal L$ be a complete residuated lattice satisfying {\rm (\ref{eq:infd})} and {\rm
(\ref{eq:infdm})}, let ${\cal A}=(A,X,\delta^A)$ be a fuzzy finite automaton over $\cal L$, let $R$
be a fuzzy quasi-order on $A$, let $R^{\,\mathrm{ri}}$ be the greatest right invariant fuzzy
quasi-order on $\cal A$ contained in $R$, and let $\{R_k\}_{k\in\Bbb N}$ be the sequence of fuzzy
quasi-orders on $A$~defined by {\rm (\ref{eq:Rk})}.~Then
\begin{equation}\label{eq:inf.Ek}
R^{\mathrm{ri}} = \bigwedge_{k\in \Bbb N}R_k .
\end{equation}
\end{theorem}

\begin{proof}
It was proved in \cite{CSIP.09} that if {\rm (\ref{eq:infd})} holds, then for all non-increasing sequences
$\{x_k\}_{k\in \Bbb N}, \{y_k\}_{k\in \Bbb N}\subseteq L$ we have
\begin{equation}\label{eq:infd2}
\bigwedge_{k\in \Bbb N}(x_k\lor y_k) = \bigl(\bigwedge_{k\in \Bbb N}x_k\bigr)\lor
\bigl(\bigwedge_{k\in \Bbb N}y_k\bigr).
\end{equation}
For the sake of simplicity set
\[
S=\bigwedge_{k\in \Bbb N}R_k .
\]
Clearly, $S$ is a fuzzy quasi-order.~To prove (\ref{eq:inf.Ek}) it is enough to prove that
$S$ is a right~invariant fuzzy quasi-order on $\cal A$. First, we have that
\begin{equation}\label{eq:Fab}
S(a,b)\le R_{k+1}(a,b)\le R_k^r(a,b) \le \delta_x^{A}\circ R_k(b,c)\to \delta_x^A\circ R_k(a,c),
\end{equation}
holds for all $a,b,c\in A$, $x\in X$ and $k\in \Bbb N$. Now, by (\ref{eq:Fab}) and (\ref{eq:res.inf.inf}) we obtain that
\begin{equation}\label{eq:Fab2}
S(a,b)\le \bigwedge_{k\in \Bbb N}\Bigl(\delta_x^{A}\circ R_k(b,c)\to \delta_x^{A}\circ R_k(a,c)\Bigr)
\le \bigwedge_{k\in \Bbb N}\bigl(\delta_x^{A}\circ R_k(b,c)\bigr)\to
\bigwedge_{k\in \Bbb N}\bigl(\delta_x^{A}\circ R_k(a,c)\bigr),
\end{equation}
for all $a,b,c\in A$ and $x\in X$. Next,
\begin{equation}\label{eq:infk}
\begin{aligned}
\bigwedge_{k\in \Bbb N}\bigl(\delta_x^{A}\circ R_k(b,c)\bigr) &
= \bigwedge_{k\in \Bbb N}\Bigl(\bigvee_{d\in A}\bigl(\delta_x^{A}(b,d)\otimes R_k(d,c)\bigr)\Bigr) &&\\
&= \bigvee_{d\in A}\Bigl(\bigwedge_{k\in \Bbb N}\bigl(\delta_x^{A}(b,d)\otimes R_k(d,c)\bigr)\Bigr)\hspace{2cm} &&
\text{(by (\ref{eq:infd2}))} \\
&= \bigvee_{d\in A}\Bigl(\delta_x^{A}(b,d)\otimes \bigl(\bigwedge_{k\in \Bbb N} R_k(d,c)\bigr)\Bigr)\hspace{2cm} &&
\text{(by (\ref{eq:infdm}))} \\
&= \bigvee_{d\in A}\Bigl(\delta_x^{A}(b,d)\otimes S(d,c)\Bigr) = (\delta_x^{A}\circ S)(b,c). &&
\end{aligned}
\end{equation}
Use of condition (\ref{eq:infd2}) is justified by the facts that $A$ is finite, and that
$\{R_k(d,c)\}_{k\in \Bbb N}$ is a non-increasing sequence, so
$\{\delta_x^{A}(b,d)\otimes R_k(d,c)\}_{k\in \Bbb N}$ is also a non-increasing sequence.
In the same way we prove that
\begin{equation}\label{eq:infk2}
\bigwedge_{k\in \Bbb N}\bigl(\delta_x^{A}\circ R_k(a,c)\bigr) = (\delta_x^{A}\circ S)(a,c).
\end{equation}
Therefore, by (\ref{eq:Fab2}), (\ref{eq:infk}) and (\ref{eq:infk2}) we obtain that
\[
S(a,b)\le (\delta_x^{A}\circ S)(b,c)\to (\delta_x^{A}\circ S)(a,c).
\]
Since this inequality holds for all $x\in X$ and $c\in A$, we have that
\[
S(a,b)\le \bigwedge_{x\in X}\bigwedge_{c\in A}(\delta_x^{A}\circ S)(b,c)\to (\delta_x^A\circ S)(a,c),
\]
and by (iii) of Theorem \ref{th:rifqo} we obtain that $S$ is a right invariant fuzzy quasi-order on $\cal A$.
\end{proof}

\section{Some special types of right and left invariant fuzzy quasi-orders}\label{sec:APPROX}

For a given fuzzy quasi-order $R$ on a fuzzy automaton $\cal A$, Theorem \ref{th:alg} gives a
procedure for computing $R^{\mathrm{ri}}$ in case when the complete residuated lattice $\cal L$ is
locally finite, and Theorem \ref{th:inf} characterizes~$R^{\mathrm{ri}}$ in case when $\cal L$
satisfies some additional distributivity conditions.~But, what to do if $\cal L$ do not satisfy
any~of these conditions?~In this case we could consider~some subset of ${\cal Q}^{\mathrm{ri}}({\cal A})$ whose greatest element
can be effectively computed when $\cal L$ is any complete residuated lattice.~Here we consider two
such subsets.~The first one is the set ${\cal Q}^{\mathrm{cri}}({\cal A})$ of all right invariant
crisp~quasi-orders on $\cal A$, and the second one is the set ${\cal Q}^{\mathrm{sri}}({\cal A})$
of strongly right invariant fuzzy quasi-orders, which~will be defined latter.

~Note that for a crisp~rela\-tion $\varrho $ and a fuzzy relation $R$ on a set $A$ we
have that $\varrho \le R$ if and only if $\varrho \subseteq \widehat R$, where $\widehat
R$ denotes the crisp part of $R$.~Let ${\cal A}=(A,X,\delta^{A})$ be a fuzzy automaton and  $R$  a fuzzy~quasi-order
on $A$.~It is  easy to verify~that the crisp part of the fuzzy quasi-order $R^r$ can
be represented as follows: for all $a,b\in A$ we have
\begin{equation}
(a,b)\in \widehat{R^r} \ \iff \ (\forall x\in X)(\forall c\in A)\ (\delta_x^A\circ
R)(b,c)\le (\delta_x^A\circ R)(a,c) .
\end{equation}
We have that $\widehat{R^r}$ is a quasi-order, since the crisp part of any fuzzy quasi-order is a quasi-order.

The following theorem gives a procedure for computing the greatest right invariant
crisp quasi-order on a fuzzy automaton contained in a given quasi-order.

\begin{theorem} \label{th:alg.c}
Let ${\cal A}=(A,X,\delta^{A})$ be a fuzzy automaton, let $\varrho $ be a quasi-order on $\cal A$ and
let $\varrho^{\mathrm{ri}}$~be~the~greatest right invariant quasi-order on $\cal A$ contained in $\varrho $.

Define inductively a sequence
$\{\varrho_k\}_{k\in\Bbb N}$ of quasi-orders on $\cal A$ as follows:
\[
\varrho_1=\varrho, \ \ \varrho_{k+1}=\varrho_k\cap \widehat{\varrho_k^r}, \ \
\text{for each}\ k\in \Bbb N .
\]
 Then
\begin{itemize}\parskip=1pt
\item[{\rm (a)}] $\varrho^{\mathrm{ri}} \subseteq \cdots \subseteq \varrho_{k+1}\subseteq \varrho_k\subseteq \cdots
\subseteq \varrho_1=\varrho $;
\item[{\rm (b)}] If $\varrho_k=\varrho_{k+m}$, for some $k,m\in \Bbb N$, then
$\varrho_k=\varrho_{k+1}=\varrho^{\mathrm{ri}}$;
\item[{\rm (c)}] If $\cal A$ is finite, then $\varrho_k=\varrho^{\mathrm{ri}}$ for some
$k\in \Bbb N$.
\end{itemize}
\end{theorem}

\begin{proof}
(a) Clearly, $\varrho_{k+1}\subseteq \varrho_k$, for every $k\in \Bbb N$, and $\varrho^{\mathrm{ri}}
\subseteq \varrho_1$. If $\varrho^{\mathrm{ri}} \subseteq \varrho_k$, for some $k\in \Bbb N$, then $(\varrho^{\mathrm{ri}})^r\le \varrho_k^r$, and also, $\varrho^{\mathrm{ri}} \le (\varrho^{\mathrm{ri}} )^r$, so we have that
\[
\varrho^{\mathrm{ri}} \subseteq \widehat{(\varrho^{\mathrm{ri}})^r}\subseteq \widehat{\varrho_k^r},
\]
and by this it follows that $\varrho^{\mathrm{ri}}\subseteq \varrho_{k+1}$. Hence, by induction we
obtain that $\varrho^{\mathrm{ri}}\subseteq \varrho_k$, for every $k\in \Bbb N$.

(b) If $\varrho_k=\varrho_{k+m}$, for some $k,m\in \Bbb N$, then
\[
\varrho_k=\varrho_{k+m}\subseteq  \varrho_{k+1}\subseteq
\widehat{\varrho_k^r}\le {\varrho_k^r},
\]
so we have that $\varrho_k$ is a right invariant  quasi-order on $\cal A$.~Therefore,
$\varrho_k = \varrho_{k+1}=\varrho^{\mathrm{ri}} $.

(c) If the set $A$ is finite, then the set of all crisp relations on $A$
is also finite, so there exist $k,m\in \Bbb N$~such that $\varrho_k=\varrho_{k+m}$, and then $\varrho_k=\varrho^{\mathrm{ri}} $.
\end{proof}

The previous theorem shows that the greatest right invariant crisp quasi-order can
be effectively compu\-ted for~any fuzzy finite automaton over an arbitrary complete residuated lattice, not necessary locally finite, and even for a fuzzy finite automaton over an arbitrary lattice-ordered monoid.~However, in cases when we are able to effectively compute the greatest right invariant fuzzy quasi-order, using it we can attain better reduction than using the greatest right invariant crisp quasi-order, as the next example shows. Namely, the greatest right invariant crisp quasi-order $\varrho^{\mathrm{ri}}$ is less or equal than the greatest right invariant fuzzy quasi-order $R^{\mathrm{ri}}$ and according to Remark \ref{re:afrcardinality} there holds $|{\cal A}/{R^{\mathrm{ri}}}|\le |\cal A/{\varrho^{\mathrm{ri}}}|$.

\begin{example}\label{ex:FQO-QO}\rm
Let $\cal L$ be the G\"odel structure, and let  ${\cal A}=(A,X,\delta^A)$ be a fuzzy
automaton over $\cal L$, where $A=\{1,2,3\}$, $X=\{x\}$, and $\delta_x^A$ is given by
\[
\delta_x^A=\begin{bmatrix}
0 & 0.1 & 0  \\
0.2 &  0  & 0 \\
0.1 & 0 & 0
\end{bmatrix}.
\]
Then the greatest right invariant fuzzy quasi-order $R^{\mathrm{ri}}$ and the greatest
right invariant crisp quasi-order~$\varrho^{\mathrm{ri}}$ on $\cal A$ are given by
\[
R^{\mathrm{ri}}=\begin{bmatrix}
1 & 0.1 & 1  \\
1 & 1 & 1 \\
1 & 0.1 & 1
\end{bmatrix},\ \ \
\varrho^{\mathrm{ri}}=\begin{bmatrix}
1 & 0 & 0  \\
0 & 1 & 1 \\
0 & 0 & 1
\end{bmatrix}.
\]
Hence, $\varrho^{\mathrm{ri}}$ do not reduce the number of states of $\cal A$, but
$R^{\mathrm{ri}}$ reduces $\cal A$ to a fuzzy automaton with two states.
\end{example}

Let ${\cal A}=(A,X,\delta^{A})$ be a fuzzy automaton.~If a fuzzy quasi-order $R$ on $A$  is a solution to system
\begin{equation}\label{eq:SRIFQO}
R\circ \delta_x^A = \delta_x^A, \ \ \text{for every}\ x\in X,
\end{equation}
then it is called a {\it strongly right invariant fuzzy quasi-order\/} on $\cal A$, and if it is a solution to system
\begin{equation}\label{eq:SLIFQO}
\delta_x^{A}\circ R = \delta_x^{A}, \ \ \text{for every}\ x\in X,
\end{equation}
then it is  a {\it strongly left invariant fuzzy quasi-order\/} on $\cal A$.~Clearly, every strongly right (resp.~left)~invariant fuzzy quasi-order
is right (resp.~left) invariant.~Let us note that a fuzzy quasi-order on $A$ is both strongly right and left invariant if and only if it is a solution to system
\begin{equation}\label{eq:SIFQO}
R\circ \delta_x^{A}\circ R = \delta_x^{A}, \ \ \text{for every}\ x\in X,
\end{equation}
and then it is called a {\it strongly invariant fuzzy quasi-order\/}.

If ${\cal A}=(A,X,\delta^{A},\sigma^A,\tau^A)$ is a fuzzy recognizer, then
by a {\it stongly right invariant fuzzy quasi-order\/}~on~$\cal A$ we mean
a fuzzy quasi-order on $A$ which is a solution to (\ref{eq:SRIFQO}) and
\begin{equation}\label{eq:SRIFQO.2}
R\circ \tau^A  = \tau^A ,
\end{equation}
and a {\it strongly left invariant fuzzy quasi-order\/} on $\cal A$ is a fuzzy quasi-order which is a solution to~(\ref{eq:SLIFQO})
and
\begin{equation}\label{eq:SLIFQO.2}
\sigma^A\circ R  = \sigma^A .
\end{equation}

In the further text we study strongly right invariant fuzzy quasi-orders.

\begin{theorem}\label{th:SRIFQO}
Let ${\cal A}=(A,X,\delta^A)$ be a fuzzy automaton and let
${\cal A}'=(A,X,\delta^A,\sigma^A,\tau^A)$ be a fuzzy~recog\-nizer belonging
to $\cal A$. Then
\begin{itemize}
\item[{\rm (a)}] The set ${\cal Q}^{\mathrm{sri}}({\cal A})$ of all strongly
right invariant fuzzy quasi-orders on $\cal A$ is a principal ideal of
the lattice ${\cal Q}(A)$.~The greatest element of this principal ideal is a fuzzy quasi-order $R^{\mathrm{sri}}$
defined by
\begin{equation}\label{eq:R.sri}
R^{\mathrm{sri}} (a,b) = \bigwedge_{x\in X}\bigwedge_{c\in A}\delta_x^A(b,c)\to
\delta_x^A(a,c), \ \ \ \text{for all $a,b\in A$} .
\end{equation}
\item[{\rm (b)}] The set ${\cal Q}^{\mathrm{sri}}({\cal A}')$ of all strongly
right invariant fuzzy quasi-orders on ${\cal A}'$ is the principal ideal of~the lattice ${\cal Q}(A)$.~The greatest element of this principal ideal is a fuzzy quasi-order $R_\tau\land R^{\mathrm{sri}}$.
\end{itemize}
\end{theorem}

\begin{proof}
(a) We have that $R^{\mathrm{sri}}$ is a fuzzy quasi-order, as an intersection
of a family of fuzzy quasi-orders~defined as in (\ref{eq:Rf}).~Let $R$ be
an arbitrary fuzzy quasi-order on $A$.~Then we have that
\[
\begin{aligned}
R\le R^{\mathrm{sri}} \ &\iff\ \  (\forall x\in X)(\forall a,b,c\in
A)\  R(a,b)\le \delta_x^A(b,c)\to \delta_x^A(a,c) \\
&\iff\ \  (\forall x\in X)(\forall a,b,c\in A)\  R(a,b)\otimes \delta_x^A(b,c)\le
\delta_x^A(a,c) \\
&\iff\ \  (\forall x\in X)(\forall a,c\in A)\  \bigvee_{b\in A} R(a,b)\otimes \delta_x^A(b,c)\le \delta_x^A(a,c) \\
&\iff\ \  (\forall x\in X)(\forall a,c\in A)\  R\circ \delta_x^A(a,c)\le \delta_x^A(a,c) \\
&\iff\ \  (\forall x\in X)\  R\circ \delta_x^A\le \delta_x^A \\
&\iff\ \  (\forall x\in X)\  R\circ \delta_x^A = \delta_x^A ,
\end{aligned}
\]
so $R$ is the strongly right invariant if and only if it belongs
to the principal ideal of ${\cal Q}(A)$ generated by
$R^{\mathrm{sri}}$.

(b) This follows immediately by (a).
\end{proof}

According to (\ref{eq:R.sri}), the greatest strongly right invariant crisp quasi-order can
be effectively compu\-ted~for any fuzzy finite automaton over an arbitrary complete residuated lattice, not necessary locally finite.~However, in cases when we are able to effectively compute the greatest right invariant fuzzy quasi-order, using it we can attain better reduction than using the greatest strongly right invariant quasi-order.~Indeed, the following
example~pre\-sents a fuzzy automaton whose number of states can be reduced
by means of   right invariant~fuzzy~quasi-orders, but it can not be reduced using
strongly right invariant ones.
\begin{example}\label{ex:ri.sri}\rm
Consider again the fuzzy automaton $\cal A$ from Example \ref{ex:R.ri}.~In this example
we  showed~that the greatest right invariant fuzzy quasi-order $R^{\mathrm{ri}}$
on $\cal A$ reduces $\cal A$ to a fuzzy automaton with two states. On the other hand,
the greatest strongly right invariant fuzzy quasi-order $R^{\mathrm{sri}}$ on $\cal
A$ is given by
\[
R^{\mathrm{sri}} =
\begin{bmatrix}
1 & 0 & 1  \\
0 & 1 & 1 \\
0 & 0 & 1
\end{bmatrix},
\]
and the related afterset fuzzy automaton
${\cal A}_2={\cal A}/R^{\mathrm{sri}}=(A_2,X,\delta^{A_2})$ has also three states
and fuzzy~transition relations $\delta_x^{A_2}$ and $\delta_y^{A_2}$ are given by
\[
\delta_x^{A_2} = \delta_x^A\circ R^{\mathrm{sri}} =
\begin{bmatrix}
1 & 0 & 1 \\
0 & 0 & 0 \\
0 & 0 & 0
\end{bmatrix}, \ \ \ \
\delta_y^{A_2} = \delta_y^A\circ R^{\mathrm{sri}} =
\begin{bmatrix}
1 & 0 & 1  \\
1 &  1  & 1 \\
1 & 0 & 1
\end{bmatrix}.
\]
Further, the greatest strongly right invariant fuzzy quasi-order $R_{2}^{\mathrm{sri}}$ on ${\cal A}_2$ is given by
\[
R_2^{\mathrm{sri}} =
\begin{bmatrix}
1 & 0 & 1  \\
0 & 1 & 1 \\
0 & 0 & 1
\end{bmatrix},
\]
and the afterset fuzzy automaton ${\cal A}_2/R_2^{\mathrm{sri}}$
is isomorphic to ${\cal A}_2$.~Therefore, the number of states of
$\cal A$ can not be reduced by means of strongly right invariant
fuzzy quasi-orders.

\end{example}

\section{Weakly right and left invariant fuzzy quasi-orders}\label{sec:WEAKLY}

In the previous sections
we have considered right and left invariant
fuzzy quasi-orders
 and some   special types of these fuzzy quasi-orders.~In this section we study some fuzzy quasi-orders which are more general
than
right and left invariant ones.

Let ${\cal A}=(A,X,\delta^A,\sigma^A,\tau^A)$ be a fuzzy recognizer.~For any
$u\in X^*$ we define fuzzy sets $\sigma_u^A,\tau_u^A\in L^A$ by
\[
\sigma_u^A(a)=\bigvee_{b\in A}\sigma^A(b)\otimes \delta_{*}^A(b,u,a),
\qquad
\tau_u^A(a)=\bigvee_{b\in A}\delta_{*}^A(a,u,b)\otimes \tau^A(b),
\]
for each $a\in A$, i.e.,~   $\sigma_u^A=\sigma^{A}\circ \delta_u^A$ and $\tau_u^A=\delta_u^A\circ \tau^A$.~Evidently, for the empty word $e\in X^*$ we~have~that $\sigma_e^A=\sigma^A$
and  $\tau_e^A=\tau^A$.~Fuzzy sets $\sigma_u^A$ have been already used in \cite{ICB.08},~and~they played a key role in~determi\-ni\-zation of fuzzy automata.~By the same
rule, for any $a\in A$ we define fuzzy languages $\sigma_a^A, \tau_a^A\in L^{X^*}$,~i.e.,
$\sigma_a^A(u)=\sigma_u^A(a)$ and $\tau_a^A(u)=\tau_u^A(a)$, for every $u\in X^*$.~Following
terminology used in \cite{CC.04} for non-determi\-nistic automata, we call $\sigma_a^A$ the {\it left fuzzy language\/} of $a$, and $\tau_a^A$ the {\it right fuzzy language\/} of $a$.~Left fuzzy languages have been already studied in \cite{ICB.08,ICBP}.

A fuzzy quasi-order $R$ on $A$ which is a solution to a system of fuzzy relation
equations
\begin{equation}\label{eq:syst.R.tau}
R\circ \tau_u^A = \tau_u^A , \ \ \text{for every}\ u\in X^* ,
\end{equation}
is called a {\it weakly right invariant fuzzy quasi-order\/} on the fuzzy recognizer
$\cal A$, and if $R$ is a solution to
\begin{equation}\label{eq:syst.R.sigma}
\sigma_u^A\circ R = \sigma_u ^A, \ \ \text{for every}\ u\in X^* ,
\end{equation}
then it is called a {\it weakly left invariant fuzzy quasi-order\/} on $\cal A$.~Fuzzy equivalences on $\cal A$ which are~solu\-tions to (\ref{eq:syst.R.tau}) will be called {\it weakly right invariant fuzzy equivalences\/}, and those which are solutions
to~(\ref{eq:syst.R.sigma}) will be called {\it weakly left invariant fuzzy equivalences\/}.

We have the following

\begin{theorem}\label{th:syst.R.tau}
Let ${\cal A}=(A,X,\delta^A,\sigma^A,\tau^A)$ be a fuzzy recognizer. Then
\begin{itemize}\parskip=-2pt
\item[{\rm (a)}] The set ${\cal Q}^{\mathrm{wri}}(A)$ of all weakly right invariant
fuzzy quasi orders on $\cal A$ is a principal ideal of the~lattice ${\cal Q}(A)$.~The greatest element of this principal ideal is a fuzzy quasi-order $R^{\mathrm{wri}}$ on $A$ defined by
\begin{equation}\label{eq:R.rq}
R^{\mathrm{wri}}(a,b) = \bigwedge_{u\in X^*} \tau_u^A(b)\to \tau_u^A(a)  , \ \ \ \ \text{for
all $a,b\in A$.}
\end{equation}
Moreover, $R^{\mathrm{wri}}$ is the greatest solution to the system {\rm (\ref{eq:syst.R.tau})} in ${\cal R}(A)$.
\item[{\rm (b)}] Every weakly right invariant fuzzy quasi-order on $\cal A$
is a solution to the general system.
\item[{\rm (c)}] Every right invariant fuzzy quasi-order on $\cal A$ is weakly right
invariant.
\end{itemize}
\end{theorem}

\begin{proof}
(a) Beeing an intersection of a family
of fuzzy quasi-orders defined as in (\ref{eq:Rf}), $R^{\mathrm{wri}}$ is a fuzzy~quasi-order.~According to   results from
\cite{San.76} (see also \cite{Perf.04,PG.03,PN.07}), $R^{\mathrm{wri}}$ is the greatest solution to (\ref{eq:syst.R.tau}), and it~is~easy to check that solutions to (\ref{eq:syst.R.tau})
in ${\cal Q}(A)$ form an ideal of ${\cal Q}(A)$, and thus, they~form~a principal
ideal~of~${\cal Q}(A)$.

Let $R$ be an arbitrary solution to (\ref{eq:syst.R.tau}) in ${\cal R}(A)$.~The
equality relation $I$  on $A$ is also a solution to (\ref{eq:syst.R.tau}), and by
(\ref{eq:comp.sup}) we obtain that $(R\lor I)^\infty $ is a solution to (\ref{eq:syst.R.tau}).~Since $(R\lor I)^\infty $ is a fuzzy quasi-order on $A$, we conclude that
$R\le (R\lor I)^\infty \le R^{\mathrm{wri}}$, and therefore, $R^{\mathrm{wri}}$ is the greatest solution
to (\ref{eq:syst.R.tau}) in ${\cal R}(A)$.

(b) Let $R$ be an arbitrary weakly right invariant fuzzy quasi-order on $\cal A$.~By induction on $n$ we~will prove that
\begin{equation}\label{eq:gen.syst.2}
 R\circ \delta_{x_1}^A\circ R\circ \delta_{x_2}^A\circ R \circ \cdots
\circ R\circ \delta_{x_n}^A\circ R\circ \tau^A = \delta_{x_1}^A\circ \delta_{x_2}^A\circ \cdots \circ
\delta_{x_n}^A\circ \tau^A ,
\end{equation}
for every $n\in \Bbb N$ and all $x_1,x_2,\ldots ,x_n\in X$.~First we note that $\tau_e^A=\tau^A $, where $e\in X^*$ is the empty~word, and by (\ref{eq:syst.R.tau}) we obtain that $R\circ \tau^A =\tau^A $. By this and by (\ref{eq:syst.R.tau}), for each $x\in X$ we have that
\[
 R\circ \delta_x^A\circ R\circ \tau^A =  R\circ \delta_x^A\circ  \tau^A =   \delta_x^A\circ  \tau^A  ,
\]
and hence, (\ref{eq:gen.syst.2}) holds for $n=1$.~Suppose now that (\ref{eq:gen.syst.2}) holds for some $n\in \Bbb N$.~Then by (\ref{eq:gen.syst.2})
and (\ref{eq:syst.R.tau}), for arbitrary $x_1,\ldots ,x_n,x_{n+1}\in X$ we have that
\[
\begin{aligned}
&R\circ \delta_{x_1}^A\circ R\circ \delta_{x_2}^A\circ  \cdots
\circ R\circ \delta_{x_n}^A\circ R\circ \delta_{x_{n+1}}^A\circ R\circ \tau^A = \\
&\hspace{20mm}= R\circ \delta_{x_1}^A\circ (R\circ \delta_{x_2}^A\circ  \cdots
\circ R\circ \delta_{x_n}^A\circ R\circ \delta_{x_{n+1}}^A\circ R\circ \tau^A )\\
&\hspace{20mm}= R\circ \delta_{x_1}^A\circ ( \delta_{x_2}^A\circ  \cdots
\circ  \delta_{x_n}^A\circ \delta_{x_{n+1}}^A\circ \tau^A) \\
&\hspace{20mm}= R\circ \delta_{x_1}^A\circ  \delta_{x_2}^A\circ  \cdots
\circ  \delta_{x_n}^A\circ \delta_{x_{n+1}}^A\circ \tau^A \\
&\hspace{20mm}= \delta_{x_1}^A\circ  \delta_{x_2}^A\circ  \cdots
\circ  \delta_{x_n}^A\circ \delta_{x_{n+1}}^A\circ \tau^A .
\end{aligned}
\]
Therefore, by induction we conclude that (\ref{eq:gen.syst.2}) holds for every $n\in \Bbb N$.~Finally, it follows immediately~by~(\ref{eq:gen.syst.2}) that $R$ is a solution
to the general system.

(c) Let $R$ be a right invariant fuzzy quasi-order on $\cal
A$.~For each $u\in X^*$ we have   $R\circ \delta_u^A\circ
R=\delta_u^A\circ R$, and also $R\circ \tau^A =\tau^A$, what
implies $R\circ \tau_u^A=R\circ \delta_u^A\circ \tau  ^A= R\circ
\delta_u^A\circ R\circ \tau^A = \delta_u^A\circ R\circ \tau^A =
\delta_u^A\circ \tau^A = \tau_u^A$.~Hence, $R$ is the weakly right
invariant.
\end{proof}

Let us note that $R^{\mathrm{wri}}$ can be also represented by
\begin{equation}\label{eq:R.rq.2}
R^{\mathrm{wri}}(a,b) = \bigwedge_{u\in X^*} \tau_b^A(u)\to \tau_a^A(u), \ \ \ \ \text{for
all $a,b\in A$,}
\end{equation}
i.e., $R^{\mathrm{wri}}(a,b)$ can be interpreted as the degree of inclusion of a  fuzzy
language $\tau_b^A$ in the fuzzy language~$\tau_a^A$.

Analogously, we can define a fuzzy quasi-order $R^{\mathrm{wli}}$ on $\cal A$ by
\begin{equation}\label{eq:R.lq}
R^{\mathrm{wli}}(a,b) =  \bigwedge_{u\in X^*} \sigma_u^A(a)\to \sigma_u^A(b) = \bigwedge_{u\in X^*} \sigma_a^A(u)\to \sigma_b^A(u), \ \ \ \ \text{for
all $a,b\in A$,}
\end{equation}
and we can prove that $R^{\mathrm{wli}}$ is the greatest weakly left invariant fuzzy
quasi-order on $\cal A$, that every weakly left invariant fuzzy quasi-order on $\cal
A$ is also a solution to the general system,~and that every left invariant fuzzy quasi-order on $\cal A$ is weakly left invariant.~We can also show that the greatest weakly right invariant fuzzy equivalence $E^{\mathrm{wrie}}$ on ${\cal A}$ is given by
\begin{equation}\label{eq:E.re}
E^{\mathrm{wrie}}(a,b) =  \bigwedge_{u\in X^*} \tau_u^A(a)\lra \tau_u^A(b) = \bigwedge_{u\in X^*} \tau_a^A(u)\lra \tau_b^A(u), \ \ \ \ \text{for
all $a,b\in A$,}
\end{equation}
and the greatest weakly left invariant fuzzy equivalence $E^{\mathrm{wlie}}$ on $\cal A$ is given by
\begin{equation}\label{eq:E.le}
E^{\mathrm{wlie}}(a,b) =  \bigwedge_{u\in X^*} \sigma_u^A(a)\lra \sigma_u^A(b) = \bigwedge_{u\in X^*} \sigma_a^A(u)\lra \sigma_b^A(u), \ \ \ \ \text{for
all $a,b\in A$,}
\end{equation}
Clearly,  $E^{\mathrm{wrie}}$ is the natural fuzzy equivalence of $R^{\mathrm{wri}}$,
and $E^{\mathrm{wlie}}$ is the natural fuzzy equivalence~of~$R^{\mathrm{wlie}}$. We will also call $R^{\mathrm{wri}}$ the {\it right Myhill-Nerode's fuzzy quasi-order\/} of  $\cal A$, $R^{\mathrm{wli}}$ the {\it left  Myhill-Nerode's~fuzzy quasi-order\/} of  $\cal A$,  $E^{\mathrm{wrie}}$ the {\it right Myhill-Nerode's fuzzy eqivalence\/} of $\cal A$, and $E^{\mathrm{wlie}}$ the {\it left Myhill-Nerode's fuzzy eqivalence\/} of $\cal A$.~Note
that a fuzzy relation $N_\sigma $ on the free monoid $X^*$ defined in a similar way
by
\[
N_\sigma (u,v) =  \bigwedge_{a\in A} \sigma_u^A(a)\lra \sigma_v^A(a) = \bigwedge_{a\in A} \sigma_a^A(u)\lra \sigma_a^A(v), \ \ \ \ \text{for all $u,v\in X^*$,}
\]
is called the {\it Nerode's fuzzy right congruence\/} on $X^*$.~Nerode's fuzzy right
congruences and Myhill's fuzzy congruences on free monoids associated with fuzzy automata have been studied in \cite{ICB.08,ICBP}.

The following example shows that there are weakly right invariant fuzzy quasi-orders which are not~right invariant, and that weakly right invariant fuzzy quasi-orders generally give better reductions than right invariant ones, according to Remark \ref{re:afrcardinality}.

\begin{example}\rm
Let $\cal L$ be the Boolean structure and  ${\cal A}=(A,X,\delta^A,\sigma^A,\tau^A)$  a fuzzy recognizer over $\cal L$, where $A=\{1,2,3,4\}$, $X=\{x\}$, $\sigma^A$
is any fuzzy subset of $A$ and  $\delta_x^A$, and  $\tau^A$ are given by
\[
\delta_x^A=\begin{bmatrix}
1 & 0 & 0 & 0 \\
0 & 0 & 0 & 1 \\
0 & 0 & 0 & 0 \\
0 & 0 & 0 & 0
\end{bmatrix},\ \ \ \
\tau^A=\begin{bmatrix}
0 \\
0 \\
1 \\
0
\end{bmatrix}.
\]
For the sake of simplicity set $\tau^A=\tau $.~As we have noted before, the greatest right invariant fuzzy~equivalence on the fuzzy recognizer $\cal A$ is the greatest right invariant fuzzy equivalence on the fuzzy automaton $(A,X,\delta^A)$
contained in the fuzzy quasi-order $R^{\tau }$. In this example we have
 \[
R^\tau =\begin{bmatrix}
1 & 1 & 0 & 1 \\
1 & 1 & 0 & 1 \\
1 & 1 & 1 & 1 \\
1 & 1 & 0 & 1
\end{bmatrix},
\]
and hence, applying the procedure from Theorem \ref{th:alg} to $R^\tau $ we obtain that the greatest right invariant fuzzy equivalence $R^{\mathrm{ri}}$ on $\cal A$ is
\[
R^{\mathrm{ri}} =
\begin{bmatrix}
1 & 1 & 0 & 1 \\
0 & 1 & 0 & 1 \\
0 & 0 & 1 & 1 \\
0 & 0 & 0 & 1
\end{bmatrix}.
\]
On the other hand, we have that $\tau_e=\tau $ and
\[
\tau_x = \delta_x^A\circ \tau  =
\begin{bmatrix}
0 \\
0 \\
0 \\
0
\end{bmatrix},\ \ \tau_{x^2}=\delta_x^A\circ \tau_x=\tau_x
\]
what means that $\tau_u=\tau_x$,
for every $u\in X^*$, $u\ne e$, whence
\[
R^{\mathrm{wri}} = R^{\tau } \land R^{\tau_x} = \begin{bmatrix}
1 & 1 & 0 & 1 \\
1 & 1 & 0 & 1 \\
1 & 1 & 1 & 1 \\
1 & 1 & 0 & 1
\end{bmatrix}\land
\begin{bmatrix}
1 & 1 & 1 & 1 \\
1 & 1 & 1 & 1 \\
1 & 1 & 1 & 1 \\
1 & 1 & 1 & 1
\end{bmatrix} =\begin{bmatrix}
1 & 1 & 0 & 1 \\
1 & 1 & 0 & 1 \\
1 & 1 & 1 & 1 \\
1 & 1 & 0 & 1
\end{bmatrix} .
\]
Hence, $R^{\mathrm{ri}}$ is strictly smaller than
$R^{\mathrm{wri}}$, and $R^{\mathrm{ri}}$ do not reduce the number
of states of $\cal A$, whereas $R^{\mathrm{wri}}$~reduces $\cal A$
to a fuzzy~recog\-nizer ${\cal
A}/R^{\mathrm{wri}}=(A_2,X,\delta^{A_2},\sigma^{A_2},\tau^{A_2})$
with two states, where $\delta_x^{A_2}$ and $\tau^{A_2}$ are
given~by
\[
\delta_x^{A_2}=\begin{bmatrix}
1 & 0 \\
1 & 0
\end{bmatrix}, \ \
\tau^{A_2}=\begin{bmatrix}
0 \\
1
\end{bmatrix},
\]
and $\sigma^{A_2}$ is defined as in (\ref{eq:sE}).
\end{example}

However, although weakly right invariant and weakly left invariant
fuzzy quasi-orders generally give better reductions than right
invariant and left invariant ones, they have a serious
shortcoming.~For~fuzzy automata and fuzzy recognizers over a
locally finite complete residuated lattice, the greatest~right and
left invariant fuzzy equivalences can be computed in a polynomial
time, using a procedure from Theorem~\ref{th:alg}, but computing
the greatest~weakly right and left invariant ones is
computationally hard.~Namely, any particular equation $R\circ
\tau_u^A=\tau_u^A$~in~(\ref{eq:syst.R.tau}) can be easily solved
if the fuzzy set $\tau_u^A$ is given, but computing  $\tau_u^A$,
for all $u\in X^*$, may be very~hard.~In~fact, computing
$\tau_u^A$, for all $u\in X^*$, is nothing else than
determinization of the reverse fuzzy recognizer of $\cal A$,
whereas computing $\sigma_u^A$, for all $u\in X^*$, is the
determinization of $\cal A$ using a procedure developed in
\cite{ICB.08},~called the {\it accessible fuzzy subset
construction\/}.~It is well-known that determini\-zation of crisp
non-deterministic recognizers may require an exponential time,
because numbers of elements of the  sets $\{\sigma_u^A\mid u\in
X^*\}$ and $\{\tau_u^A\mid u\in X^*\}$  may be exponential in the
number of states~of $\cal A$, and in the case of fuzzy
recognizers~these sets may even be infinite.~Conditions under
which these sets must be finite have been determined
in~\cite{ICB.08,ICBP}.~Moreover, because of exponential growth in
the~number~of states during determinization of non-deterministic
recognizers, state reduction procedures are often used to decrease
the number of states prior to determinization. But, here we have
that determinization is needed prior to the state reduction by
means of  the greatest~weakly right and left invariant fuzzy
quasi-orders.

\section{Alternate reductions}\label{sec:ALTERNATE}

In this section we show that better reductions can be obtained alternating
reductions by means of the greatest right and left invariant fuzzy quasi-orders,
or the greatest weakly right and left invariant fuzzy quasi-orders.~We show that
even if any of these fuzzy quasi-orders separately do not reduce the number~of states, alternating right and left invariant ones, or weakly right and left invariant ones, the number of states can be reduced.

\begin{theorem}\label{th:R.RS}
Let ${\cal A}$ be a fuzzy automaton or a fuzzy recognizer, let $R$ be a right invariant fuzzy
quasi-order on~$\cal A$ and let $S$ be a fuzzy quasi-order on the set of states of
 $\cal A$ such that $R\le S$.~Then
\begin{itemize}\parskip=-2pt
\item[{\rm (a)}] $S$ is a right invariant fuzzy quasi-order on
$\cal A$ if and only if $S/R$ is a right invariant fuzzy
quasi-order on ${\cal A}/R$; \item[{\rm (b)}] $S$ is the greatest
right invariant fuzzy quasi-order on $\cal A$ if and only if $S/R$
is the greatest right~invar\-iant fuzzy quasi-order on ${\cal
A}/R$; \item[{\rm (c)}] $R$ is the greatest right invariant fuzzy
quasi-order on ${\cal A}$ if and only if $\widetilde R$ is the
greatest right~invar\-iant fuzzy quasi-order on  ${\cal A}/R$.
\end{itemize}
\end{theorem}

\begin{proof} First we note that $R\le S$ is equivalent to $R\circ S=S\circ R=S$.

(a) Assume first that ${\cal A}=(A,X,\delta^{A} )$ is a fuzzy automaton. Consider any $a,b\in A$ and~$x\in X$. Then
\begin{equation}\label{eq:SR1}
\begin{aligned}
(\delta_x^{A/R}\circ S/R)(R_a,R_b) &= \bigvee_{c\in A}\delta_x^{A/R}(R_a,R_c)\otimes S/R(R_c,R_b) = \bigvee_{c\in A}(R\circ \delta_x^{A}\circ R)(a,c)\otimes S(c,b) \\
&= (R\circ \delta_x^{A}\circ R\circ S)(a,b) = (\delta_x^{A}\circ R\circ S)(a,b) =
(\delta_x^{A}\circ S)(a,b) ,
\end{aligned}
\end{equation}
and by~the proof of Theorem \ref{th:SIT} it follows that
\begin{equation}\label{eq:SR2}
(S/R\circ \delta_x^{A/R}\circ S/R)(R_a,R_b) = (S\circ \delta_x^A\circ S)(a,b).
\end{equation}
Therefore, by (\ref{eq:SR1}) and (\ref{eq:SR2}) we obtain that (a) holds.

Next, let ${\cal A}=(A,X,\delta^{A},\sigma^A,\tau^A )$ be a fuzzy recognizer. Then
for any $a\in A$ we have that
\[
(S/R\circ \tau^{A/R})(R_a)=\bigvee_{b\in A} S/R(R_a,R_b)\otimes \tau^{A/R}(R_b)=\bigvee_{b\in
A}S(a,b)\otimes \tau^A(b)=(S\circ \tau^A)(a),
\]
so $S/R\circ \tau^{A/R}=\tau^{A/R}$ if and only if $S\circ \tau^A=\tau^A$.
Therefore, in this case we also have that (a) holds.

(b) Let $S$ be the greatest right invariant fuzzy quasi-order on $\cal A$.~By  (a), $S/R$ is a right~invar\-iant fuzzy quasi-order on ${\cal A}/R$.~Let $Q$ be the greatest right invariant fuzzy quasi-order on ${\cal A}/R$. Define
a fuzzy relation $T$ on $\cal A$ by
\[
T(a,b)=Q(R_a,R_b), \qquad \text{for all $a,b\in A$.}
\]
It is easy to verify that $T$ is a fuzzy quasi-order on $\cal A$.~According to (a), $\widetilde R$ is a right invariant~fuzzy quasi-order on ${\cal A}/R$,
what implies $\widetilde R\le Q$, and for arbitrary $a,b\in A$ we
obtain that
\[
R(a,b)=\widetilde R(R_a,R_b)\le Q(R_a,R_b) = T(a,b),
\]
what means that $R\le T$. Therefore, we have that $Q=T/R$, and by  (a)  we obtain~that $T$ is a right invariant fuzzy quasi-order on $\cal A$, what implies  $T\le S$. Now, according to (\ref{eq:FE.GE}), we have that $Q=T/R\le S/R$,
and since $S/R$ is a right invariant fuzzy quasi-order on ${\cal A}/R$, we conclude
that $Q=S/R$, i.e., $S/R$ is the greatest right invariant fuzzy quasi-order on ${\cal A}/R$.

Conversely, let $S/R$ be the greatest right invariant fuzzy quasi-order on ${\cal
A}/R$.~According to (a), $S$ is~a right invariant fuzzy
quasi-order on $\cal A$. Let $T$ be the greatest right invariant~fuzzy quasi-order
on~$\cal A$.~Then  we have that $R\le S\le T$, and by (a) it follows that $T/R$ is
a right invariant fuzzy quasi-order on ${\cal A}/R$, what yields $T/R\le S/R$. Now,
by (\ref{eq:FE.GE}) it follows that $T\le S$, and hence, $T=S$,~and we have proved
that $S$ is the greatest right invariant fuzzy quasi-order on $\cal A$.

(c) This assertion follows immediately by (b).
\end{proof}

Certainly, the previous theorem also holds for left invariant fuzzy quasi orders.
Furthermore, we~have~that a similar theorem concerning weakly right invariant fuzzy
quasi-orders is true:

\begin{theorem}\label{th:R.RS.wri}
Let ${\cal A}=(A,X,\delta^{A},\sigma^A,\tau^A)$ be a fuzzy recognizer, let $R$ be a weakly right invariant fuzzy
quasi-order on~$\cal A$ and let $S$ be a fuzzy quasi-order on $A$ such that $R\le S$.~Then
\begin{itemize}\parskip=-2pt
\item[{\rm (a)}] $S$ is a weakly right invariant fuzzy quasi-order on $\cal A$ if and only
if $S/R$ is a weakly right invariant fuzzy quasi-order on ${\cal A}/R$;
\item[{\rm (b)}] $S$ is the greatest weakly right invariant fuzzy quasi-order on $\cal A$ if and only
if $S/R$ is the greatest weakly right~invar\-iant fuzzy quasi-order on ${\cal A}/R$;
\item[{\rm (c)}] $R$ is the greatest weakly right invariant fuzzy quasi-order on ${\cal A}$ if and only
if $\widetilde R$ is the greatest weakly right~invar\-iant fuzzy quasi-order on  ${\cal A}/R$.
\end{itemize}
\end{theorem}

\begin{proof} (a) For arbitrary $a\in A$ and $u=x_1\dots x_n\in X^*$,  $x_1,\dots
,x_n\in X$, by (\ref{eq:gen.syst.2}) we obtain that
\[
\begin{aligned}
\tau_u^{A/R}(R_a)&= (\delta_u^{A/R}\circ \tau^{A/R})(R_a) =
\bigvee_{b\in A}\delta_u^{A/R}(R_a,R_b)\otimes \tau^{A/R}(R_b) \\
& = \bigvee_{b\in A}(R\circ \delta_{x_1}^A\circ R\circ \cdots \circ R\circ \delta_{x_n}^A\circ
R)(a,b)\otimes (R\circ \tau^A)(b) \\
& =  (R\circ \delta_{x_1}^A\circ R\circ \cdots \circ R\circ \delta_{x_n}^A\circ
R\circ \tau^A)(a) \\
& =  (\delta_{x_1}^A\circ \cdots \circ \delta_{x_n}^A\circ \tau^A)(a) \\
& = \tau_u^A(a).
\end{aligned}
\]
Next, for any $a\in A$ and $u\in X^*$ we have that
\[
(S/R\circ \tau_u^{A/R})(R_a) = \bigvee_{b\in A}S/R(R_a,R_b)\otimes \tau_u^{A/R}(R_b)
= \bigvee_{b\in A}S(a,b)\otimes \tau_u^A(b) = (S\circ \tau_u^A)(a).
\]
Therefore $S/R\circ \tau_u^{A/R}=\tau_u^{A/R}$ if and only if $S\circ \tau_u^A=\tau_u^A$,
and we have proved that (a) is true.

The assertion (b) can be proved similarly as (b) of Theorem \ref{th:R.RS},
and (c) follows immediately by (b).
\end{proof}

~Let $\cal A$ be a fuzzy automaton.~A sequence ${\cal A}_1, {\cal A}_2, \ldots ,{\cal A}_n$ of
fuzzy automata we will call a {\it ${\cal Q}^{\mathrm{ri}}$-reduc\-tion\/} of $\cal A$ if ${\cal A}_1={\cal A}$
and for each $k\in \{1,2,\ldots ,n-1\}$ we have that ${\cal A}_{k+1}$ is the afterset fuzzy automaton~of~${\cal A}_k$ w.r.t.~the greatest right invariant fuzzy quasi-order on ${\cal A}_k$.~Analogously, using left invariant fuzzy~quasi-orders instead of
right invariant ones we define a {\it ${\cal Q}^{\mathrm{li}}$-reduction\/}  of~${\cal A}$, using strongly right and left invariant fuzzy quasi-orders
we define a {\it ${\cal Q}^{\mathrm{sri}}$-reduction\/} and a {\it ${\cal Q}^{\mathrm{sli}}$-reduction\/} of $\cal A$, and using right and left invariant fuzzy equivalences we define a {\it ${\cal E}^{\mathrm{ri}}$-reduction\/} and a {\it ${\cal E}^{\mathrm{li}}$-reduction\/} of $\cal A$.~If we consider fuzzy recognizers,~in a similar way we define ${\cal Q}^{\mathrm{wri}}$- and {\it
${\cal Q}^{\mathrm{wli}}$-reductions\/},
as well as ${\cal Q}^{\mathrm{ri}}$- and {\it ${\cal Q}^{\mathrm{li}}$-reductions\/}
of fuzzy recognizers.

Let us note that for each fuzzy finite automaton ${\cal A}$ there exists a ${\cal Q}^{\mathrm{ri}}$-reduction ${\cal A}_1, {\cal A}_2, \ldots ,{\cal A}_n$ of~$\cal A$ such that for every ${\cal Q}^{\mathrm{ri}}$-reduction ${\cal A}_1, {\cal A}_2, \ldots ,{\cal A}_n,{\cal A}_{n+1},\ldots ,
{\cal A}_{n+m}$ of $\cal A$ which is a continuation of this reduc\-tion we have that
\[
|{\cal A}_n|=|{\cal A}_{n+1}|=\dots =|{\cal A}_{n+m}|,\vspace{-2mm}
\]
i.e., all fuzzy automata ${\cal A}_{n+1}, \ldots ,{\cal A}_{n+m}$ have the same number of states as ${\cal A}_n$.~Also, there is a~shortest ${\cal Q}^{\mathrm{ri}}$-reduction ${\cal A}_1, {\cal A}_2, \ldots ,{\cal A}_n$ of $\cal A$ having this property, which we will call the {\it shortest ${\cal Q}^{\mathrm{ri}}$-reduction\/}~of~$\cal A$, and then we will call ${\cal A}_n$ a {\it ${\cal Q}^{\mathrm{ri}}$-reduct\/} of ${\cal A}$, and we will cal $n$ the~{\it length\/}~of this shortest ${\cal Q}^{\mathrm{ri}}$-reduction.~If~a
fuzzy automaton $\cal A$ is its own ${\cal Q}^{\mathrm{ri}}$-reduct, then it is called {\it
${\cal Q}^{\mathrm{ri}}$-reduced}.~Analogously we define a {\it ${\cal Q}^{\mathrm{li}}$-reduct\/} of ${\cal A}$ and a {\it ${\cal Q}^{\mathrm{li}}$-reduced\/} fuzzy
automaton, as well as {\it ${\cal Q}^{\mathrm{sri}}$-\/} and {\it ${\cal Q}^{\mathrm{sli}}$-reducts\/}, {\it ${\cal Q}^{\mathrm{sri}}$-\/} and {\it ${\cal Q}^{\mathrm{sli}}$-reduced\/} fuzzy automata, {\it ${\cal E}^{\mathrm{ri}}$-\/} and {\it ${\cal E}^{\mathrm{li}}$-reducts\/}, {\it ${\cal E}^{\mathrm{ri}}$-\/} and {\it ${\cal E}^{\mathrm{li}}$-reduced\/} fuzzy
auto\-mata, and other related notions.~For fuzzy recognizers we similarly define
 ${\cal Q}^{\mathrm{wri}}$- and {\it
${\cal Q}^{\mathrm{wli}}$-reducts\/}, ${\cal Q}^{\mathrm{wri}}$- and {\it
${\cal Q}^{\mathrm{wli}}$-reduced\/} fuzzy~recognizers,~${\cal Q}^{\mathrm{ri}}$-~and {\it
${\cal Q}^{\mathrm{li}}$-reducts\/}, ${\cal Q}^{\mathrm{ri}}$- and {\it
${\cal Q}^{\mathrm{li}}$-reduced\/} fuzzy~recognizers, and so forth.

The next theorem shows that  length of the shortest ${\cal Q}^{\mathrm{ri}}$- and ${\cal Q}^{\mathrm{li}}$-reductions do not exceed
2.~

\begin{theorem}\label{th:rred} A fuzzy recognizer {\rm ({\it automaton\/})} ${\cal A}$ is
${\cal Q}^{\mathrm{ri}}$-reduced if and only if the greatest right invariant fuzzy quasi-order $R^{\mathrm{ri}}$ on $\cal A$ is a fuzzy order.

Consequently, for each fuzzy finite recognizer {\rm ({\it automaton\/})} ${\cal A}$, the afterset fuzzy~recognizer {\rm ({\it automaton\/})} ${\cal A}/R^{\mathrm{ri}}$ is ${\cal Q}^{\mathrm{ri}}$-reduced.
\end{theorem}

\begin{proof}
Let  $\cal A$ be ${\cal Q}^{\mathrm{ri}}$-reduced.~If $R^{\mathrm{ri}}$ is not a fuzzy order, then $|{\cal A}/R^{\mathrm{ri}}|<|{\cal A}|$, what contradicts our~starting hypothesis that $\cal A$ is ${\cal Q}^{\mathrm{ri}}$-reduced.~Thus, we conclude that $R^{\mathrm{ri}}$ is a fuzzy order.

Conversely, let $R^{\mathrm{ri}}$ be a fuzzy order. Consider an arbitrary ${\cal Q}^{\mathrm{ri}}$-reduction ${\cal
A}_1={\cal A}, {\cal A}_2, \ldots ,{\cal A}_n$~of~$\cal A$.
For each $k\in \{1,2,\ldots ,n\}$ let $R^{\mathrm{ri}}_k$ be the greatest right invariant fuzzy quasi-order
on ${\cal A}_k$.~By Theorem~\ref{th:R.RS}, for every $k\in \{2,\ldots ,n\}$
we have that $R^{\mathrm{ri}}_k=\widetilde R^{\mathrm{ri}}_{k-1}$, so $R^{\mathrm{ri}}_k$ is a fuzzy order, and by the
hypothesis, $R^{\mathrm{ri}}_1=R^{\mathrm{ri}}$ is a fuzzy order.~Now, for every $k\in \{2,\ldots ,n\}$
we have that $|{\cal A}_k|=|({\cal A}_{k-1})/{R^{\mathrm{ri}}_{k-1}}|=|{\cal A}_{k-1}|$, and hence,
$|{\cal A}|=|{\cal A}_1|=|{\cal A}_2|=\dots =|{\cal A}_n|$. Therefore, the fuzzy recognizer (automaton) $\cal A$ is ${\cal Q}^{\mathrm{ri}}$-reduced.

Further, let $\cal A$ be an arbitrary fuzzy finite recognizer (automaton) and  $R^{\mathrm{ri}}$  the greatest
right invariant fuzzy quasi-order on $\cal A$.~Then by Theorem~\ref{th:R.RS} it follows
that $\widetilde R^{\mathrm{ri}}$ is the greatest right invariant fuzzy~quasi-order on the afterset
fuzzy recognizer (automaton) ${\cal A}/R^{\mathrm{ri}}$, and since it is a fuzzy order, we conclude~that ${\cal A}/R^{\mathrm{ri}}$~is~${\cal Q}^{\mathrm{ri}}$-reduced.
\end{proof}

Similarly we prove the following:

\begin{theorem}\label{th:wri.red} A fuzzy recognizer ${\cal A}$ is
${\cal Q}^{\mathrm{wri}}$-reduced if and only if the greatest weakly right invariant fuzzy quasi-order $R^{\mathrm{wri}}$ on $\cal A$ is a fuzzy order.

Consequently, for each fuzzy finite recognizer ${\cal A}$, the afterset fuzzy~recognizer ${\cal A}/R^{\mathrm{wri}}$ is ${\cal Q}^{\mathrm{wri}}$-reduced.
\end{theorem}

If a fuzzy automaton ${\cal A}=(A,X,\delta^{A} )$ is ${\cal
Q}^{\mathrm{ri}}$-reduced, that is, if the greatest right
invariant fuzzy quasi-order $R^{\mathrm{ri}}$ on $\cal A$ is a
fuzzy order, then the afterset fuzzy automaton ${\cal
A}/R^{\mathrm{ri}}$ has the same cardinality as~$\cal A$, but
it~is not necessary isomorphic to $\cal A$ (see Example
\ref{ex:all.in.one}).~If~the afterset fuzzy automaton~${\cal
A}/R^{\mathrm{ri}}$~is~isomorphic to~$\cal A$, ~then~$\cal A$ is
called {\it completely ${\cal
Q}^{\mathrm{ri}}$-reduced\/}.~Analogously we define
 {\it completely ${\cal Q}^{\mathrm{li}}$-, ${\cal
 E}^{\mathrm{ri}}$-}
and {\it ${\cal E}^{\mathrm{li}}$-reduced\/} fuzzy automata, as
well as {\it completely ${\cal Q}^{\mathrm{ri}}$-, ${\cal
Q}^{\mathrm{li}}$-, ${\cal Q}^{\mathrm{wri}}$-}, and {\it ${\cal
Q}^{\mathrm{wli}}$-reduced\/} fuzzy recognizers.

Example \ref{ex:all.in.one} will show that even if a fuzzy
recognizer or a fuzzy automaton $\cal A$ is ${\cal
Q}^{\mathrm{wri}}$- and/or ${\cal Q}^{\mathrm{wli}}$-reduced, or
it is ${\cal Q}^{\mathrm{ri}}$- and/or ${\cal
Q}^{\mathrm{li}}$-reduced, it is still possible to continue
reduction of  the number of states of $\cal A$ alternating
reductions by means of the greatest weakly right and left
invariant fuzzy quasi-orders,~or~by means of the greatest right
and left invariant fuzzy quasi-orders.~For that reason we
introduce the following concepts.

Let $\cal A$ be a fuzzy automaton.~A sequence ${\cal A}_1, {\cal A}_2, \ldots ,{\cal A}_n$ of
fuzzy automata will be called an {\it alternate ${\cal Q}$-reduction\/} of $\cal A$ if ${\cal A}_1={\cal A}$
and for every $k\in \{1,2,\ldots ,n-2\}$ the following is true:
\begin{itemize}\parskip=-2pt
\item[(1)]  ${\cal A}_{k+1}$ is the afterset fuzzy automaton of  ${\cal A}_k$ w.r.t.~the greatest right invariant or the greatest left invariant fuzzy quasi-odred on ${\cal A}_k$;
\item[(2)] If ${\cal A}_{k+1}$ is the afterset fuzzy automaton of  ${\cal A}_k$ w.r.t.~the greatest
right invariant fuzzy quasi-order on ${\cal A}_k$, then ${\cal A}_{k+2}$ is the afterset fuzzy automaton of ${\cal A}_{k+1}$
w.r.t.~the greatest left invariant fuzzy quasi-order on ${\cal A}_k$;
\item[(3)]  If ${\cal A}_{k+1}$ is the afterset fuzzy automaton of ${\cal A}_k$ w.r.t.~the greatest
left invariant fuzzy quasi-order on ${\cal A}_k$, then ${\cal A}_{k+2}$ is the afterset fuzzy automaton of ${\cal A}_{k+1}$
w.r.t.~the greatest right invariant fuzzy quasi-order on ${\cal A}_k$.
\end{itemize}
If ${\cal A}_2$ is the afterset fuzzy automaton of ${\cal A}_1$  w.r.t.~the greatest right invariant fuzzy
quasi-order on ${\cal A}_1$,~then this alternate ${\cal Q}$-reduction is called an {\it alternate ${\cal Q}^{\mathrm{rl}}$-reduction\/}, and if  ${\cal A}_2$ is the afterset fuzzy automaton of ${\cal A}_1$  w.r.t.~the greatest left invariant fuzzy
quasi-order on ${\cal A}_1$, then this alternate ${\cal Q}$-reduction is called an {\it alternate ${\cal Q}^{\mathrm{lr}}$-reduction\/}.

Note that for each fuzzy finite automaton ${\cal A}$ there exists an alternate ${\cal Q}^{\mathrm{rl}}$-reduction ${\cal A}_1, {\cal A}_2, \ldots ,{\cal A}_n$ of $\cal A$ such that for every alternate ${\cal Q}^{\mathrm{rl}}$-reduction  ${\cal A}_1, {\cal A}_2, \ldots ,{\cal A}_n,{\cal A}_{n+1},\ldots ,
{\cal A}_{n+m}$ which is a continua\-tion~of this reduction we have that
\[
|{\cal A}_n|=|{\cal A}_{n+1}|=\dots =|{\cal A}_{n+m}|,\vspace{-2mm}
\]
i.e., all fuzzy automata ${\cal A}_{n+1}, \ldots ,{\cal A}_{n+m}$ have the same number of states as ${\cal A}_n$.~Also, there is a~shortest alternate ${\cal Q}^{\mathrm{rl}}$-reduction ${\cal A}_1, {\cal A}_2, \ldots ,{\cal A}_n$ of $\cal A$ having this property, which we will call the {\it shortest  alternate ${\cal Q}^{\mathrm{rl}}$-reduction\/} of $\cal A$, and we will call ${\cal A}_n$  an {\it
alternate ${\cal Q}^{\mathrm{rl}}$-reduct\/} of ${\cal A}$, whereas the number $n$
will be~called
the {\it length\/} of the shortest alternate ${\cal Q}^{\mathrm{rl}}$-reduction of $\cal A$.~Analogously~we~define the {\it shortest alternate ${\cal Q}^{\mathrm{lr}}$-reduction\/}, its length, and the {\it alternate ${\cal Q}^{\mathrm{lr}}$-reduct\/} of ${\cal A}$.~Using the greatest right and left invariant~fuzzy equivalences instead of the greatest right~and left invariant fuzzy quasi-orders, we~also define {\it alternate ${\cal E}$-reduc\-tions\/},
{\it alternate ${\cal E}^{\mathrm{rl}}$-\/}and
{\it ${\cal E}^{\mathrm{lr}}$-reductions\/}, {\it alter\-nate ${\cal E}^{\mathrm{rl}}$-\/}and {\it ${\cal E}^{\mathrm{lr}}$-reducts\/}, etc.~For fuzzy recognizers, weakly right invariant and weakly left invariant fuzzy quasi-orders, similarly we define {\it alternate ${\cal Q}^{\mathrm{w}}$-\/} and {\it ${\cal E}^{\mathrm{w}}$-reductions\/},
{\it alternate\/}~${\cal Q}^{\mathrm{wrl}}$- and {\it
${\cal Q}^{\mathrm{wlr}}$-reductions\/}, {\it alternate\/}~${\cal Q}^{\mathrm{wrl}}$- and {\it
${\cal Q}^{\mathrm{wlr}}$-reducts\/}, as well as  {\it alternate ${\cal Q}$-reductions\/},
{\it alternate\/}~${\cal Q}^{\mathrm{rl}}$-~and {\it
${\cal Q}^{\mathrm{lr}}$-reductions\/}, {\it alternate\/}~${\cal Q}^{\mathrm{rl}}$- and {\it
${\cal Q}^{\mathrm{lr}}$-reducts\/}
of fuzzy
recognizers, and other related concepts.

Consider now the following example.

\begin{example}\label{ex:all.in.one}\rm
Let $\cal L$ be the Boolean structure and let ${\cal A}=(A,X,\delta^A,\sigma^A,\tau^A)$ be a fuzzy recognizer~over $\cal L$, where $A=\{1,2,3\}$, $X=\{x,y\}$, and  $\delta_x^A$,
$\delta_y^A$, $\sigma^A$ and  $\tau^A$ are given by
\[
\delta_x^A=\begin{bmatrix}
1 & 0 & 0 \\
0 & 0 & 0 \\
0 & 0 & 0
\end{bmatrix},\ \ \ \
\delta_y^A=\begin{bmatrix}
0 & 1 & 0 \\
1 & 1 & 1 \\
1 & 0 & 0
\end{bmatrix},\ \ \ \
\sigma^A=\begin{bmatrix}
1 & 0 & 0 \end{bmatrix},
\ \ \ \ \tau^A=\begin{bmatrix}
0 \\
1 \\
1
\end{bmatrix}.
\]
Let us note that the fuzzy automaton $(A,X,\delta^A)$ has been already considered
in Example \ref{ex:all.in.one}.~

The greatest~weakly right invariant fuzzy quasi-order $R^{\mathrm{wri}}$ on $\cal A$ and related afterset fuzzy recognizer ${\cal A}_2={\cal A}/R^{\mathrm{wri}}
=(A_2,X,\delta^{A_2},\sigma^{A_2},\tau^{A_2})$ are given by
\[
R^{\mathrm{wri}} =\begin{bmatrix}
1 & 0 & 0 \\
0 & 1 & 1 \\
0 & 0 & 1
\end{bmatrix}, \quad \delta_x^{A_2}=\begin{bmatrix}
1 & 0 & 0 \\
0 & 0 & 0 \\
0 & 0 & 0
\end{bmatrix}, \quad \delta_y^{A_2}=\begin{bmatrix}
0 & 1 & 1 \\
1 & 1 & 1 \\
1 & 0 & 1
\end{bmatrix}, \quad \sigma^{A_2}=\begin{bmatrix}
1 & 0 & 0 \end{bmatrix}, \quad \tau^{A_2}=\begin{bmatrix}
0 \\
1 \\
1
\end{bmatrix} ,
\]
and the greatest weakly left invariant fuzzy quasi-order $R_{2}^{\mathrm{wli}}$ on ${\cal A}_2$ and related afterset fuzzy recognizer ${\cal A}_3={\cal A}_2/R^{\mathrm{wri}}
=(A_3,X,\delta^{A_3},\sigma^{A_3},\tau^{A_3})$ are given by
\[
R_{2}^{\mathrm{wli}} = \begin{bmatrix}
1 & 0 & 0 \\
0 & 1 & 1 \\
0 & 1 & 1
\end{bmatrix}, \quad \delta_x^{A_3}=\begin{bmatrix}
1 & 0 \\
0 & 0
\end{bmatrix}, \quad \delta_y^{A_3}=\begin{bmatrix}
0 & 1 \\
1 & 1
\end{bmatrix}, \quad \sigma^{A_3}=\begin{bmatrix}
1 & 0 \end{bmatrix}, \quad \tau^{A_3}=\begin{bmatrix}
0 \\
1
\end{bmatrix} .
\]
It can be easily verified that both the greatest weakly right invariant fuzzy
quasi-order and the greatest weakly left invariant fuzzy
quasi-order on ${\cal A}_3$ coincide with the equality relation on $A_3$,
and the afterset fuzzy recognizers of ${\cal A}_3$ w.r.t.~these fuzzy quasi-orders are isomorphic to ${\cal A}_3$.~By this it follows that none alter\-nate ${\cal Q}^{\mathrm{w}}$-reduction decreases the number of states of ${\cal A}_3$, and we obtain that the sequence ${\cal A}={\cal
A}_1$, ${\cal A}_2$, ${\cal A}_3$ is the shortest alternate ${\cal Q}^{\mathrm{wrl}}$-reduction of $\cal A$, and~${\cal A}_3$ is the alternate ${\cal Q}^{\mathrm{wrl}}$-reduct of~$\cal A $.

On the other hand, the greatest~weakly left invariant fuzzy quasi-order $R^{\mathrm{wli}}$ on $\cal A$ and the afterset~fuzzy recognizer ${\cal A}_2'={\cal A}/R^{\mathrm{wli}}
=(A_2',X,\delta^{A_2'},\sigma^{A_2'},\tau^{A_2'})$ are given by
\[
R^{\mathrm{wli}} = \begin{bmatrix}
1 & 0 & 0 \\
0 & 1 & 0 \\
1 & 1 & 1
\end{bmatrix}, \quad \delta_x^{A_2'}=\begin{bmatrix}
1 & 0 & 0 \\
0 & 0 & 0 \\
1 & 0 & 0
\end{bmatrix}, \quad \delta_y^{A_2'}=\begin{bmatrix}
0 & 1 & 0 \\
1 & 1 & 1 \\
1 & 1 & 1
\end{bmatrix}, \quad \sigma^{A_2'}=\begin{bmatrix}
1 & 0 & 0 \end{bmatrix}, \quad \tau^{A_2'}=\begin{bmatrix}
0 \\
1 \\
1
\end{bmatrix} ,
\]
and both the greatest weakly right invariant fuzzy
quasi-order and the greatest weakly left invariant~fuzzy
quasi-order on ${\cal A}_2'$ coincide with $R^{\mathrm{wli}}$,
and the afterset fuzzy recognizers of ${\cal A}_2'$ w.r.t.~these fuzzy quasi-orders are isomorphic to ${\cal A}_2'$.~This means that none alter\-nate ${\cal Q}^{\mathrm{w}}$-reduction of ${\cal A}_2'$ decreases the number of states of ${\cal A}_2'$, i.e., none alter\-nate ${\cal Q}^{\mathrm{wlr}}$-reduction decreases the number~of~states of ${\cal A}$, and we obtain that  ${\cal A}$ is its own alternate ${\cal Q}^{\mathrm{wlr}}$-reduct.

Let us note that $R^{\mathrm{wri}}$ and $R_2^{\mathrm{wri}}$ are also the greatest
right invariant fuzzy quasi-orders on fuzzy recogni\-zers $\cal A$ and ${\cal A}_2$,
as well as on fuzzy automata $(A,X,\delta^A)$ and $(A_2,X,\delta^{A_2})$, and
$R^{\mathrm{wli}}$ is also the greatest left invariant fuzzy quasi-order on the
fuzzy recogni\-zer $\cal A$ and the fuzzy automaton $(A,X,\delta^A)$.~Therefore,
everything we have shown for weakly right invariant and weakly left invariant fuzzy quasi-orders holds~also
for right invariant and left invariant ones.
\end{example}

 Example \ref{ex:all.in.one} shows that even if a fuzzy recognizer  $\cal
A$ is ${\cal Q}^{\mathrm{wri}}$- and/or ${\cal
Q}^{\mathrm{wli}}$-reduced, it is still possible to continue
reduction of  the number of states of $\cal A$ alternating
reductions by means of the greatest weakly right and left
invariant fuzzy quasi-orders.~Namely, the fuzzy recognizer $\cal
A$ from this example is both ${\cal Q}^{\mathrm{wri}}$- and ${\cal
Q}^{\mathrm{wli}}$-reduced, but alter\-nate ${\cal
Q}^{\mathrm{wrl}}$-reduction decreases its number of states.~The
same example~also~shows that  shortest alternate ${\cal
Q}^{\mathrm{wrl}}$- and ${\cal Q}^{\mathrm{wlr}}$-reductions can
have different lengths, and that alter\-nate ${\cal
Q}^{\mathrm{wrl}}$- and ${\cal Q}^{\mathrm{wlr}}$-reducts  can
have different number of states.~Indeed, alter\-nate ${\cal
Q}^{\mathrm{wrl}}$-reduction reduces $\cal A$ from~three to two
states, whereas alternate ${\cal Q}^{\mathrm{wlr}}$-reduction do
not decrease number~of states~of~$\cal A$.~The above remarks also
hold for alternate ${\cal Q}^{\mathrm{rl}}$- and ${\cal
Q}^{\mathrm{lr}}$-reductions.

The state reduction of non-deterministic automata and recognizers
by means of right invariant and left~invar\-iant quasi-orders has
been studied by Champarnaud and Coulon \cite{CC.03,CC.04}, Ilie,
Navarro and Yu \cite{INY.04}, and~Ilie, Solis-Oba and Yu
\cite{ISY.05} (see also \cite{IY.02,IY.03}).~In these papers  a
non-deterministic recognizer $\cal A$ has been redu\-ced using
factor recognizers ${\cal A}/E_{R^\mathrm{ri}}$ and ${\cal
A}/E_{R^\mathrm{li}}$ w.r.t.~natural equivalences of
$R^\mathrm{ri}$ and~$R^\mathrm{li}$,~but none~of~the mentioned
authors have considered afterset recognizers ${\cal
A}/{R^\mathrm{ri}}$ and ${\cal A}/{R^\mathrm{li}}$.~As we have
noted earlier,~recognizers ${\cal A}/E_{R^\mathrm{ri}}$ and ${\cal
A}/{R^\mathrm{ri}}$, as well as ${\cal A}/E_{R^\mathrm{li}}$ and
${\cal A}/{R^\mathrm{li}}$, are not necessary~isomorphic, but they
have the same number of states and both of them are equivalent to
$\cal A$.~Therefore,~it~is all the same if we~use  ${\cal
A}/E_{R^\mathrm{ri}}$ or ${\cal A}/{R^\mathrm{ri}}$, and ${\cal
A}/E_{R^\mathrm{li}}$ or ${\cal A}/{R^\mathrm{li}}$.~However,
there are differences if we work with alter\-nate reductions.~For
the recognizer $\cal A$ with three states given in Example
\ref{ex:all.in.one}, natural equivalences $E_{R^\mathrm{ri}}$ and
$E_{R^\mathrm{li}}$ coincide with the equality relation, so
alternate reductions by means of these equivalences do not
decrease the number of states of $\cal A$, but the alter\-nate
${\cal Q}^{\mathrm{wrl}}$-reduction of  $\cal A$ gives a
recognizer~with two~states.~The~same conclusion can be drawn for
alternate $\cal E$-reductions.~Equivalences $E^{\mathrm{ri}}$ and
$E^{\mathrm{li}}$ on $A$ also coincide with the equality relation,
and none~alternate $\cal E$-reduction decrease the number of
states of $\cal A$.

In alternate ${\cal Q}^{\mathrm{w}}$-reductions considered in
Example \ref{ex:all.in.one} we have obtained three consecutive
members which are isomorphic, and by this fact we have concluded
that none alternate ${\cal Q}^{\mathrm{w}}$-reduction can further
decrease the number of states.~A similar conclusion we can draw in
cases when we obtain a fuzzy recognizer with only one
state.~However, we have no yet a general procedure to decide
whether we have reached the smallest number of states in an
alternate $\cal Q$- or ${\cal Q}^{\mathrm{w}}$-reduction.~An
exception are alter\-nate $\cal E$- and~${\cal
E}^{\mathrm{w}}$-reductions of non-deterministic automata and
recognizers,~for which there exists such general procedure.
Indeed, if after two successive steps the number of states did not
changed, then we can be sure that we have reached the smallest
number of states and this alternate $\cal E$- or~${\cal
E}^{\mathrm{w}}$-reduction is finished.~In other words, an
alternate $\cal E$-reduction finishes when~we obtain a
non-deterministic automaton which is both ${\cal
E}^{\mathrm{ri}}$- and ${\cal E}^{\mathrm{li}}$-reduced, and this
automaton is an~alternate $\cal E$-reduct of the staring
automaton.~The same holds for alternate ${\cal
E}^{\mathrm{w}}$-reductions of non-deterministic
recognizers.~Alter\-nate $\cal Q$- and~${\cal
Q}^{\mathrm{w}}$-reductions do not have this property even in the
case of non-deterministic automata and recognizers,   because
making an afterset auto\-maton or~recog\-nizer w.r.t.~an order
relation we change~the transition relation  and we obtain an
auto\-maton or recognizer which~is not necessary isomorphic to the
original one, what makes possible to continue an alternate $\cal
Q$- or~${\cal Q}^{\mathrm{w}}$-reduction and decrease the number
of states (see again~Example~\ref{ex:all.in.one}).~The same
conclusion can be drawn for alter\-nate $\cal Q$-,~${\cal
Q}^{\mathrm{w}}$-, $\cal E$- and~${\cal
E}^{\mathrm{w}}$-reductions of fuzzy automata and recognizers.

Finally, let us give several remarks concerning strongly right and
left invariant fuzzy quasi-orders.~It~can be easily verified that
for every fuzzy quasi-order $R$ on a fuzzy automaton $\cal A$, the
fuzzy order $\widetilde R$ on the afterset fuzzy automaton ${\cal
A}/R$ is strongly invariant, i.e., it is both strongly right and
strongly left invariant.~Consequently, for the greatest right
invariant fuzzy quasi-order $R^{\mathrm{ri}}$ on $\cal A$, by
Theorem~\ref{th:R.RS}~it follows that $\widetilde R^{\mathrm{ri}}$
is the greatest right invariant fuzzy quasi-order on ${\cal
A}/R^{\mathrm{ri}}$, and hence, $\widetilde R^{\mathrm{ri}}$ is
the greatest strongly right invariant fuzzy quasi-order on ${\cal
A}/R^{\mathrm{ri}}$, and every right invariant fuzzy quasi-order
on ${\cal A}/R^{\mathrm{ri}}$ is a strongly right invariant.

However, for the greatest strongly right invariant fuzzy quasi-order $R^{\mathrm{sri}}$ on $\cal A$ we have that $\widetilde R^{\mathrm{sri}}$ is a strongly right invariant fuzzy quasi-order on ${\cal A}/R^{\mathrm{sri}}$, but the next~exam\-ple
shows that it is not necessary the greatest element of ${\cal Q}^{\mathrm{sri}}({\cal
A})$.~For that reason,~the analogue
of Theorem \ref{th:rred} does not hold for strongly right invariant fuzzy quasi-orders,
i.e., the  afterset fuzzy automaton ${\cal A}/R^{\mathrm{sri}}$ is not necessary ${\cal Q}^{\mathrm{sri}}$-reduced, and contrary to ${\cal Q}^{\mathrm{ri}}$-reductions,
a ${\cal Q}^{\mathrm{sri}}$-reduction does not necessary stop after its first step.~This will be also shown by the next example.

\begin{example}\label{ex:qsri.red}\rm
Let $\cal L$ be the Boolean structure and let ${\cal A}=(A,X,\delta^A)$ be a fuzzy
automaton over $\cal L$, where $A=\{1,2,3\}$, $X=\{x\}$, and  a fuzzy transition relation $\delta_x^A$ is given by
\[
\delta_x^A=\begin{bmatrix}
1 & 0 & 1 \\
1 & 0 & 0 \\
1 & 0 & 0
\end{bmatrix}.
\]
Then the greatest strongly right invariant fuzzy quasi-order $R^{\mathrm{sri}}$ on $\cal A$ is given by
\[
R^{\mathrm{sri}} =
\begin{bmatrix}
1 & 1 & 1 \\
0 & 1 & 1 \\
0 & 1 & 1
\end{bmatrix},
\]
the afterset fuzzy automaton ${\cal A}_2={\cal A}/R^{\mathrm{sri}}=(A_2,X,\delta^{A_2})$
has two states, i.e., $A_2=\{1,2\}$, and a fuzzy transition relation $\delta_x^{A_2}$
is given by
\[
\delta_x^{A_2}=
\begin{bmatrix}
1 & 1 \\
1 & 1
\end{bmatrix}.
\]
Consequently, the greatest strongly right invariant fuzzy quasi-order $R_2^{\mathrm{sri}}$ on ${\cal A}_2$ is given by
\[
R_2^{\mathrm{sri}}=
\begin{bmatrix}
1 & 1 \\
1 & 1
\end{bmatrix},
\]
and it reduces ${\cal A}_2$ to a fuzzy automaton
${\cal A}_3={\cal A}_2/R_2^{\mathrm{sri}}=(A_3,X,\delta^{A_3})$ having only one state
and a fuzzy transition relation $\delta_x^{A_3}=\begin{bmatrix}1\end{bmatrix}$.~Therefore,
the sequence ${\cal A}={\cal A}_1$, ${\cal A}_2$, ${\cal A}_3$ is the shortest ${\cal Q}^{\mathrm{sri}}$-reduction~of~$\cal A$.

This example also shows that the converse implication in (a) of Theorem \ref{th:R.RS}
does not necessary hold~for strongly right invariant fuzzy quasi-orders.~Namely,
if we assume that $S$ is the universal relation~on~$A$,~then we have that
$S/R^{\mathrm{sri}}=R_2^{\mathrm{sri}}$ is a strongly right invariant fuzzy quasi-order
on ${\cal A}/R^{\mathrm{sri}}$, but $S$ is not a strongly right invariant fuzzy quasi-order
on $\cal A$.
\end{example}

\section{An example demonstrating some applications to fuzzy discrete event systems}

In this section we give an example demonstrating some applications of weakly left invariant fuzzy
quasi-orders to fuzzy discrete event systems.~A more complete study of fuzzy discrete event systems
will be~a~subject of our further work.

A {\it discrete event system\/} ({\it DES\/}) is a dynamical system whose state space is described
by a discrete set, and states evolve as a result of asynchronously occurring discrete events over
time \cite{CL.08,HZ.07}.~Such~systems~have significant applications in many fields of computer
science and engineering, such as concurrent and distributed software systems, computer and
communication networks, manufacturing, transportation and traffic~control systems,~etc.~Usually, a
discrete event system is modeled by a finite state automaton (deterministic~or nondeterministic),
with events modeled by input letters, and the behavior of a discrete~event~system~is~described by the language generated~by the auto\-maton.~However, in many situations
states and state transitions, as well as control strategies, are somewhat imprecise, uncertain and
vague.~To take this kind of uncertainty into account, Lin and Ying extended   classical discrete
event systems to {\it fuzzy discrete event systems\/} ({\it FDES\/})~by proposing a fuzzy finite
automaton model \cite{LY.01,LY.02}.~Fuzzy discrete event systems have been since studied in a
number of papers \cite{CY.05,CY.06,CYC.07,K.08,LY.01,LY.02,Lall.07,LL.08,Qiu.05,QL.09}, and they
have been successfully applied to biomedical control for HIV/AIDS treatment planning, robotic
control, intelligent~vehicle control, waste-water treatment, examination of chemical reactions, and
in other fields.

In \cite{LY.01,LY.02}, and later in \cite{CYC.07,K.08,Qiu.05,QL.09}, fuzzy discrete event systems have been modeled by automata~with fuzzy states and fuzzy inputs, whose transition function is defined  over the sets of fuzzy states and fuzzy~inputs in
a deterministic way.~In fact, such an automaton can be regarded as the
determinization of a fuzzy~automaton (defined as in this paper) by means of the accessible fuzzy subset construction (see \cite{ICB.08,ICBP}).~On the other hand, in \cite{CY.05,CY.06,LL.08}
fuzzy discrete event systems have been modeled by fuzzy automata~with single crisp
initial states.~In all mentioned papers membership values have been taken in the G\"odel or product structure.~

Here, a fuzzy discrete event system will be modeled by a fuzzy finite recognizer
${\cal A}=(A,X,\delta^A,\sigma^A,\tau^A)$ over a complete residuated lattice $\cal
L$, defined as in Section~\ref{ssec:FAL}.~Two kinds of fuzzy languages associated~with this fuzzy recognizer play a key role in study of fuzzy discrete event systems.~The first
one is the {\it fuzzy language $L({\cal A})$ recognized
by\/} $\cal A$, which is defined as in (\ref{eq:recog}) (or (\ref{eq:recog.comp})), and the second one is the {\it fuzzy language $L_g({\cal A})$ generated by\/} $\cal A$, which is defined by
\begin{equation}\label{eq:Lg}
L_{g}({\cal A})(u) = \bigvee_{a,b\in A} \sigma^{A} (a)\otimes
\delta_{*}^A(a,u,b) =\bigvee_{b\in A}(\sigma^{A} \circ
\delta^{A}_u)(b)=\bigvee_{b\in A}\sigma^{A}_u(b),
\end{equation}
for every $u\in X^*$.~Intuitively, $L_{g}({\cal A})(u)$ represents the degree to which the input word $u$ causes a transition from some initial state to any other
state.~Two fuzzy recognizers $\cal A$ and $\cal B$ are called {\it language-equivalent\/}
if $L({\cal A})=L({\cal B})$ and $L_{g}({\cal A})=L_{g}({\cal B})$.

Discrete event models of complex dynamic systems are
built rarely in a monolithic manner.~Instead, a modular
approach is used where models of individual components
are built first, followed by the composition of these models to
obtain the model of the overall system.~In the automaton modeling formalism the~compo\-si\-tion of individual automata (that model
interacting system components) is usually formalized by the {\it parallel composition\/} of automata.~Once a complete system model has been obtained by parallel
composition of a set of automata, the resulting monolithic model  can be used to analyze the properties of the system.

Let ${\cal A}=(A,X,\delta^{A} ,\sigma^{A} ,\tau^{A})$ and ${\cal
B}=(A,Y,\delta^{B} ,\sigma^{B} ,\tau^{B})$ be fuzzy recognizers.~The {\it product\/}
of $\cal A$ and~$\cal B$~is a fuzzy recognizer
${\cal A}\times {\cal B}=(A\times B, X\cap Y, \delta^{A\times B}, \sigma^{A\times
B},\tau^{A\times B})$, defined by
\begin{equation}\label{eq:product}
\begin{aligned}
&\delta^{A\times B}((a,b),x,(a',b')) = \delta^A(a,x,a')\otimes \delta^B(b,x,b') ,\\
&\sigma^{A\times B}(a,b)=\sigma^A(a)\otimes\sigma^B(b),\quad\tau^{A\times B}(a,b)=\tau^A(a)\otimes\tau^B(b),
\end{aligned}
\end{equation}
for all $a,a'\in A$, $b,b'\in B$ and $x\in X\cap Y$, and the
{\it parallel composition\/}~of $\cal A$ and $\cal B$  is a fuzzy recognizer ${\cal A}\Vert {\cal
B}=(A\times B, X\cup Y, \delta^{A\Vert B}, \sigma^{A\Vert B},
\tau^{A\Vert B})$, defined by
\begin{equation}\label{eq:parallel}
\begin{aligned}
&\delta^{A\Vert B}((a,b),x,(a',b')) =\begin{cases}
\delta^A(a,x,a')\otimes \delta^B(b,x,b') & \ \ \text{if}\ x\in X\cap Y\\
\delta^A(a,x,a') & \ \ \text{if $x\in X\setminus Y$ and $b=b'$}\\
\delta^B(b,x,b') & \ \ \text{if $x\in Y\setminus X$ and $a=a'$}\\
 0 &  \ \ \text{otherwise}
\end{cases},\\
&\sigma^{A\Vert B}(a,b)=\sigma^A(a)\otimes\sigma^B(b),\quad\tau^{A\Vert B}(a,b)=\tau^A(a)\otimes\tau^B(b),
\end{aligned}
\end{equation}for all $a,a'\in A$, $b,b'\in B$.~Associativity  is used to extend the definition of parallel
composition to more than two automata.

In the parallel composition of fuzzy automata $\cal A$ and $\cal B$, a common input~letter from $X\cap Y$ is executed in both automata simultaneously,
what means that these two
automata are synchro\-nized on the common~input letter.~On~the other hand,
a private
input letter from $X\setminus Y$ is
executed in~$\cal A$,~while $\cal B$ is staying in the~same state, and similarly
for private letters from $Y\setminus X$.~Clearly, if $X=Y$, then the parallel composition
reduces to the product.~However, even if $X\ne Y$, the parallel composition of fuzzy
automata can be regarded as the product of suitable input extensions of these fuzzy automata, what will be shown in the sequel.~If $X\cap Y=\emptyset$,
then no synchronized transitions occur and ${\cal A}\Vert {\cal B}$ is the {\it
concurrent behavior\/} of ${\cal A}$ and ${\cal B}$.~This behavior is often termed the {\it shuffle\/} of ${\cal A}$ and ${\cal B}$.

Let ${\cal A}=(A,X,\delta^{A} ,\sigma^{A} ,\tau^{A} )$ be a fuzzy
recognizer and let $Y$ be an alphabet such that $X\subseteq Y$.~Let~us~define a new
transition function
$\delta^{A_Y}:A\times Y\times A\to L$ by
\begin{equation}\label{eq:inp.ext}
\delta^{A_Y}(a,x,a')=\begin{cases}
\delta^A(a,x,a') & \text{if $x\in X$} \\
1 & \text{if $x\in Y\setminus X$ and $a=a'$} \\
0 & \text{otherwise}
\end{cases},
\end{equation}
for all $a,a'\in A$ and $x\in Y$. Then a fuzzy recognizer ${\cal A}_Y=(A,Y,\delta^{A_Y}
,\sigma^{A} ,\tau^{A} )$ is called a $Y$-{\it input extension\/} of $\cal A$.~In~other words,
input letters from $X$ cause in ${\cal A}_Y$ the same transitions as in $\cal
A$, while those from~$Y\setminus X$ cause ${\cal A}_Y$ to stay in the
same state.~Evidently, $\delta_u^{A_Y}$ is the equality relation on $A$, for each $u\in (Y\setminus X)^*$.

An operation frequently performed on words and languages
is the so-called natural projection, which transforms words over an alphabet $Y$ to words over a smaller alphabet $X\subseteq Y$. Formally, a {\it natural projection\/},
or briefly a {\it projection\/}, is a mapping
$\natp{X}:Y^*\to X^*$, where $X\subseteq Y$, defined inductively by
\begin{equation}\label{eq:projection}
\natp{X}(w)=\begin{cases} e &\quad\text{if $w\in (Y\setminus X)^*$}\\
w &\quad\text{if $w\in X^*$}\\
\natp{X}(u)\natp{X}(v) &\quad\text{if $w=uv$, for some $u,v\in Y^*$}
\end{cases},
\end{equation}for each $w\in Y^*$ (cf. \cite{CL.08}).~In other words, the word $\natp{X}(w)\in
X^*$ is obtained from $w$ by deleting all~appear\-ances of letters from $Y\setminus X$.

First we prove the following:

\begin{lemma}\label{lema3} Let ${\cal A}=(A,X,\delta^{A} ,\sigma^{A} ,\tau^{A} )$ be a fuzzy recognizer, let $Y$ be an alphabet such that $X\subseteq Y$,~and~let ${\cal A}_Y=(A,Y,\delta^{A_Y} ,\sigma^{A} ,\tau^{A})$ be the $Y$-input extension of $\cal
A$.~Then for every $u\in Y^*$ we have that
\[
L_g({\cal A}_Y)(u)=L_g({\cal A})(\natp{X}(u))\ \ \text{and}\ \ L({\cal A}_Y)(u)=L({\cal A})(\natp{X}(u)).
\]
\end{lemma}

\begin{proof}
An arbitrary word $u\in Y^*$ can be represented in the form $u=u_1v_1u_2v_2\cdots u_nv_nu_{n+1}$, where $n\in \Bbb N$, $u_1,u_2,\dots,u_{n+1}\in (Y\setminus X)^*$, and $v_1,v_2,\ldots ,v_n\in X^*$, and clearly, $\natp{X}(u)=v$, where $v= v_1v_2\cdots v_n$.~Since~$\delta_p^{A_Y}$ is the equality relation on $A$ and $\delta_q^{A_Y}=\delta_q^A$,
for all $p\in (Y\setminus X)^*$ and $q\in X^*$, then we have that
\[
\begin{aligned}
L_g({\cal A}_Y)(u)&=\bigvee_{a\in A}(\sigma^A \circ
\delta^{A_Y}_u)(a)=\bigvee_{a\in A}(\sigma^A \circ \delta^{A_Y}_{u_1}\circ
\delta^{A_Y}_{v_1}\circ\delta^{A_Y}_{u_2}\circ\delta^{A_Y}_{v_2}\circ\cdots\circ
\delta^{A_Y}_{u_n}\circ\delta^{A_Y}_{v_n}\circ\delta^{A_Y}_{u_{n+1}})(a)\\
&=\bigvee_{a\in A}(\sigma^A\circ
\delta^A_{v_1}\circ\delta^A_{v_2}\circ\cdots\circ\delta^A_{v_n})(a)=\bigvee_{a\in
A}(\sigma^A\circ \delta^A_v)(a)=L_g({\cal A})(v) = L_g({\cal A})(\natp{X}(u)),
\end{aligned}
\]
and similarly, $L({\cal A}_Y)(u)=L({\cal A})(\natp{X}(u))$.
\end{proof}

Now we prove the following:

\begin{theorem}\label{th:par.ext}
 Let ${\cal A}=(A,X,\delta^{A} ,\sigma^{A} ,\tau^{A})$ and ${\cal
B}=(B,Y,\delta^{B} ,\sigma^{B} ,\tau^{B})$ be fuzzy recognizers, let $Z=X\cup Y$,
and let ${\cal A}_Z=(A,Z,\delta^{A_Z},\sigma^{A} ,\tau^{A})$ and ${\cal B}_Z=(B,Z,\delta^{B_Z} ,\sigma^{B},\tau^{B})$ be respectively their $Z$-input extensions.

Then fuzzy recognizers ${\cal A}\Vert {\cal B}$ and ${\cal A}_Z\Vert {\cal B}_Z$ are isomorphic, and
for each $u\in Z^*$ we have that
\begin{align}
&L_g({\cal A}\Vert {\cal B})(u)=L_g({\cal
A}_Z)(u)\otimes L_g({\cal B}_Z)(u)=L_g({\cal A})(\natp{X}(u))\otimes L_g({\cal B})(\natp{Y}(u)),\label{eq:Lg-par}\\
&L({\cal A}\Vert {\cal B})(u)=L({\cal A}_Z)(u)\otimes L({\cal B}_Z)(u)= L({\cal A})(\natp{X}(u))\otimes L({\cal B})(\natp{Y}(u)).\label{eq:L-par}
\end{align}
\end{theorem}

\begin{proof}According to (\ref{eq:inp.ext}) and (\ref{eq:parallel}), for every $x\in
Z=X\cup Y$ we have that
\[
\begin{aligned}
\delta^{A_Z\Vert B_Z}((a,b),x,(a',b'))&=\delta^{A_Z}(a,x,a')\otimes
\delta^{B_Z}(b,x,b') \\
&=\begin{cases}
\delta^{A}(a,x,a')\otimes\delta^{B}(b,x,b'),&\ \ \text{if}\ x\in X\cap Y\\
\delta^{A}(a,x,a')\otimes 1,&\ \ \text{if $x\in X\setminus Y$ and $b=b'$}\\
1\otimes\delta^{B}(b,x,b'),&\ \ \text{if $x\in Y\setminus X$ and $a=a'$}
\end{cases}\\
&=\delta^{A\Vert B}((a,b),x,(a',b')),
\end{aligned}
\]for every $a,a'\in A$, $b,b'\in B$. Since fuzzy recognizers ${\cal A}$ and ${\cal A}_Z$, as well as $\cal B$ and ${\cal B}_Z$, have the same fuzzy sets of initial
and terminal states, we conclude that ${\cal A}\Vert {\cal B}$
and ${\cal A}_Z\Vert {\cal B}_Z$ are isomorphic.~Moreover, according to Lemma~\ref{lema3},~for~each $u\in Z^*=(X\cup Y)^*$ we have that
\[
\begin{aligned}
L_g({\cal A}\Vert {\cal B})(u)&=L_g({\cal A}_Z\Vert {\cal B}_Z)(u)= \bigvee_{(a,b),(a',b')\in A\times B}
\sigma^{A\Vert B}(a,b)\otimes \delta^{A_Z\Vert B_Z}((a,b),u,(a',b'))\\
&=\biggl(\bigvee_{a,a'\in A} \sigma^{A}(a)\otimes\delta^{A_Z}(a,u,a')\biggr)\otimes
\biggl(\bigvee_{b,b'\in B} \sigma^{B}(b)\otimes \delta^{B_Z}(b,u,b')\biggr)\\
&=L_g({\cal A}_Z)(u)\otimes L_g({\cal B}_Z)(u)=L_g({\cal A})(\natp{X}(u))\otimes L_g({\cal B})(\natp{Y}(u)).
\end{aligned}
\]
 The rest of the assertion can be
proved in a similar way.
\end{proof}

In particular, if $X=Y$, i.e., if ${\cal A}\Vert {\cal B}={\cal A}\times {\cal B}$,
then by (\ref{eq:Lg-par}) and (\ref{eq:L-par}) it follows that
\begin{align}
&L_g({\cal A}\times {\cal B})(u)=L_g({\cal
A})(u)\otimes L_g({\cal B})(u),\label{eq:Lg-prod}\\
&L({\cal A}\times {\cal B})(u)=L({\cal A})(u)\otimes L({\cal B})(u),\label{eq:L-prod}
\end{align}
for every $u\in X^*$.

One of the key reasons for using automata to model discrete event systems is their amenability to analysis for answering various questions about the structure and behavior of the system, such as safety properties, blocking properties, diagnosability, etc.~In the context of fuzzy automata we will consider blocking properties, which are originally
concerned with the presence of deadlock and/or livelock in the automaton,~i.e.,~with~the problem
of checking whether a terminal state can be reached from every reachable state.~

A {\it prefix-closure} of a fuzzy language $f\in L^{X^*}$, denoted by $\overline{f}$,
  is a fuzzy language in $L^{X^*}$ defined by
\begin{equation}\label{eq:pref.close}
\overline{f}(u)=\bigvee_{v\in X^*}f(uv),
\end{equation}
for any $u\in X^*$.~It is easy to verify that the mapping $f\mapsto \overline{f}$
is a closure operator on $L^{X^*}$, i.e., for arbitrary $f,f_1,f_2\in L^{X^*}$ we have that
\begin{equation}\label{eq:cl.oper}
f\le \overline{f},\ \ \ \overline{\overline{f}}=\overline{f}\ \ \text{and}\ \ f_1\le f_2\ \ \text{implies}\ \ \overline{f_1}\le \overline{f_2}.
\end{equation}
A fuzzy language $f\in L^{X^*}$ is called {\it prefix-closed\/} if $f=\overline{f}$.

We have that the following is true:

\begin{lemma}\label{lema2} Let ${\cal A}=(A,X,\delta^{A} ,\sigma^{A} ,\tau^{A}
)$ be a fuzzy recognizer.~Then
\begin{equation}\label{eq:LA.clos}
L({\cal A})\le \overline{L({\cal A})}\le L_{g}({\cal A})= \overline{L_g({\cal A})}.
\end{equation}
\end{lemma}

\begin{proof}
According to $L({\cal A})\le L_g({\cal A})$ and (\ref{eq:cl.oper}), it is enough
to prove $\overline{L_g({\cal A})}\le L_g({\cal A})$. Indeed, for arbitrary $a,b,c\in
A$ and $u,v\in X^*$ we have that
\[
\sigma^A(a)\otimes \delta_u^A(a,c)\otimes \delta_v^A(c,b) \le \sigma^A(a)\otimes \delta_u^A(a,c)\le L_g({\cal A})(u),
\]
what implies that
\begin{align*}
\overline{L_g({\cal A})}(u)&=\bigvee_{v\in X^*}L_g({\cal A})(uv)=\bigvee_{v\in X^*}\bigvee_{a,b\in
A}\sigma^A(a)\otimes \delta_{uv}^A(a,b)\\
&=\bigvee_{v\in X^*}\bigvee_{a,b\in A}\bigvee_{c\in A}\sigma^A(a)\otimes \delta_u^A(a,c)\otimes \delta_v^A(c,b)= \bigvee_{a,c\in A}\Bigl(\sigma^A(a)\otimes \delta_u^A(a,c)\Bigr)\otimes \Bigl(\bigvee_{v\in X^*}\bigvee_{b\in A}\delta_v^A(c,b)\Bigr)\\
&\le \bigvee_{a,c\in A}\sigma^A(a)\otimes \delta_u^A(a,c) = L_g({\cal A})(u),
\end{align*}
for every $u\in X^*$. Therefore, $\overline{L_g({\cal A})}\le L_g({\cal A})$.
\end{proof}

It is worth noting that the fuzzy language $\overline{L({\cal A})}$ can be represented
by
\[
\overline{L({\cal A})}(u)=\bigvee_{v\in X^*}L({\cal A})(uv) =\bigvee_{v\in X^*}\sigma^A\circ
\delta^A_{uv}\circ \tau^A = \bigvee_{v\in X^*}\sigma^A\circ
\delta^A_u\circ \delta^A_v\circ \tau^A = \bigvee_{v\in X^*}\sigma^A_u\circ \tau^A_v, \]
for every $u\in X^*$.

A fuzzy recognizer $\cal A$ is said to be {\it blocking\/} if $\overline{L({\cal A})}< L_g({\cal
A})$, where the inequality is proper, and~other\-wise, if $\overline{L({\cal A})}=L_g({\cal A})$,
then~$\cal A$ is referred to as {\it nonblocking\/}.~These concepts generalize related concepts of
the crisp discrete event systems theory, where a crisp automaton is considered to be blocking if
it can~reach a state from which no terminal state can be reached anymore.~This includes both the
possibility of a deadlock, where an automaton is stuck and unable to continue at all, and a
livelock, where an automaton continues~to run forever without achieving any further progress.

When we work with parallel compositions, the term conflicting is used instead of blocking.~Namely, fuzzy recognizers $\cal A$ and $\cal B$ are said to be {\it
nonconflicting\/} if their parallel composition ${\cal A}\Vert {\cal B}$ is nonblocking,
and otherwise they are said to be {\it conflicting\/}.~The parallel composition of a~set~of~auto\-mata may be~blocking~even if each of the individual components is nonblocking (cf.~\cite{CL.08}), and hence, it is necessary to examine
the transition structure of the parallel composition to answer blocking properties.~But,
the size of the~state~set~of~the parallel composition may
in the worst case grow exponentially in the number of
automata that are~co\-mposed. This process is known
as the ''curse of dimensionality'' in the study of complex
systems composed of many interacting components.

The mentioned problems in analysis of large discrete event models may be mitigated
if we adopt modular reasoning, which can make it possible to
replace components in the parallel~compo\-sition by smaller equivalent automata, and then to analyse a simpler system.~Such an approach has been used in \cite{MSR.04}
in study of conflicting properties of crisp discrete event systems.~Here we will show that every fuzzy recognizer $\cal A$ is conflict-equivalent with the afterset
fuzzy recognizer ${\cal A}/R$ w.r.t.~any weakly left invariant fuzzy quasi-order
$R$ on $\cal A$.~This means that in the parallel~composi\-tion of fuzzy recog\-nizers every component can be replaced~by such afterset fuzzy recognizer, what results in
a smaller fuzzy recognizer to be analysed, and do not affect
conflicting properties of the components.

Two fuzzy recognizers $\cal A$ and $\cal B$ are said to be {\it
conflict-equivalent\/} if for every fuzzy
recognizer $\cal C$ we have that ${\cal A}\Vert {\cal C}$ is nonblocking if and
only if ${\cal B}\Vert {\cal C}$ is nonblocking, i.e., if $\cal A$ and $\cal B$ are
nonconflicting (conflicting) with the same fuzzy recognizers (cf. \cite{MSR.04}).~

Now we are ready to state and prove the main results of this section.

\begin{theorem}\label{th:conf.eq}
Let ${\cal A}=(A,X,\delta^{A} ,\sigma^{A} ,\tau^{A})$ be a fuzzy
recognizer and let $R$  be a weakly left invariant fuzzy quasi-order on $\cal A$.~Then fuzzy recognizers $\cal A$
and ${\cal A}/R$ are language-equivalent, and consequently, they are conflict-equivalent.
\end{theorem}

\begin{proof}
As we already know, $L({\cal A})=L({\cal A}/R)$.~Moreover, according to the dual
statement of (\ref{eq:gen.syst.2}), for an arbitrary
$u=x_1\cdots x_n\in X^*$, where~$n\in \Bbb N$ and $x_1,\ldots ,x_n\in X$,
we have that
\begin{align*}
L_{g}({\cal A}/R)(u)& =\bigvee_{b\in A}(\sigma^{A/R}\circ \delta^{A/R}_u)(R_b) =
\bigvee_{b\in A} (\sigma^{A/R}\circ \delta^{A/R}_{x_1}\circ
\cdots \circ \delta^{A/R}_{x_n} )(R_b) \\
&=
\bigvee_{a_1,a_{2},\ldots,a_n,b\in A}\sigma^{A/R}(R_{a_1})\otimes \delta^{A/R}_{x_1}(R_{a_1},R_{a_2})
\otimes \cdots \otimes \delta^{A/R}_{x_n}(R_{a_n},R_b)\\
&=
\bigvee_{a_1,a_{2,}\ldots,a_n,b\in A}(\sigma^{A}\circ R)(a_1)\otimes (R\circ \delta^A_{x_1}\circ
R)(a_1,a_2)
\otimes  \cdots \otimes (R\circ \delta^{A}_{x_n}\circ
R)(a_n,b)\\
&=\bigvee_{b\in A}(\sigma^{A}\circ R\circ \delta^A_{x_1}\circ R\circ \cdots \circ
R\circ \delta^{A}_{x_n}\circ R)(b)\\
&=\bigvee_{b\in A}(\sigma^{A}\circ\delta^{A}_{x_1}\circ \cdots \circ \delta^A_{x_n})(b)
=\bigvee_{b\in A}(\sigma^{A}\circ\delta^{A}_u)(b) =L_{g}({\cal A})(u),
\end{align*}
~and therefore, $L_{g}({\cal A}/R)=L_{g}({\cal A})$.~Hence, $\cal A$ and ${\cal A}/R$
are language-equivalent.

Next, let ${\cal B}=(B,Y,\delta^B ,\sigma^B ,\tau^B)$ be an arbitrary
fuzzy recognizer, and let $Z=X\cup Y$.~By the language-equivalence of $\cal A$
and ${\cal A}/R$ and Theorem \ref{th:par.ext},
for every $u\in Z^*=(X\cup Y)^*$ we~have that
\[
\begin{aligned}
L_g(({\cal A}/R)\Vert {\cal B})(u)&=L_g(({\cal A}/R)_Z)(u)\otimes L_g({\cal B}_Z)(u)=
L_g(({\cal A}/R))(\natp{X}(u))\otimes L_g({\cal B})(\natp{Y}(u))\\
&=L_g({\cal A})(\natp{X}(u))\otimes L_g({\cal B})(\natp{Y}(u)= L_g({\cal A}_Z)(u)\otimes L_g({\cal B}_Z)(u) = L_g({\cal A}\Vert {\cal B})(u),
\end{aligned}
\]
and hence, $L_g(({\cal A}/R)\Vert {\cal B})=L_g({\cal A}\Vert {\cal B})$.~Similarly we prove
that $L(({\cal A}/R)\Vert {\cal B})=L({\cal A}\Vert {\cal B})$, and by this it follows  that $\overline{L(({\cal A}/R)\Vert {\cal B})}=\overline{L({\cal A}\Vert {\cal
B})}$.

Hence, we have that $\overline{L({\cal A}\Vert {\cal B})}=L_g({\cal A}\Vert {\cal
B})$ if and only if $\overline{L(({\cal A}/R)\Vert {\cal B})}=L_g(({\cal A}/R)\Vert {\cal B})$,
what~means~that $\cal A$~and ${\cal A}/R$ are conflict-equivalent.
\end{proof}

The following example shows that the previous theorem do not hold for weakly right
invariant fuzzy~quasi-orders, i.e., a fuzzy recognizer and its afterset fuzzy recognizer
w.r.t.~a weakly right invariant fuzzy quasi-order are not necessary language-equivalent
nor conflict-equivalent.

\begin{example}\label{ex:ex1}\rm
Let $\cal L$ be the Boolean structure and let ${\cal A}=(A,X,\delta^A,\sigma^A,\tau^A)$  be a fuzzy recognizer~over~$\cal L$, where $A=\{1,2,3,4\}$, $X=\{x\}$, and $\delta_x^A$,
 $\sigma^A$ and  $\tau^A$ are given by
\[
\delta_x^A=\begin{bmatrix}
1 & 0 & 0 & 0 \\
0 & 0 & 0 & 1 \\
0 & 0 & 0 & 0 \\
0 & 0 & 0 & 0
\end{bmatrix},\ \ \ \
\sigma^A= \begin{bmatrix}
0 & 1 & 0 & 1
\end{bmatrix},\ \ \ \
\tau^A=\begin{bmatrix}
0 \\
1 \\
0 \\
1
\end{bmatrix}.
\]
For every $u\in X^*$ we have that
\[
\overline{L({\cal A})}(u)= L_{g}({\cal A})(u) = \begin{cases} 1 &\ \text{if}\ u=e \ \text{or}\ u=x \\ 0 &\ \text{if}\ u=x^n,\ \text{for some}\ n\ge 2
\end{cases} ,
\]
and hence, the fuzzy recognizer $\cal A$ is nonblocking.

A fuzzy relation $R$ on $A$ given by
\[
R=\begin{bmatrix}
1 & 0 & 1 & 0 \\
1 & 1 & 1 & 1 \\
1 & 0 & 1 & 0 \\
1 & 0 & 1 & 1
\end{bmatrix},
\]
is a weakly right invariant fuzzy quasi-order on $\cal A$ (it is just the greatest
one), and the related afterset fuzzy recognizer  is ${\cal A}/R=(A/R,X,\delta^{A/R},\sigma^{A/R},
\tau^{A/R})$, where $\delta_x^{A/R}$, $\sigma^{A/R}$ and $\tau^{A/R}$ are given by
\[
\delta_x^{A/R}=\begin{bmatrix}
0 & 1 & 1 \\
0 & 1 & 0 \\
0 & 1 & 0
\end{bmatrix},\ \ \ \
\sigma^{A/R}= \begin{bmatrix}
1 & 1 & 1
\end{bmatrix},\ \ \ \
\tau^{A/R}=\begin{bmatrix}
1 \\
0 \\
1
\end{bmatrix}.
\]
For every $u\in X^*$ we have that
\[
\overline{L({\cal A}/R)}(u)=  \begin{cases} 1 &\ \text{if}\ u=e \ \text{or}\ u=x \\ 0 &\ \text{if}\ u=x^n,\ \text{for some}\ n\ge 2
\end{cases} ,
\]
and $L_{g}({\cal A}/R)(u)=1$, for each $u\in X^*$.~Hence, $\overline{L({\cal A}/R)}<L_{g}({\cal A}/R)$, and we have that the fuzzy~recognizer ${\cal A}/R$
is blocking.~We also have that
$L_{g}({\cal A})\ne L_{g}({\cal A}/R)$, what means that $\cal A$ and ${\cal A}/R$
are not language-equivalent.

Next, let ${\cal B}=(B,X,\delta^B,\sigma^B,\tau^B)$, where $B=\{b\}$, $\delta^B(b,x,b)=1$,
for each $x\in X$, and $\sigma^B(b)=\tau^B(b)=1$. Then ${\cal A}\Vert {\cal B}={\cal
A}\times {\cal B}$, and by (\ref{eq:Lg-prod}) and (\ref{eq:L-prod}) it follows that
\begin{align*}
&L_g({\cal A}\Vert {\cal B})=L_g({\cal
A})\otimes L_g({\cal B})= L_g({\cal A}), \ \ L({\cal A}\Vert {\cal B})=L({\cal
A})\otimes L({\cal B})= L({\cal A}),\\
&L_g(({\cal A}/R)\Vert {\cal B})=L_g({\cal
A}/R)\otimes L_g({\cal B})= L_g({\cal A}/R), \ \ L(({\cal A}/R)\Vert {\cal B})=L({\cal
A}/R)\otimes L({\cal B})= L({\cal A}/R).
\end{align*}
Therefore,
\[
\overline{L({\cal A}\Vert {\cal B})} = \overline{L({\cal A})}=
L_g({\cal A}) = L_g({\cal A}\Vert {\cal B}), \ \ \ \overline{L(({\cal A}/R)\Vert {\cal B})} = \overline{L({\cal A}/R)}<
L_g({\cal A}/R) = L_g(({\cal A}/R)\Vert {\cal B}),
\]
what means that ${\cal A}\Vert {\cal B}$ is nonblocking and $({\cal A}/R)\Vert {\cal B}$ is blocking, and hence, ${\cal A}$ and ${\cal A}/R$ are not conflict-equivalent.
\end{example}

\section{Concluding remarks}

In our recent paper we have established close relationships
between the state reduction of a fuzzy recognizer and~the
resolution of a particular system of fuzzy relation equations. We
have studied reductions by means of those solutions which are
fuzzy equivalences.~In this paper we demonstrated that in some
cases better reductions can be~obtained using the solutions of
this system that are fuzzy quasi-orders.~Although~by Theorem
\ref{th:lang.ER} we have proved that in the general case fuzzy
quasi-orders and fuzzy~equivalences are equally good in the state
reduction, we have shown that in some cases fuzzy quasi-orders
give better reductions.~The meaning of state reductions by means of fuzzy
quasi-orders and fuzzy equivalences is in their possible effectiveness,
as opposed to the minimization which is not effective.~By Theorem~3.5 we have shown~that minimization of some fuzzy recognizers
can not be realized as its state reduction by means of fuzzy quasi-orders
or fuzzy equivalences.

We gave a procedure for computing the
greatest right invariant fuzzy quasi-order on a fuzzy automaton or
fuzzy recognizer, which works if the underlying structure of truth
values is a locally finite, but not~only in this case.~We also gave procedures for computing
the greatest right invariant crisp quasi-order and the greatest
strongly right invariant fuzzy quasi-order.~They work for fuzzy
automata over any complete residuated lattice.~However, although
these procedures are applicable to a larger class of fuzzy
automata, we have proved that right invariant fuzzy quasi-orders
give better reductions than right invariant~crisp quasi-orders and
strongly right invariant fuzzy quasi-orders.~We also have studied
a~more general type~of~fuzzy quasi-orders, weakly right and left
invariant ones.~These fuzzy quasi-orders give better reductions
than right and left invariant ones, but are harder to compute.~In
fact, weakly right and~left~invariant fuzzy quasi-orders on a
fuzzy recognizer are defined by means of two systems of fuzzy
relation equations whose resolution include determinization of
this fuzzy recognizer and its reverse fuzzy
recognizer.

Finally,~we~have shown that better results in the
state reduction can be achieved if we alternate reductions by
means of right and left invariant fuzzy quasi-orders, or weakly
right and left invariant fuzzy quasi-orders.~Furthermore, we
show~that~alternate reductions by means of fuzzy quasi-orders give
better results than those by means of fuzzy~equivalences.~It is
worth noting that the presented state reduction methods are based
on the construction of the afterset fuzzy recognizer w.r.t.~a
fuzzy quasi-order, and we have proved that such approach gives
better results in alternate reductions than approach by
Champarnaud and Coulon, Ilie, Navarro and Yu, and Ilie, Solis-Oba
and Yu, whose state reduction methods are based on the
construction~of the factor recognizer w.r.t.~the natural
equivalence of a quasi-order.

At the end of the paper we have demonstrated some applications
of weakly left invariant fuzzy quasi-orders in conflict analysis of fuzzy discrete
event systems.~Another interesting problem is application of state reductions by means of fuzzy quasi-orders in fault diagnosis~of discrete event systems.~Since this problem is very complex and deserves special attention, it will be discussed in a separate paper.

Several questions remained unsolved, too.~They include determining more precise conditions under which our iterative procedures for
computing the greatest right and left invariant fuzzy quasi-orders terminate in a finite number of steps, finding alternative algorithms for computing the greatest right and left invariant fuzzy quasi-orders for use in cases where the mentioned iterative procedures do not terminate in a finite number of steps, as well as finding even faster algorithms for
computing such fuzzy quasi-orders, and general procedures to decide whether we have reached the smallest number of states in
alternate reductions, and so forth.~All these issues will be topics of our future research.

\end{document}